\documentclass[11pt,a4paper,onehalfspacing]{article}%article}
\usepackage{amssymb}
\usepackage{rotating}
\usepackage{amsfonts}
\usepackage{amsmath}

\usepackage{amsthm}
\usepackage{caption}
\usepackage{float}
\usepackage{graphicx}
\usepackage[colorlinks=true, urlcolor=bleu, pdfborder={0 0 0}]{hyperref}
\usepackage{enumerate}
\usepackage{multicol}
\usepackage{array}
\usepackage{bbm}
\usepackage{mathtools}
\usepackage{leftidx}
\usepackage[para]{footmisc}
\usepackage[nolist]{acronym}

\newtheorem{theorem}{Theorem}
\newtheorem{proposition}[theorem]{Proposition}
\newtheorem{lemma}[theorem]{Lemma}
\newtheorem*{lemma*}{Lemma}
\newtheorem*{assumption*}{Assumption}
\newcounter{assumptionc}

\newtheorem{corollary}[theorem]{Corollary}

\newtheorem{remark}[theorem]{Remark}

\numberwithin{equation}{section}
\numberwithin{theorem}{section}

\newtheorem{subassumption}{Assumption}[assumptionc]

\newenvironment{assumption+}
 {\ifnum\value{subassumption}=0 \stepcounter{assumptionc}\fi\subassumption}
 {\endsubassumption}

\def\text#1{\hbox{#1}}

\def\build #1_#2{\mathrel{\mathop{\kern 0pt #1}\limits_\zs{#2}}}
\newcommand{\zs}[1]{{\mathchoice{#1}{#1}{\lower.25ex\hbox{$\scriptstyle#1$}}
{\lower0.25ex\hbox{$\scriptscriptstyle#1$}}}}

\numberwithin{equation}{section}

\usepackage[top=1in, bottom=1in, left=1.in, right=0.9in]{geometry}

\newtheorem*{example*}{Example}

\newcommand{\hochkomma}{$^{,}$}

\newcommand{\PP}{\mathbb{P}}
\newcommand{\EX}{\mathbb{E}}
\newcommand{\Real}{\mathbb{R}}
\newcommand{\N}{\mathbb{N}}

\newcommand{\Var}{\text{Var}}

\newcommand{\1}{\mathbbm{1}}

\newcommand{\diff}{\mathrm{d}}
\newcommand\norm[1]{\left\lVert#1\right\rVert}

\renewcommand{\norm}[1]{\left\lVert#1\right\rVert}

\makeatletter
\newenvironment{myproof}[1][\proofname]{%
  \par\pushQED{\qed}\normalfont%
  \topsep6\p@\@plus6\p@\relax
  \trivlist\item[\hskip\labelsep\bfseries#1\@addpunct{.}]%
  \ignorespaces
}{%
  \popQED\endtrivlist\@endpefalse
}
\newcommand{\myitem}[1]{%
	\item[#1]\protected@edef\@currentlabel{#1}%
}
\makeatother

\def\dfrac{\displaystyle\frac}

\def\FF{\mathcal{F}}

\usepackage{accents}

\allowdisplaybreaks

\begin{document}
	
\begin{acronym}
	\acro{eu}[EU]{expected utility}
\end{acronym}

\title{Mean-Variance Optimization for Participating Life Insurance Contracts\footnote{Declarations of interest: none}}
\author{Felix Fie{\ss}inger\footnote{University of Ulm, Institute of Insurance Science and Institute of Mathematical Finance, Faculty of Mathematics and Economics, Ulm, Germany. Email: felix.fiessinger@uni-ulm.de} \hochkomma \footnote{Corresponding author}  \, and Mitja Stadje\footnote{University of Ulm, Institute of Insurance Science and Institute of Mathematical Finance, Faculty of Mathematics and Economics, Ulm, Germany. Email: mitja.stadje@uni-ulm.de}}
\date{\today}
\maketitle
\begin{abstract}
	This paper studies the equity holders' mean-variance optimal portfolio choice problem for (non-)protected participating life insurance contracts. We derive explicit formulas for the optimal terminal wealth and the optimal strategy in the multi-dimensional Black-Scholes model, showing the existence of all necessary parameters. %In incomplete markets, we state Hamilton-Jacobi-Bellman equations for the value function. 
	Moreover, we provide a numerical analysis of the Black-Scholes market. The equity holders on average increase their investment into the risky asset in bad economic states and decrease their investment over time.
\end{abstract}



\section{Introduction}

This paper investigates a mean-variance optimization from the perspective of the equity holders of an insurance company for the two standard designs of participating life insurance contracts, i.e., with a protected or non-protected guarantee. In a participating life insurance contract, the policyholder gets a (possibly protected) guarantee and participates proportionally at maturity from the portfolio value exceeding a pre-defined threshold higher than the guarantee. The policyholders receive at least the guarantee value at maturity if the guarantee is protected. If the guarantee is not protected, then the policyholders get, at most, the portfolio value where, initially, the insurance company gives additional equity to the premium for the investment, i.e., the equity has only limited liability. 

The portfolio theory research goes back to the 1950s with the pioneering work of Markowitz \cite{markowitz1952portfolio,markowitz1959portfolio}, who used variance to measure the riskiness of stock returns in a one-period setting. Mean-variance was later also analyzed in dynamic settings, see for instance Hakansson \cite{hakansson1971capital}, Samuelson \cite{samuelson1975lifetime} or Merton \cite{merton1972analytic,merton1975optimum}, and has been extended in various directions. For a more detailed overview of mean-variance portfolio optimizations, see the literature review from Zhang et al. \cite{zhang2018portfolio}. In these works, the management is typically assumed to maximize the mean-variance of the entire portfolio, and no distinction between equity holders and debt holders or policyholders is made.

This paper aims to provide an explicit formula for the optimal terminal wealth (from the perspective of the equity holders of an insurance company), show that the optimal solution exists (which is non-trivial since the solution includes an additional parameter), and prove an analytical formula for the optimal strategy. Moreover, we %derive a Hamilton-Jacobi-Bellman (HJB) equation for incomplete markets and 
discuss the characteristics of the terminal wealth and the optimal strategy in a numerical analysis.

While the optimal investment problem for participating life insurance contracts was solved for \ac{eu}, an analysis for mean-variance is lacking, which is the contribution of our paper. This analysis is relevant as mean-variance is widely spread in the industry, and most finance papers and textbooks use it as the benchmark to measure risk; see, for instance, Cochrane \cite{cochrane2009asset}. The reason is that mean-variance provides a straightforward interpretation of risk vs. reward, admits tractable statistical properties, and also, due to historical contingencies, has become the leading standard taught in business schools, making it easier for a risk manager to justify its use compared to specifying a utility function. Finally, due to the nonlinearity of variance, the mathematical analysis differs from \ac{eu}. In particular, one needs to specify an equivalent problem which, contrary to \ac{eu}-maximization, adds an additional parameter $\lambda$ to the model. The existence of such a $\lambda$ is not apparent, and we are only able to derive its existence and compute it in a semi-explicit way in the case of the Black-Scholes model. 

As discussed in the beginning, the insurance company offering a participating life insurance contract pays back a surplus to the policyholder in good economic states. Insurance policies with profit participation play an essential role in the life insurance sector. According to the European Insurance Overview 2023 \cite{EIOPAreport}, issued by the European Insurance and Occupational Pensions Authority (EIOPA), policyholders spent 2022 around a quarter of their gross premiums on profit participation insurance policies in the life sector (includes life, health, and pension insurance). In Croatia, Italy, and Belgium, these policies have a market share of over 50 \% in the life sector. Several publications are concerned with the valuation and hedging of such insurance policies, e.g., Bryis and de Varenne \cite{briys1997risk}, Bacinello and Persson \cite{rita2002design}, Gatzert and Kling \cite{gatzert2007analysis}, Schmeiser and Wagner \cite{schmeiser2015proposal} or Mirza and Wagner \cite{mirza2018policy}. In recent years, the study of optimal investments for participating life insurance contracts in continuous time started. Lin et al. \cite{lin2017optimal} analyzed 2017 the optimal investment for an \ac{eu}-problem for a specific S-shaped utility function. Afterwards, Nguyen and Stadje \cite{nguyen2020nonconcave} examined a similar setting for more general S-shaped utility functions under a Value at Risk constraint and mortality risk. He et al. \cite{he2020weighted} made another generalization by considering a weighted utility function between the insurer and the policyholders. Moreover, Dong et al. \cite{dong2020optimal} added to the \ac{eu}-problem a Value at Risk and a portfolio insurance constraint. Chen et al. \cite{chen2018optimal} considered Value at Risk, expected shortfall, and Average Value at Risk constraints. To the best of our knowledge, we are the first to consider participating insurance contracts in continuous time under mean-variance optimization.

A significant factor of an optimal policy design is the payoff structure: For a participating life insurance contract, the payoff has at least one point where it is non-differentiable, e.g., at the threshold where the proportional surplus participation starts. (The payoff of the non-protected product also has a second point of non-differentiability at the guarantee value.) At these points, the payoff changes its slope, resulting in a payoff that is generally neither convex nor concave (see Figure \ref{fig: Participating life insurance} for an illustration of the insurer's payoff).
In such situations, one often uses concavification techniques; see, for instance, Larsen \cite{larsen2005optimal} or Reichlin \cite{reichlin2013utility}. Liang et al. \cite{liang2021unified} used a generalization of this method. For alternative approaches in non-concave portfolio optimization, see, for instance, Kraft and Steffensen \cite{kraft2013dynamic}, Dai et al. \cite{dai2019non}, or Qian and Yang \cite{qian2023non}.
In this work, we use the specific structure of a payoff in a complete market to carefully compare different solution candidates and show that the optimal terminal wealth actually has a closed and, up to the implicit parameters whose existence proof and computation in our setting is delicate, simple form. We also use a standard Lagrangian approach to get the optimal terminal wealth and use the price density process to give the analytic formula for the optimal strategy. This method was used to solve several other optimization problems in complete markets; see, for instance, Basak and Shapiro \cite{basak2001value}, Cuoco et al. \cite{cuoco2008optimal}, Chen et al. \cite{chen2018optimal}, Chen et al. \cite{chen2019constrained}, Nguyen and Stadje \cite{nguyen2020nonconcave}, Mi et al. \cite{mi2023optimal}, {Chen et al.} \cite{chen2024equivalence}, or {Avanzi et al.} \cite{avanzi2024optimal}. 
To use this Lagrangian approach for mean-variance, we start to give an equivalent problem in the spirit of Zhou and Li \cite{zhou2000continuous}. The ansatz described above then leads to having two multipliers instead of only one Lagrangian multiplier.  
We obtain a system of non-linear equations with several variables yielding implicit functions whose properties, through various intermediate results, then entail the existence of a solution. Although mean-variance does not respect first-order stochastic dominance, we do get similar results compared to \ac{eu}-optimization regarding the general form of the optimal solution, but with the terminal wealth having a simpler structure. In particular, in a complete market, the optimal terminal wealth is piecewise linear in the price density. %On the other hand, in an incomplete market, the optimal solution itself can only be characterized implicitly as the (viscosity) solution of a certain PDE.
{Compared to the papers discussed above, the main difference of our ansatz is to consider an optimization functional, where we include an option-like payoff of the participation instead of adding further conditions to a non-option-like payoff. We show in Proposition \ref{prop: equivalent-like problem} that our optimization corresponds to a target problem, where one aims to be as close as possible to certain exogenously given constants. This alternative formulation is close to Avanzi et al. \cite{avanzi2024optimal}, who minimize a variance-like function under solvency conditions. However, while they face a 1-target problem our optimization corresponds to a 2-target problem, where the exogenously given target depends on the region of the terminal wealth. An additional constraint ensures that mass can only be shifted from one region to another at a rate which depends on the amount of participation. Target-based optimization is common in the literature as an alternative to mean-variance, see for instance Avanzi et al. \cite{avanzi2024optimal}, Menoncin and Vigna \cite{menoncin2017mean}, or Li and Forsyth \cite{li2019data}.}
Finally, our numerical results show a somewhat different investment behavior compared to \ac{eu}-maximization, {in particular for S-shaped utility functions}. Specifically, we observe that the investment gets more conservative with shorter maturities, and the equity holders on average increase their investment into the risky asset in bad economic states. Moreover, the insurer invests more riskily when offering a non-protected participation life insurance product than when offering a protected one due to a reduced downside potential. 
Since variance is the most widely used risk measure in industry and the cornerstone of most modern portfolio theory and the finance literature, these results might be potentially relevant from a descriptive point of view.

Section \ref{chapter: model setup} describes participating life insurance contracts and briefly introduces the functional to optimize. In Section \ref{chapter: complete market}, we explicitly show the optimal terminal wealth, the optimal strategy, and the existence of the necessary parameters in the Black-Scholes model. %In Section \ref{chapter: incomplete}, we give the SDEs for possibly incomplete markets, and
Section \ref{numerics} analyzes some numerical results for the Black-Scholes model. Finally, Section \ref{conclusion} concludes the paper.

\section{Model Setup} \label{chapter: model setup}

\subsection{Participating life insurance contracts}

Participating insurance contracts or, in general, insurance contracts with some profit participation play a crucial role in the life sector. Figure \ref{fig: diagram} shows that, on average, around a quarter of the gross premiums 2022 in the life sector are spent on policies with profit participation. In some countries, like Croatia, Belgium, or Italy, even more than 50 \% of the gross premiums are invested in such policies.

\begin{figure}[!htb]
	\centering
	\includegraphics[trim= 22mm 98mm 28mm 28mm, clip,width=\textwidth]{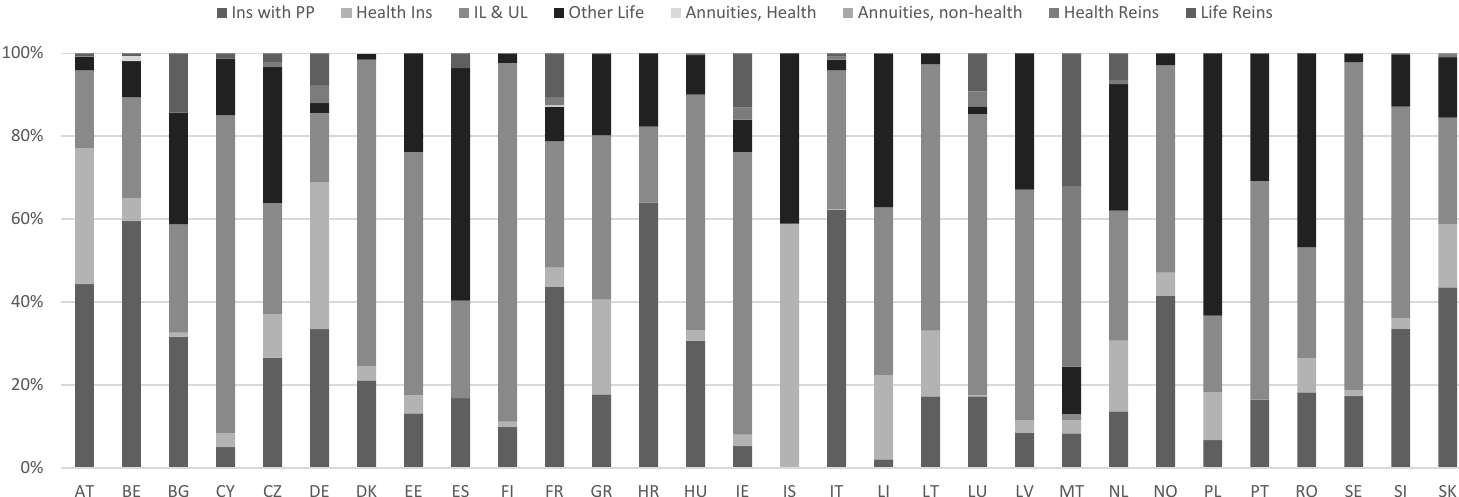}
	\caption{Market share in 2022 of the gross premium separated by the line of business in the life sector. Data source: European Insurance Overview from the EIOPA \cite{EIOPAreport}}
	\label{fig: diagram}
\end{figure}	

This paper focuses on two standard designs of participating life insurance contracts. Both products offer a guarantee value $G$ and a proportional surplus participation rate $\alpha_2$ when the portfolio value exceeds a threshold, which we denote by $k_2$. Note that $k_2$ is always higher or equal to $G$. The difference between the two designs is that the first product offers a non-protected guarantee, whereas the second provides a protected one. Protected means, in this case, that the policyholders get at least their guarantee value, independent of the economic situation at maturity. In contrast, in the non-protected case, the insurance company declares bankruptcy if the portfolio value is below $G$. In this case, the policyholders only get the portfolio value. For the portfolio, it is essential to note that the initial portfolio value $x_0$ is the sum of the premiums from the policyholders plus some initial capital from the insurer (i.e., the equity holders). There is no rule on how to set the guarantee value. One possible example of setting $G$ is to take the sum of the premiums in addition to a guaranteed interest rate below the risk-free interest rate. Moreover, product designers often set the threshold $k_2$ as the sum of the premiums divided by the share of the policyholders in the portfolio. 

Hence, we conclude the following payoffs $V$ for the policyholders (pol) and the insurer (ins) for the contracts with non-protected (non) resp. protected (pro) guarantees and terminal portfolio value $X_T \geq 0$:\footnote{For simplicity, we assume that the management or the regulator does not allow the total wealth to become negative. However, if either we extend the following definitions also for $X_T<0$, or let the insurer cover in both cases all losses stemming from a negative terminal portfolio value, all our results hold (see also Remark \ref{optimal wealth remark}\ref{remark: part extension smaller 0}).}
\begin{align*}
	V_{\scriptsize \text{pol}}^{\scriptsize\text{non}} (X_T) &= \begin{cases}
		X_T & \text{if } X_T < G, \\
		G & \text{if } G \leq X_T < k_2, \\
		G + \alpha_2(X_T-k_2) & \text{if } X_T \geq k_2,
	\end{cases} \\
	V_{\scriptsize \text{pol}}^{\scriptsize\text{pro}} (X_T) &= \begin{cases}
		G & \text{if } X_T < k_2, \\
		G + \alpha_2(X_T-k_2) & \text{if } X_T \geq k_2,
	\end{cases} \\
	V_{\scriptsize \text{ins}}^{\scriptsize\text{non}} (X_T) &= \begin{cases}
		0 & \text{if } X_T < G, \\
		X_T-G & \text{if } G \leq X_T < k_2, \\
		X_T - G -\alpha_2 (X_T-k_2) & \text{if } X_T \geq k_2,
	\end{cases} \\
	V_{\scriptsize \text{ins}}^{\scriptsize\text{pro}} (X_T) &= \begin{cases}
		X_T-G & \text{if } X_T < k_2, \\
		X_T - G -\alpha_2(X_T-k_2) & \text{if } X_T \geq k_2.
	\end{cases}
\end{align*}
Note that solely $V_{\scriptsize \text{ins}}^{\scriptsize\text{pro}}$ can attain negative values. For a more detailed explanation of participating insurance, we refer to Nguyen and Stadje \cite{nguyen2020nonconcave}. 

\subsection{Optimization Functional}

Before defining the optimization functional, let us introduce the basic financial market, which in Section \ref{chapter: complete market} and \ref{numerics} will become the Black-Scholes model. \\
Let $(\Omega,\FF,(\FF)_{t \in [0,T]},\PP)$ be a filtered probability space with time horizon $T>0$. The filtration is generated by the $d$-dimensional Brownian Motion $W$ satisfying the usual conditions. We consider an arbitrage-free market with a risk-free asset $B$ and a deterministic interest rate $r_t \geq 0$ following the price process $\diff B_t = B_t r_t \diff t$, and $d$ risky assets $S^i$, $i \in \{1,\ldots,d\}$ with adapted price processes. Let the dynamic strategy $u$ represent the fraction of wealth invested in the corresponding risky asset, while the remaining money is invested in the risk-free asset with the corresponding wealth process $X_t$ and initial value $X_0 = x_0$. We denote by $\mathcal{U}$ the set of all admissible strategies given by all $u$, which are progressively measurable and induce an integrable $X_T$.

Next, we introduce a general functional to optimize, the previously defined participating life insurance products being special cases. We look for the optimal strategy $\hat{u} \in \mathcal{U}$ such that 
\begin{align} \label{J definition}
	J (0,T,\hat{u},x_0) = \sup_{u \in \mathcal{U}} J(0,T,u,x_0),
\end{align} 
where the value functional $J$ is defined as
\begin{align*}
	J(0,T,u,x_0) := \EX [ F(0,T,u,x_0)] - \gamma \Var(F(0,T,u,x_0))
\end{align*}
with the risk aversion parameter $\gamma>0$. Hence, we are looking for the mean-variance optimal strategy for the (continuous) function $F$, which we define for $s<t$ as
\begin{align*}
	F(s,t,u,x) :=&\ \alpha\left( (X_t - k_1)_+ - k_0 \right)-\alpha_2 (X_t-k_2)_+ \\
	=&\ \begin{cases}
		-\alpha k_0 & X_t<k_1, \\
		\alpha(X_t-k_1-k_0) & k_1 \leq X_t < k_2, \\
		\tilde{\alpha}(X_t-k_2) + \alpha(k_2-k_1-k_0) & X_t \geq k_2,
	\end{cases}
\end{align*}
where $X_s=x$, $0 \leq k_0, k_1 \leq k_2 < \infty$ with $k_0+k_1 \leq k_2$, $K_2>0$, $0\leq \alpha_2 < \alpha < \infty$ with $\tilde{\alpha}:=\alpha-\alpha_2$. Note that $0 < \tilde{\alpha} \leq \alpha$ holds. We also write $F(X_T)$ instead of $F(0,T,u,x)$ when the trading strategy is clear (for instance, for the optimal terminal wealth $\hat{X}_T$), suppressing the initial value.

Now, one can observe that if $\alpha=1$, $k_0=0$, and $k_1=G$, the function $F$ reduces to the payoff of the insurer for the non-protected participating life insurance contract $V_{\scriptsize \text{ins}}^{\scriptsize\text{non}}$. If $\alpha=1$, $k_0=G$, and $k_1=0$, $F$ reduces to the payoff of the insurer for the protected product $V_{\scriptsize \text{ins}}^{\scriptsize\text{pro}}$. In the following Figure \ref{fig: Participating life insurance}, we show the payoff of the insurer for these two variants of participating life insurance contracts. 

\begin{figure}[!htb]
	\centering
	\begin{minipage}{0.48\textwidth}
		\includegraphics[trim= 10mm 15mm 10mm 8mm, clip,width=\textwidth]{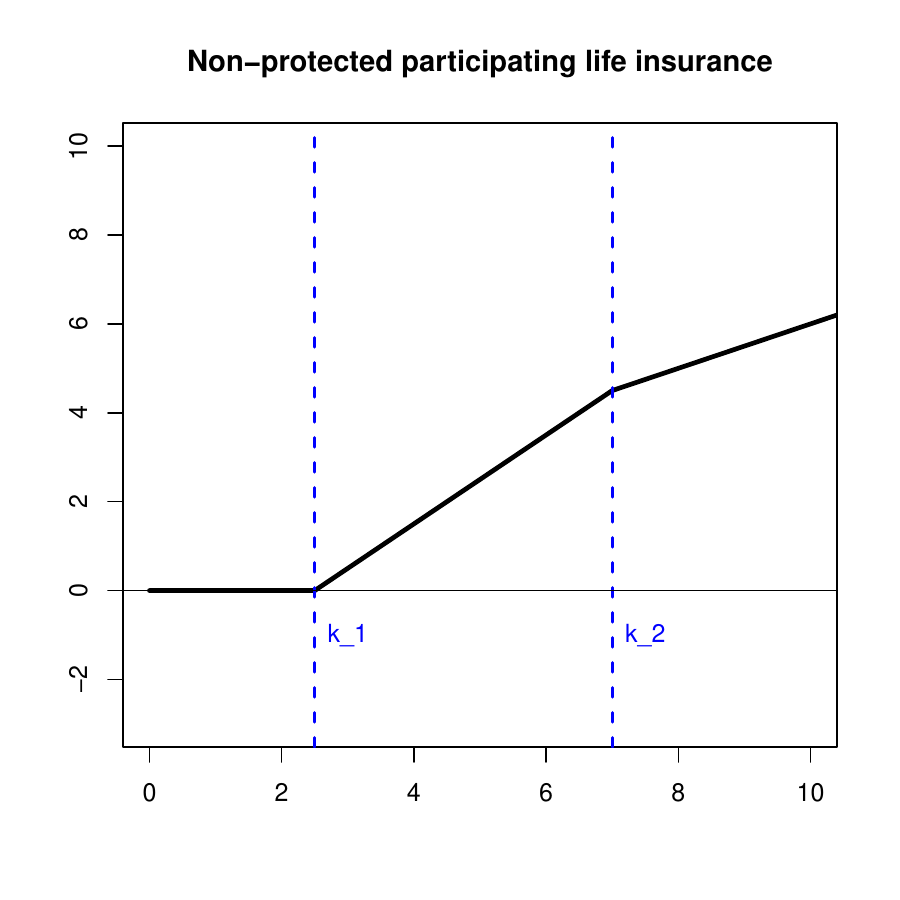}
	\end{minipage}
	\quad
	\begin{minipage}{0.48\textwidth}
		\includegraphics[trim= 10mm 15mm 10mm 8mm, clip,width=\textwidth]{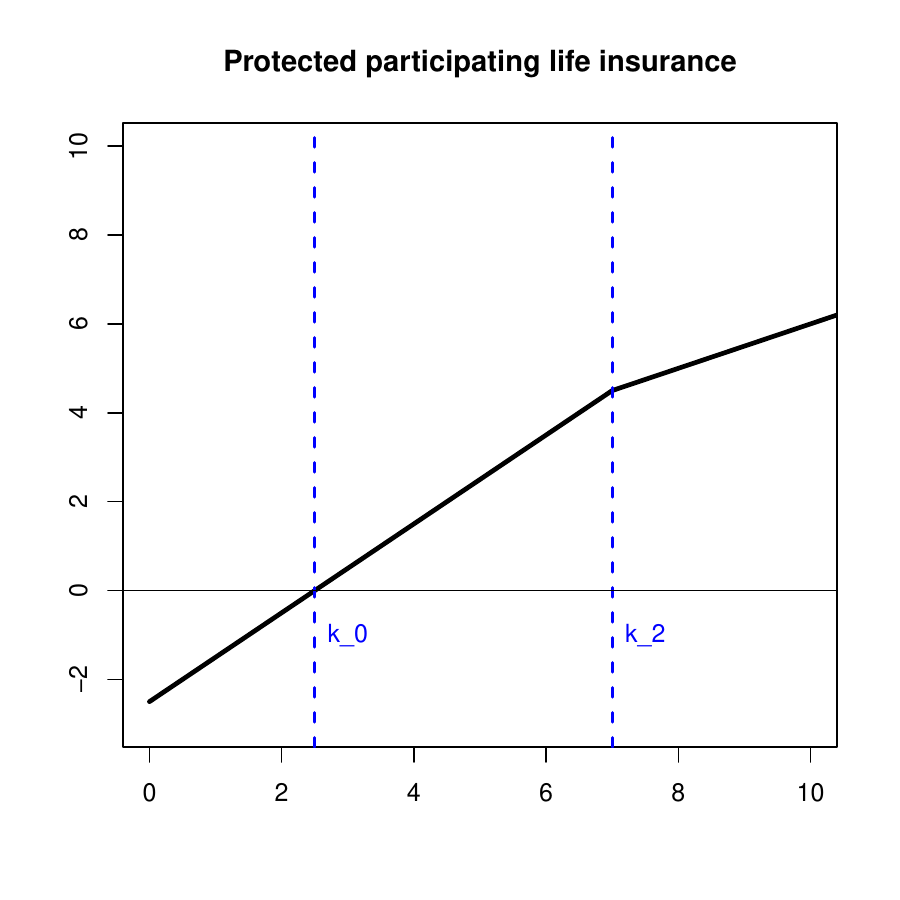}
	\end{minipage}
	\caption{Payoffs for the insurer with $k_1 = 2.5$ (left) resp. $k_0 = 2.5$ (right), $k_2 = 7$, $\alpha=1$, $\alpha_2=0.5$.}
	\label{fig: Participating life insurance}
\end{figure}	

\begin{remark}
	{Since the problem in this paper is a variant of a mean-variance optimization, it inherits also the time-inconsistency of the classical mean-variance approach, see, for instance, Basak and Chabakauri \cite{basak2010dynamic} or Bj{\"o}rk and Murgoci \cite{bjork2010general}. There are several ways to approach this time-inconsistency. We take in this paper the perspective of a so-called pre-committed investor. Such an investor optimizes once in the beginning and sticks to this strategy, i.e., the chosen strategy of the investor is not dynamically reoptimized. Note that our problem is not covered by the general theory from Bj{\"o}rk and Murgoci \cite{bjork2010general} which assumes a sophisticated investor, i.e., she/he knows that her/his optimal strategy changes over time and only optimizes over all strategies she/he would actually follow.}
\end{remark}

\section{Optimization in a Black-Scholes market} \label{chapter: complete market}

This section assumes that the underlying financial market follows the Black-Scholes model, i.e., the market is complete. Then, the price dynamics for the $d$ risky assets $S^i$, $i \in \{1,\ldots,d\}$ are given by
\begin{align*}
	\diff S^i_t &= S^i_t \mu^i_t \diff t + S^i_t \sigma^i_t \diff W_t,
\end{align*}
where $W$ is the Brownian Motion generating the filtration $\FF$, $\mu_i$ is the deterministic drift of the $i$'th asset and $\sigma^{ij}$ is the deterministic volatility between the $i$'th and the $j$'th asset. We assume that $\sigma_t$ is bounded, bounded away from zero and invertible. Then, there exists a unique price density process $\xi$ with the wealth process $X$ and the price density process $\xi$ admitting the following dynamics
\begin{align} \label{X definition}
	\diff X_t &= X_t \left[r_t + u_t^T (\mu_t -r) \right] \diff t + X_t u_t^T \sigma_t \diff W_t, \\
	\diff \xi_t &= -\xi_t r_t \diff t - \xi_t \kappa_t^T \diff W_t, \label{xi definition}
\end{align}
with $X_0 = x_0$ and $\xi_0 =1$, where $\kappa_t = (\sigma_t)^{-1} (\mu_t - r_t)$ is the Sharpe ratio process and $\cdot^T$ denotes the transpose of a vector. The term $\xi_T(\omega)$, $\omega \in \Omega$, can be interpreted as the Arrow-Debreu value per probability unit in state $\omega$ at time $T$. Note that $\xi_T(\omega)$ can be written as a decreasing function of the stock price, and therefore attains high values in times of a bad economy and low values in times of a good economy. We assume that the processes $r_t$ and $\mu_t$ are integrable and the processes $\sigma_t$ and $\kappa_t$ are square-integrable over $[0,T]$ to ensure that the previous SDEs and all of the following integrals are well-defined.

\subsection{Derivation of the optimal terminal wealth}

The mean-variance optimization \eqref{J definition} is challenging to solve directly due to the term $(\EX[X])^2$ in the decomposition formula $\Var(X) = \EX[X^2] - (\EX[X])^2$. Hence, we show in the following Lemma \ref{Alternative optimization} that if an optimal strategy exists, we can alternatively look for the optimal strategy considering the value functional $\tilde{J}$ which we define as:
\begin{align} \label{J tilde definition}
	\tilde{J}(0,T,u,x_0) := \EX [\lambda F(0,T,u,x_0) - \gamma F(0,T,u,x_0)^2]
\end{align}
with $\lambda = 1 + 2\gamma \EX \left[ F(0,T,\hat{u},x_0) \right]$ where $\hat{u}$ is the optimal strategy.

\begin{lemma} \label{Alternative optimization}
	If $\hat{u}$ is an optimal strategy for $J$, it is also an optimal strategy for $\tilde{J}$.
\end{lemma}

Consequently, the optimal terminal wealth also coincides for the two problems maximizing $J$ resp. $\tilde{J}$. This result is a slight generalization of Theorem 3.1 in Zhou and Li \cite{zhou2000continuous}, who showed this lemma in the case of $F$ being the identity. We will include the proof in Appendix \ref{proofs} for the reader's convenience.

Since we have a complete market, we optimize with the following three steps. First, we use a Lagrangian approach to find the optimal terminal wealth $\hat{X}$ (using the alternative problem). Second, we derive the optimal strategy, and third, we determine the $\lambda$ and the Lagrangian multiplier $y$ since the optimal terminal wealth and strategy depend on these values initially. In the following, we suppress the dependence on $\lambda$ and $y$ for the sake of simplicity in the notation unless stated otherwise in some proofs, and we use the convention that $(a, a]= \emptyset$ and $[a,b] = \emptyset$ if $b<a$.

\begin{theorem} \label{optimal wealth}
	The optimal terminal wealth $\hat{X}_T$ is given by:
	\begin{align} \label{eq: optimal terminal wealth}
		\hat{X}_T := \begin{cases}
			k_2 + \dfrac{\lambda \tilde{\alpha} - y \xi_T}{2 \gamma \tilde{\alpha}^2} - \dfrac{\alpha}{\tilde{\alpha}} (k_2-k_1-k_0) & \xi_T \in (0,\xi_1^*], \\
			k_2 & \xi_T \in (\tilde{\alpha}\hat{\xi},\xi_2^*], \\
			k_0 + k_1 + \dfrac{\lambda \alpha - y \xi_T}{2 \gamma \alpha^2} & \xi_T \in (\alpha\hat{\xi},\xi_3^*], \\
			0 & \text{else,}
		\end{cases}
	\end{align}
	where $y$ is the Lagrangian multiplier which solves $\EX [\xi_T \hat{X}_T (y)]  = \xi_0 x_0$, $\lambda = 1 + 2\gamma \EX \left[ F(0,T,\hat{u},x_0) \right]$, and
	\begin{align*}
		\hat{\xi} &:= \max \left\{ 0,\dfrac{\lambda-2 \gamma \alpha (k_2-k_1-k_0)}{y} \right\}, \\
		\bar{\xi} &:= \dfrac{\lambda \alpha}{y} + \dfrac{2 \gamma \alpha^2 k_0}{y}, \\
		\tilde{\xi}_1^* &:= \tilde{\alpha} \hat{\xi} - \dfrac{2 \gamma \tilde{\alpha}}{y} \left( \sqrt{ \max \left\{ 0, (\alpha (k_0+k_1) - \alpha_2 k_2)^2 - \alpha^2k_0^2 + \dfrac{\lambda}{\gamma} (\alpha k_1 - \alpha_2 k_2)\right\}} - \tilde{\alpha} k_2 \right), \\
		\xi_1^* &:= \max \left\{0, \min \left\{ \tilde{\alpha} \hat{\xi}, \tilde{\xi}_1^* \right\}\right\}, \\
		\tilde{\xi}_2^* &:= \dfrac{\alpha \lambda}{y} - \dfrac{\gamma \alpha^2 (k_2 - k_1)^2 - 2 \gamma \alpha^2 k_0 (k_2-k_1) + \lambda \alpha k_1}{yk_2}, \\
		\xi_2^* &:= \max \left\{ \tilde{\alpha} \hat{\xi} , \min \left\{ \alpha \hat{\xi}, \tilde{\xi}_2^* \right\}\right\}, \\
		\xi_3^* &:= \max\left\{ \alpha \hat{\xi}, \bar{\xi} - \dfrac{2 \gamma \alpha^2}{y} \left( \sqrt{k_1^2 +k_1 \left(2k_0 + \dfrac{\lambda}{\gamma \alpha}\right)} - k_1 \right)\right\}.
	\end{align*}
	In particular, such $\hat{u},t,y$ exist. Moreover, let $\xi^*$ be defined by:
	\begin{align*}
		\xi^* := \begin{cases}
			\xi_3^* & \text{,if } \xi_3^* > \alpha \hat{\xi}, \\
			\xi_2^* & \text{,if } \xi_3^* = \alpha \hat{\xi}, \xi_2^* > \tilde{\alpha} \hat{\xi}, \\
			\xi_1^* & \text{,if } \xi_3^* = \alpha \hat{\xi}, \xi_2^* = \tilde{\alpha} \hat{\xi}. 
		\end{cases}
	\end{align*}
	Then, it holds that $\xi^*>0$ and $\hat{X}_T >0$ for $\xi \in (0,\xi^*)$ and $\hat{X}_T = 0$ for $\xi > \xi^*$.
\end{theorem}

\begin{remark} \label{optimal wealth remark}
	\begin{enumerate}[(a)]
		\item The terminal wealth $\hat{X}_T$ denotes the wealth before distributing the wealth to the insurer and the policyholders. The terminal wealth of the insurer is given by $F(\hat{X}_T) = \alpha ( (\hat{X}_T - k_1)_+ - k_0 )-\alpha_2 (\hat{X}_T-k_2)_+$. In particular, in the case of a protected participating life insurance contract, i.e., $\alpha=1$, $k_1=0$, the insurer makes a loss if $\hat{X}_T < k_0$. In the case of a non-protected participating life insurance contract, the insurer cannot make a loss by construction.
		\item If $\alpha_2 = \alpha$ resp. $\tilde{\alpha} = 0$, the surplus over $k_2$ is fully distributed to the policyholders. In this case, it holds that $\tilde{\xi}_1^* = 0$ and $\tilde{\alpha} \hat{\xi} = 0$, i.e., the result is similar with the exception that the first case, $\hat{X}_T = k_2 + \frac{\lambda \tilde{\alpha} - y \xi_T}{2 \gamma \tilde{\alpha}^2} - \frac{\alpha}{\tilde{\alpha}} (k_2-k_1-k_0)$, does not need to be considered. The proofs are similar, but we exclude this case for the ease of exposition.
		\item\label{remark: part extension smaller 0} We can generalize the result for $X_T \in \Real$ when we extend the function $F$ to $\Real_{< 0}$ as $F(X_T) = \alpha (X_T-k_0)$. Then, the optimal terminal wealth has the same structure as before, but with the additional case that $\hat{X}_T = k_0 + \frac{\lambda \alpha - y \xi_T}{2 \gamma \alpha^2}$ if $\xi_T > \bar{\xi}$. Hence, the restriction to $X_T \geq 0$ is not crucial for the structure of the solution. Note that $F(\hat{X}_T (\xi_T = \bar{\xi})) = 0$, and that for $\xi_T > \bar{\xi}$ the function $F$ is linearly decreasing in $\xi_T$. Again, the proofs are similar to before, but we nevertheless restrict to $X_T \geq 0$ to avoid too technical proofs.
		\item The representation of $\hat{X}_T$ also holds in a general complete market (without necessarily using a Black-Scholes market model) if we assume the existence of the {multipliers} $y$ and $\lambda$. However, that these exist is then not clear. {Note that the price density process, $\xi$, in a different complete market would change, but still admit a geometric form. For example, if we replace the Brownian Motion $W$ by a Poisson process $N$ with parameter $\lambda_2$ (in the Black-Scholes model), the price density admits a dynamic of the form (see, e.g., Cont and Tankov \cite[p.303]{cont2004financial}):
		\begin{align*}
			\diff \xi_{t-} &= \xi_{t-} (\lambda_2-\lambda_1- r_t) \diff t - \xi_{t-} \ln \tfrac{\lambda_2}{\lambda_1} \diff N_t,
		\end{align*}
		where $\lambda_1>0$ is an appropriately chosen constant.}
		%where $\lambda_1$ (resp. $\lambda_2$) is the intensity of the risk-neutral (resp. real-world) Poisson process.}
		\item The way how to compute the parameters $\lambda$ and $y$ is given in the proof of Proposition \ref{y lambda exist}. 
	\end{enumerate}
\end{remark}

From the definition formula of $\hat{X}_T$, we can derive the following two propositions, which are proven in Appendix \ref{proofs}:

\begin{proposition} \label{prop: * Eigenschaften}
	It holds that $(0,\xi_1^*] \cup (\tilde{\alpha} \hat{\xi},\xi_2^*] \cup (\alpha \hat{\xi} ,\xi_3^*] = (0,\xi^*]$ with these constants defined in Theorem \ref{optimal wealth}, i.e., the three intervals of the terminal wealth are connected. 
\end{proposition}

\begin{proposition} \label{prop: xhat continuous}
	It holds that $\hat{X}_T$ as a function of $\xi$ is continuous and non-increasing in $(0,\xi^*) \cup (\xi^*,\infty)$ with $\xi^*$ as in Theorem \ref{optimal wealth}. If $k_1>0$, then $\hat{X}_T$ is always discontinuous at $\xi^*$.
\end{proposition}

Therefore, $\hat{X}_T$ as a function of $\xi_T$ is non-increasing, taking the value $0$ for $\xi_T > \xi^*$ and one possible discontinuity point. See the black lines in Figure \ref{fig: Optimal Value} for a visualization.

{In the following proposition, we give an alternative formulation for our optimization problem. Again, we give the proof in Appendix \ref{proofs}.}

\begin{proposition} \label{prop: equivalent-like problem}
	{Let $\tilde{L}(x) = {\scriptsize\begin{cases}
				\gamma \alpha^2 k_0^2, & \text{if }x<k_1, \\
				\gamma \alpha^2 (x-k_1-k_0)^2, & \text{if }k_1 \leq x < k_2, \\
				\gamma \tilde{\alpha}^2 (x-k_2+\tfrac{\alpha}{\tilde{\alpha}}(k_2-k_1-k_0))^2, & \text{if }x \geq k_2.
		\end{cases}}$ and consider the optimization problem $\min_{u \in \mathcal{U}} \EX [\tilde{L}(X_T)]$ under the constraint $\EX[\alpha(X_T(\lambda)-k_1) \1_{k_1 \leq X_T(\lambda) < k_2} +\linebreak \tilde{\alpha}(X_T(\lambda)-k_2+\tfrac{\alpha}{\tilde{\alpha}}(k_2-k_1)) \1_{X_T(\lambda) \geq k_2}] = \alpha k_0$ and the budget constraint $\EX [\xi_T X_T (y)]  = \xi_0 x_0$, where $\lambda$ and $y$ are the respective Lagrangian multipliers.}
	
	{Then, the optimal terminal wealth is identical to the one of our main problem, i.e., given by \eqref{eq: optimal terminal wealth}.}
\end{proposition}

{Note that the loss function $\tilde{L}$ is continuous and non-decreasing in $x$, and that the definition and therefore the interpretation of the parameter $\lambda$ is different than for mean-variance. However, Proposition \ref{prop: equivalent-like problem} still gives an equivalent formulation. As already stated in the introduction, this formulation is a 2-target problem. Depending on the region of the terminal wealth, we want to minimize the distance to different exogenously given target values ($k_1+k_0$ if the terminal wealth $X \in [k_1,k_2)$ respectively to $k_2-\tfrac{\alpha}{\tilde{\alpha}}(k_2-k_1-k_0)$ if $X \geq k_2$). Moreover, the first constraint entails that one can only reduce mass in one region when one simultaneously increases the mass in another region where the reduction and the enlargement are, in general, not of equal size. Especially, the weight of the region $[k_2,\infty)$ is heavily influenced by the participation rate $\alpha_2$ through the pre-factor $\tilde{\alpha} (=\alpha-\alpha_2)$.}

\begin{myproof}[Proof of Theorem \ref{optimal wealth}]
We use the Lagrangian multiplier method to prove this theorem. Therefore, choose $\lambda>0$ and $y>0$ as in Proposition \ref{y lambda exist} in the appendix and we define the Lagrangian function $L$ for $X\geq0$ as
\begin{align}
	L(X,y) :=& \ \lambda F(0,T,u,x_0) - \gamma F(0,T,u,x_0)^2 - y \xi X \label{eq: L function}\\
	=& \ \lambda \left[\alpha\left((X - k_1)_+-k_0\right) -\alpha_2 (X-k_2)_+\right] - \gamma \left[\alpha\left((X - k_1)_+-k_0\right) -\alpha_2 (X-k_2)_+ \right]^2 - y \xi X \notag\\
	=& \ \lambda \left[\alpha(X - k_1-k_0) \1_{k_1 \leq X < k_2} + \tilde{\alpha} (X-k_2) \1_{X \geq k_2} + \alpha (k_2-k_1-k_0) \1_{X \geq k_2} - \alpha k_0 \1_{X < k_1} \right] \notag\\
	&- \gamma \left[\alpha(X - k_1-k_0) \1_{k_1 \leq X < k_2} + \tilde{\alpha} (X-k_2) \1_{X \geq k_2} + \alpha (k_2-k_1-k_0) \1_{X \geq k_2} -\alpha k_0 \1_{X < k_1}\right]^2 \notag\\
	&-  y \xi X, \notag\\
	=& \ \lambda \alpha(X - k_1-k_0) \1_{k_1 \leq X < k_2} + \lambda \left[ \tilde{\alpha} (X-k_2) + \alpha (k_2-k_1-k_0) \right] \1_{X \geq k_2} - \lambda \alpha k_0 \1_{X < k_1} \notag\\
	&- \gamma \alpha^2(X - k_1-k_0)^2 \1_{k_1 \leq X < k_2} - \gamma \left[\tilde{\alpha} (X-k_2) + \alpha (k_2-k_1-k_0) \right]^2  \1_{X \geq k_2} - \gamma \alpha^2 k_0^2 \1_{X < k_1} \notag\\
	&- y \xi X \notag\\
	=& \begin{cases}
		-\lambda \alpha k_0 - \gamma \alpha^2 k_0^2 - y \xi X & X \in [0,k_1), \\
		\lambda \alpha (X-k_1-k_0) - \gamma \alpha^2 (X-k_1-k_0)^2 - y \xi X & X \in [k_1,k_2), \\
		\lambda \left[ \tilde{\alpha} (X-k_2) + \alpha (k_2-k_1-k_0) \right] - \gamma \left[\tilde{\alpha} (X-k_2) + \alpha (k_2-k_1-k_0) \right]^2 - y \xi X & X \in [k_2,\infty),
	\end{cases} \notag
\end{align}
where we suppress the $\omega$, write $X$ instead of $X_T$, $\xi$ instead of $\xi_T$, and use $y$ as the multiplier. \\
The following paragraph considers $L$ as a function of $X$. Obviously, $L$ is not smooth in $k_1$ and $k_2$, but it is continuous in $[0,\infty)$. Hence, we optimize the piecewise concave function $L$ in the regions $[0,k_1]$, $[k_1,k_2]$, and $[k_2,\infty)$ separately and compare afterwards the maximal points. Now, for the optimization, we get:
\begin{align*}
	\dfrac{\partial L}{\partial X} = \begin{cases}
		-y \xi & X \in (0,k_1), \\
		\lambda \alpha -2\gamma \alpha^2 (X-k_1-k_0) - y \xi & X \in (k_1,k_2), \\
		\lambda \tilde{\alpha}-2\gamma \tilde{\alpha} \left[\tilde{\alpha}(X-k_2)+\alpha(k_2-k_1-k_0)\right] - y \xi & X \in (k_2,\infty).
	\end{cases}
\end{align*}
Hence, it follows that we get the following possible maximal points (if they are in the respective interval): $X_1=0$, $X_2 = k_0 + k_1 + \frac{\lambda \alpha-y \xi}{2\gamma\alpha^2}$, and $X_3 = k_2 + \frac{\lambda \tilde{\alpha} -y \xi}{2 \gamma \tilde{\alpha}^2} - \frac{\alpha}{\tilde{\alpha}} (k_2-k_1-k_0)$. To have $X_2$ and $X_3$ in the correct interval, we get the following conditions (since $\xi >0$):
\begin{align}
	X_2 \in [k_1,k_2] &&\Leftrightarrow&&& \xi \in \left[\dfrac{\lambda \alpha - 2\gamma \alpha^2(k_2-k_1-k_0)}{y},\dfrac{\lambda \alpha}{y} + \dfrac{2 \gamma \alpha^2 k_0}{y}\right] \notag \\
	&&\Leftrightarrow&&& \xi \in \left[\alpha \hat{\xi}, \bar{\xi} \right], \label{equation: condition X2} \\
	X_3 \in [k_2,\infty) &&\Leftrightarrow&&& \xi \in \left(0, \dfrac{\lambda\tilde{\alpha}-2\gamma \tilde{\alpha} \alpha (k_2-k_1-k_0)}{y} \right] = \left( 0, \tilde{\alpha} \hat{\xi}\right]. \label{equation: condition X3}
\end{align}
Moreover, for the maximum of $L$, we must consider the different boundary values of the corresponding intervals. The first boundary value $k_1$ is always dominated by $X_1$ due to $L$ being non-increasing in $[0,k_1]$. The second boundary value $k_2$ is dominated by $X_2$ or by $X_3$ if $\xi$ is in one of the intervals from \eqref{equation: condition X2} or \eqref{equation: condition X3}. If $\xi > \bar{\xi}$, then $L$ is non-increasing in $[0,\infty)$ since then $\frac{\partial L}{\partial X} \leq 0$ for all $X \geq 0$ and hence $k_2$ is dominated by $X_1$. If $\xi \in [\tilde{\alpha} \hat{\xi},\alpha \hat{\xi}]$, then $L$ has a local maximum in $k_2$ since $\frac{\partial L}{\partial X} \geq 0$ for $X \in (k_1,k_2)$ and $\frac{\partial L}{\partial X} \leq 0$ for $X>k_2$. Hence, we add 
\begin{align} \label{equation: condition X4}
	X_4 = k_2 \quad \text{for} \quad \xi \in [\tilde{\alpha} \hat{\xi},\alpha \hat{\xi}]
\end{align}
to the list of potential maximum points. We note that the intervals of $\xi$ for $X_2$, $X_3$, and $X_4$ are disjoint except the lower (resp. upper) boundary points of $X_2$ (resp. $X_3$) with the upper (resp. lower) boundary point of $X_4$. However, in this case $X_2$ (resp. $X_3$) coincides with $X_4$. It is remarkable that due to the disjointness of the related intervals for $X_2$, $X_3$, and $X_4$, we do not have to compare these values themselves with each other (for the potential maximum of $L$). Thus, we replace the closed intervals for $\xi$ by left-open, right-closed intervals and denote by $X^*$ the combined solution of $X_2$, $X_3$, and $X_4$ with $X^*=0$ for $\xi > \bar{\xi}$, i.e., 
\begin{align*}
	X^* := \begin{cases}
		k_2 + \dfrac{\lambda \tilde{\alpha} - y \xi}{2 \gamma \tilde{\alpha}^2} - \dfrac{\alpha}{\tilde{\alpha}} (k_2-k_1-k_0) & \xi \in (0,\tilde{\alpha} \hat{\xi}], \\
		k_2 & \xi \in (\tilde{\alpha}\hat{\xi},\alpha \hat{\xi}], \\
		k_0 + k_1 + \dfrac{\lambda \alpha - y \xi}{2 \gamma \alpha^2} & \xi \in (\alpha\hat{\xi},\bar{\xi}], \\
		0 & \xi \in (\bar{\xi},\infty).
	\end{cases}
\end{align*}
Note that $X^*$ takes the same values as $\hat{X}_T$ in \eqref{eq: optimal terminal wealth}, but the intervals are not truncated. Moreover, we observe that $X^*$ is continuous in $\xi$ on $(0,\bar{\xi}]$. Therefore, we only have to compare $L(X^*,y)$, $y$ fixed, with $L(X_1,y)= -\gamma \alpha^2 k_0^2 - \lambda \alpha k_0$. Thus, we compare these values in the following separately for the three cases, i.e., depending on $\xi$, which value out of $\{X_2,X_3,X_4\}$ $X^*$ attains. Note that we consider, from now on, $L$ as a function of $\xi$. To emphasize this, we write $L(X,y,\xi)$.

\underline{Case 1:} $\xi \in (\alpha \hat{\xi},\bar{\xi}]$, i.e., $X^* = X_2 \in [k_1,k_2]$:\\
Then, we get:
\begin{align} \label{LX_2 formula of xi}
	L(X_2,y,\xi) &= \lambda \alpha \dfrac{\lambda \alpha-y \xi}{2\gamma\alpha^2} - \gamma \alpha^2 \left( \dfrac{\lambda \alpha-y \xi}{2\gamma\alpha^2} \right)^2 - y \xi (k_0 + k_1) - y \xi \dfrac{\lambda \alpha-y \xi}{2\gamma\alpha^2} \notag \\
	&= \dfrac{y^2}{4 \gamma \alpha^2} \xi^2 - y\left(k_0 + k_1 + \dfrac{\lambda}{2\gamma \alpha}\right) \xi + \dfrac{\lambda^2}{4 \gamma}.
\end{align}
Then $L(X_2,y,\xi) \geq -\gamma \alpha^2 k_0^2 - \lambda \alpha k_0 = L(X_1,y,\xi)$ if $\xi \leq \xi_-$ or $\xi \geq \xi_+$ with 
\begin{align*}
	\xi_{\pm} &= \dfrac{2 \gamma \alpha^2}{y^2} \left( y\left(k_0+k_1+\dfrac{\lambda}{2\gamma \alpha}\right) \pm \sqrt{y^2\left(k_0 + k_1 + \dfrac{\lambda}{2\gamma \alpha}\right)^2 - \dfrac{y^2}{\gamma \alpha^2} \left( \dfrac{\lambda^2}{4\gamma} + \gamma \alpha^2 k_0^2 + \lambda \alpha k_0\right)} \right) \\
	&= \dfrac{\lambda \alpha}{y} + \dfrac{2\gamma \alpha^2 (k_0 + k_1)}{y} \pm \dfrac{2 \gamma \alpha^2}{y^2} \sqrt{y^2 k_1^2 +\dfrac{y^2k_1\lambda}{\gamma \alpha} + 2 y^2 k_0 k_1} \\
	&= \dfrac{\lambda \alpha}{y} + \dfrac{2\gamma \alpha^2 (k_0 + k_1)}{y} \pm \dfrac{2 \gamma \alpha^2}{y} \sqrt{k_1^2 + k_1 \left(2k_0 + \dfrac{\lambda}{\gamma \alpha}\right)} \\
	&= \bar{\xi} + \dfrac{2 \gamma \alpha^2}{y} \left( \pm \sqrt{k_1^2 + k_1 \left(2k_0 + \dfrac{\lambda}{\gamma \alpha}\right)} + k_1 \right).
\end{align*}
We notice that $\xi_+ > \bar{\xi}$ and $\max \left\{\alpha \hat{\xi}, \xi_-\right\} = \xi_3^*$ where $\xi_3^*$ is defined as in Theorem \ref{optimal wealth}. If $\xi_- < \alpha \hat{\xi}$, $X_2$ is not in the interval $[k_1,k_2)$ which entails that $(\alpha \hat{\xi},\xi_3^*] = \emptyset$. Otherwise, it holds that $\xi_- = \xi_3^*$ by which we get \eqref{eq: optimal terminal wealth} for this case. 

\underline{Case 2:} $\xi \in (0,\tilde{\alpha} \hat{\xi}]$, i.e., $X^* = X_3 \in [k_2,\infty)$:\\
Without loss of generality, we assume that $\hat{\xi}>0$, since otherwise $(0,\tilde{\alpha} \hat{\xi}] = \emptyset$ and there is nothing to show. Then it holds:
\begin{align} \label{LX_3 formula of xi}
	L(X_3,y,\xi) =& \ \lambda \left( \dfrac{\lambda \tilde{\alpha}-y \xi}{2\gamma\tilde{\alpha}} - \tilde{\alpha} \dfrac{\alpha}{\tilde{\alpha}} (k_2-k_1-k_0) + \alpha (k_2-k_1-k_0) \right) \notag \\
	&- \gamma \left( \tilde{\alpha} \dfrac{\lambda \tilde{\alpha}-y \xi}{2\gamma\tilde{\alpha}^2} - \tilde{\alpha} \dfrac{\alpha}{\tilde{\alpha}} (k_2-k_1-k_0) + \alpha (k_2-k_1-k_0) \right)^2 \notag \\
	&- y \xi k_2 - y \xi \dfrac{\lambda \tilde{\alpha}-y \xi}{2\gamma\tilde{\alpha}^2} + y \xi \dfrac{\alpha}{\tilde{\alpha}} (k_2-k_1-k_0) \notag \\
	=& \ \dfrac{y^2}{4 \gamma \tilde{\alpha}^2} \xi^2 - y\left(k_2 - \dfrac{\alpha}{\tilde{\alpha}} (k_2-k_1-k_0) + \dfrac{\lambda}{2\gamma \tilde{\alpha}}\right) \xi + \dfrac{\lambda^2}{4 \gamma}.
\end{align}
Note that $k_2 - \frac{\alpha}{\tilde{\alpha}} (k_2-k_1-k_0) = \frac{\alpha}{\tilde{\alpha}} (k_0+k_1) - \frac{\alpha_2}{\tilde{\alpha}} k_2$. First we assume that the discriminant, i.e., the term under the square root in the following equation, is non-negative. Then, we have $L(X_3,y,\xi) \geq -\gamma \alpha^2 k_0^2 - \lambda \alpha k_0 = L(X_1,y,\xi)$ if $\xi \leq \xi_-$ or $\xi \geq \xi_+$ with
\begin{align*}
	\xi_{\pm} =& \dfrac{2 \gamma \tilde{\alpha}^2}{y^2} \left( y\left(k_2- \dfrac{\alpha}{\tilde{\alpha}} (k_2-k_1-k_0)+\dfrac{\lambda}{2\gamma \tilde{\alpha}}\right) \right) \\
	&\pm \dfrac{2 \gamma \tilde{\alpha}^2}{y^2} \sqrt{y^2\left(k_2 - \dfrac{\alpha}{\tilde{\alpha}} (k_2-k_1-k_0) + \dfrac{\lambda}{2\gamma \tilde{\alpha}}\right)^2 - \dfrac{y^2}{\gamma \tilde{\alpha}^2} \left(\dfrac{\lambda^2}{4 \gamma} + \gamma \alpha^2 k_0^2 + \lambda \alpha k_0 \right) } \\
	=& \dfrac{2 \gamma \tilde{\alpha}^2}{y^2} \left( y\left(k_2- \dfrac{\alpha}{\tilde{\alpha}} (k_2-k_1-k_0)+\dfrac{\lambda}{2\gamma \tilde{\alpha}}\right) \right) \\
	&\pm \dfrac{2 \gamma \tilde{\alpha}^2}{y^2} \sqrt{y^2\left(\frac{\alpha^2}{\tilde{\alpha}^2} (2k_0k_1 + k_1^2) + \frac{\alpha_2^2}{\tilde{\alpha}^2} k_2^2 - 2\frac{\alpha \alpha_2}{\tilde{\alpha}^2} (k_0+k_1) k_2 + \frac{\lambda}{\gamma \tilde{\alpha}^2} (\alpha k_1 - \alpha_2 k_2) \right)} \\
	=& \dfrac{\lambda\tilde{\alpha}}{y} - \dfrac{2\gamma \alpha \tilde{\alpha}(k_2-k_1-k_0)}{y} + \dfrac{2 \gamma \tilde{\alpha}^2 k_2}{y}	\\
	&\pm \dfrac{2\gamma \tilde{\alpha}}{y} \sqrt{\alpha^2 (2k_0k_1 + k_1^2) + \alpha_2^2 k_2^2 - 2\alpha \alpha_2 (k_0+k_1) k_2 + \frac{\lambda}{\gamma} (\alpha k_1 - \alpha_2 k_2) } \\
	=& \tilde{\alpha} \hat{\xi}+ \dfrac{2\gamma \tilde{\alpha}}{y} \left(\pm \sqrt{ (\alpha (k_0 + k_1) - \alpha_2 k_2)^2 - \alpha^2 k_0^2 + \frac{\lambda}{\gamma} (\alpha k_1 - \alpha_2 k_2) } + \tilde{\alpha} k_2 \right).
\end{align*}
Again, we notice that $\xi_+ > \tilde{\alpha} \hat{\xi}$ and since the discriminant is non-negative, we get $\xi_- = \tilde{\xi}_1^*$ where we can ignore this case if $\xi_- \leq 0$ due to the assumption that $\xi$ has to be in the interval $(0,\tilde{\alpha} \hat{\xi}]$. Now, if the discriminant is negative, then $L(X_3,y,\xi) \geq -\gamma \alpha^2 k_0^2 - \lambda \alpha k_0$ for all $\xi$ and $\tilde{\xi}_1^* > \tilde{\alpha} \hat{\xi}$. Hence, $X_3$ is optimal for all $\xi \in ( 0, \tilde{\alpha} \hat{\xi}]$ and it holds that $\xi_1^*= \tilde{\alpha} \hat{\xi}$ which implies \eqref{eq: optimal terminal wealth} for this case. 

\underline{Case 3:} $\xi \in (\tilde{\alpha} \hat{\xi},\alpha \hat{\xi}]$, i.e., $X^* = X_4 = k_2$:\\
Now, we get:
\begin{align} \label{LX_4 formula of xi}
	L(X_4,y,\xi) &= \lambda \alpha (k_2-k_1-k_0) - \gamma \alpha^2 (k_2-k_1-k_0)^2 - y \xi k_2.
\end{align}
Then  $L(X_4,y,\xi) \geq -\gamma \alpha^2 k_0^2 -\lambda \alpha k_0 = L(X_1,y,\xi)$ if 
\begin{align*}
	\xi \leq \dfrac{\lambda \alpha}{y} - \dfrac{\gamma \alpha^2(k_2-k_1)^2-2\gamma \alpha^2 k_0 (k_2-k_1)+\lambda\alpha k_1}{y k_2} = \tilde{\xi}_2^*.
\end{align*}
Hence, it follows that $X_4$ is optimal for $\xi \in (\tilde{\alpha} \hat{\xi},\tilde{\xi}_2^*]$ which implies \eqref{eq: optimal terminal wealth} in this case as well. Hence, \eqref{eq: optimal terminal wealth} holds in total. 

Finally, we must have that $\xi^*>0$ since otherwise $(0,\xi_1^*] \cup (\tilde{\alpha} \hat{\xi},\xi_2^*] \cup (\alpha \hat{\xi} ,\xi_3^*] = \emptyset$ which implies that $\hat{X}_T \equiv 0$. Hence, the budget constraint $\EX [\xi_T \hat{X}_T (y)]  = \xi_0 x_0$ cannot be fulfilled, which is a contradiction to Proposition \ref{y lambda exist} whose proof is given below. The property of $\hat{X}_T$ follows directly from Proposition \ref{prop: * Eigenschaften} and $y < \infty$.
\end{myproof}

We deduce from Theorem \ref{optimal wealth}:

\begin{corollary}
	The optimal payoff of the insurer is given by:
	\begin{align*}
		F(\hat{X}_T) = \begin{cases}
			\dfrac{\lambda \tilde{\alpha} - y \xi_T}{2 \gamma \tilde{\alpha}}  & \xi_T \in (0,\xi_1^*], \\
			\alpha (k_2-k_1-k_0) & \xi_T \in (\tilde{\alpha}\hat{\xi},\xi_2^*], \\
			\dfrac{\lambda \alpha - y \xi_T}{2 \gamma \alpha} & \xi_T \in (\alpha\hat{\xi},\xi_3^*], \\
			-\alpha k_0 & \text{else,}
		\end{cases}
	\end{align*}
	where $\hat{\xi}$, $\xi_1^*$, $\xi_2^*$, and $\xi_3^*$ are defined as in Theorem \ref{optimal wealth}.
\end{corollary}

\begin{proof}
	This proof follows immediately by simple calculations when we plug \eqref{eq: optimal terminal wealth} into the function $F$.
\end{proof}

In the following Figure \ref{fig: Optimal Value}, we derive the shape of the optimal value $\hat{X}_T$ and the corresponding wealth of the insurance company as a function of $\xi_T$.

\begin{figure}[!htb]
	\centering
	\begin{minipage}{0.48\textwidth}
		\includegraphics[trim= 1mm 6mm 10mm 8mm, clip,width=\textwidth]{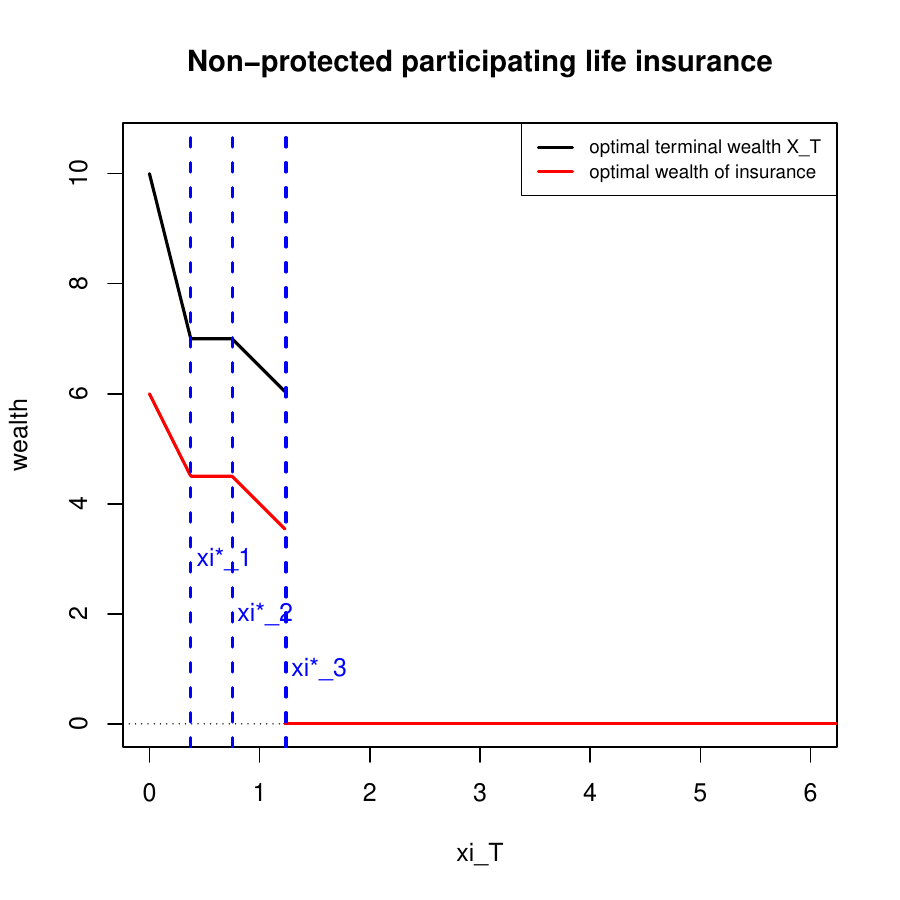}
	\end{minipage}
	\quad
	\begin{minipage}{0.48\textwidth}
		\includegraphics[trim= 1mm 6mm 10mm 8mm, clip,width=\textwidth]{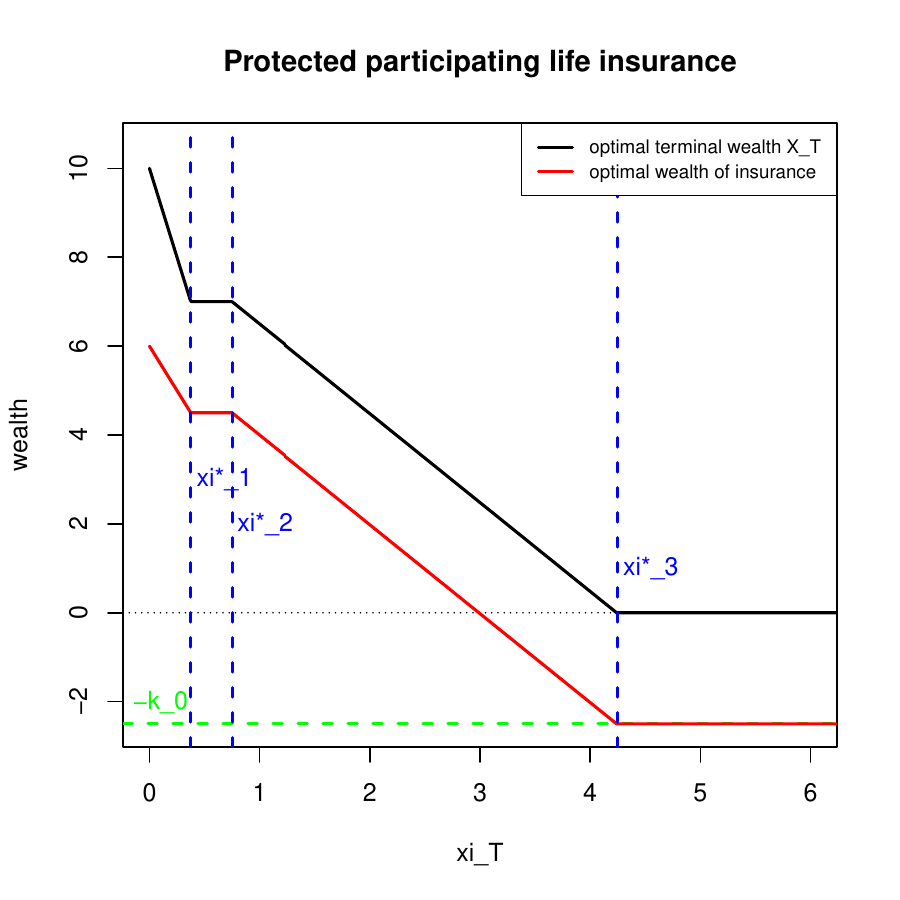}
	\end{minipage}
	\caption{The optimal value $\hat{X}_T$ and the optimal terminal wealth of the insurance company as a function of $\xi_T$ with $k_1 = 2.5$ (left) resp. $k_0 = 2.5$ (right), $k_2 = 7$, $\alpha=1$, $\alpha_2=0.5$, $\lambda=3$, $y=1$, $\gamma=0.25$.}
	\label{fig: Optimal Value}
\end{figure}	

\subsection{Derivation of the optimal strategy} \label{chapter: Black-Scholes}

Next, we derive the optimal strategy for our problem where the proof is somewhat technical and therefore given in Appendix \ref{proofs}.

\begin{theorem} \label{optimal strategy}
	The optimal solution $\hat{u}$ is non-negative and given by:
	\begin{align*}
		\hat{u}_t = (\sigma_t^T)^{-1} \kappa_t \dfrac{v_t}{\hat{X}_t},
	\end{align*}
	where
	\begin{align*}
		v_t =& \left(k_2 + \dfrac{\lambda}{2 \gamma \tilde{\alpha}} - \dfrac{\alpha}{\tilde{\alpha}}(k_2-k_1-k_0) \right) \dfrac{e^{-\int_t^Tr_s \diff s}}{\sqrt{\int_t^T \norm{\kappa_s}^2 \diff s}} \varphi\left(d_1 \left(\xi_1^*,t\right)\right) \\
		&+ \dfrac{y}{2\gamma \tilde{\alpha}^2} \xi_t e^{\int_t^T -(2r_s-\norm{\kappa_s}^2) \diff s} \left[ \Phi \left( d_2 \left(\xi_1^*,t\right) \right) - \dfrac{1}{\sqrt{\int_t^T \norm{\kappa_s}^2 \diff s}} \varphi\left(d_2 \left(\xi_1^*,t\right)\right) \right] \\
		&+ k_2 \dfrac{e^{-\int_t^Tr_s \diff s}}{\sqrt{\int_t^T \norm{\kappa_s}^2 \diff s}} \left( \varphi\left(d_1 \left(\xi_2^*,t\right)\right) - \varphi\left(d_1 \left(\tilde{\alpha} \hat{\xi},t\right)\right) \right) \\
		&+ \left(k_0 + k_1 + \dfrac{\lambda}{2 \gamma \alpha} \right) \dfrac{e^{-\int_t^Tr_s \diff s}}{\sqrt{\int_t^T \norm{\kappa_s}^2 \diff s}} \left( \varphi\left(d_1 \left(\xi_3^*,t\right)\right) - \varphi\left(d_1 \left(\alpha \hat{\xi},t\right)\right) \right) \\
		&+ \dfrac{y}{2\gamma \alpha^2} \xi_t e^{\int_t^T -(2r_s-\norm{\kappa_s}^2) \diff s} \left[ \left(\Phi \left( d_2 \left(\xi_3^*,t\right) \right) - \Phi \left( d_2 \left(\alpha \hat{\xi},t\right) \right) \right) \right. \\
		&\hspace{128pt}- \left. \dfrac{1}{\sqrt{\int_t^T \norm{\kappa_s}^2 \diff s}} \left(\varphi\left(d_2 \left(\xi_3^*,t\right)\right) - \varphi\left(d_2 \left(\alpha \hat{\xi},t\right)\right) \right) \right], \\
		\hat{X}_t =& \left(k_2 + \dfrac{\lambda}{2 \gamma \tilde{\alpha}} - \dfrac{\alpha}{\tilde{\alpha}}(k_2-k_1-k_0) \right) e^{-\int_t^Tr_s \diff s} \Phi\left(d_1 \left(\xi_1^*,t\right)\right) \\
		&- \dfrac{y}{2\gamma \tilde{\alpha}^2} \xi_t e^{\int_t^T -(2r_s-\norm{\kappa_s}^2) \diff s} \Phi \left( d_2 \left(\xi_1^*,t\right) \right)  \\
		&+ k_2 e^{-\int_t^Tr_s \diff s} \left( \Phi\left(d_1 \left(\xi_2^*,t\right)\right) - \Phi\left(d_1 \left(\tilde{\alpha} \hat{\xi},t\right)\right) \right) \\
		&+ \left(k_0+ k_1 + \dfrac{\lambda}{2 \gamma \alpha} \right) e^{-\int_t^Tr_s \diff s} \left( \Phi\left(d_1 \left(\xi_3^*,t\right)\right) - \Phi\left(d_1 \left(\alpha \hat{\xi},t\right)\right) \right) \\
		&- \dfrac{y}{2\gamma \alpha^2} \xi_t e^{\int_t^T -(2r_s-\norm{\kappa_s}^2) \diff s} \left(\Phi \left( d_2 \left(\xi_3^*,t\right) \right) - \Phi \left( d_2 \left(\alpha \hat{\xi},t\right) \right) \right)
	\end{align*}
	with
	\begin{align*}
		d_1 (x,t) =& \dfrac{\ln x - \ln \xi_t + \int_t^T (r_s - \frac{\norm{\kappa_s}^2}{2}) \diff s}{\sqrt{\int_t^T \norm{\kappa_s}^2 \diff s}}, \\
		d_2 (x,t) =& \dfrac{\ln x - \ln \xi_t + \int_t^T (r_s - \frac{3\norm{\kappa_s}^2}{2}) \diff s}{\sqrt{\int_t^T \norm{\kappa_s}^2 \diff s}} = d_1(x,t) - \sqrt{\int_t^T \norm{\kappa_s}^2 \diff s},
	\end{align*}
	where $\Phi$ denotes the cdf and $\phi$ the density of a standard normal distributed random variable. The values $\hat{\xi}$, $\xi_1^*$, $\xi_2^*$, and $\xi_3^*$ are defined as in Theorem \ref{optimal wealth}.
\end{theorem}

\section{Numerical Results} \label{numerics}

In the following section, we discuss the numerical results for the Black-Scholes model with a given calibration of the parameters with a setting as in Section \ref{chapter: complete market}. The first two figures show the impact of different values for the risk aversion resp. the participation rate and of different optimization strategies (\ac{eu} instead of mean-variance) on the optimal terminal wealth as a function of $\xi_T$. The remaining four figures show simulation outcomes for the optimal terminal wealth and the optimal strategy. We analyze the results for the non-protected resp. the protected participation life insurance product, and compare those (a) to each other, (b) to a product with no surplus participation, and (c) to the results with an \ac{eu}-optimization.

For the participating life insurance contract, we use the parameters $k_0 = 0$, $k_1 = 2.5$, $k_2 = 7$, $\alpha = 1$, and $\alpha_2 = 0.25$ for the non-protected insurance product. We change the values of $k_0$ and $k_1$, i.e., $k_0 = 2.5$ and $k_1 = 0$, to get the protected product. The initial wealth, i.e., the sum of the premiums and the initial capital from the insurance company, is given by $x_0 = 4$, the time horizon by $T=10$, and the risk aversion by $\gamma = 0.25$. We assume a constant risk-free interest rate of $r := r_t \equiv 0.02$ for the financial market and consider one risky asset with constant mean $\mu := \mu_t \equiv 0.08$ and volatility $\sigma_t \equiv 0.2$. In particular, the Sharpe ratio is also constant and given by $\kappa_t \equiv 0.3$. For the implementation, we used a time step of $\delta = 0.01$, i.e., we have 1000 intermediate points. Moreover, we simulate {10000} realizations of the Brownian Motion.

In Figure \ref{fig: xiplot}, we present the optimal terminal wealth as a function of $\xi_T$ for different values of the risk aversion parameter $\gamma$. (Note that the risk aversion parameters are not comparable for the different optimization strategies.) For both the non-protected and the protected case, we compare the optimal terminal wealth stemming from our mean-variance optimization with the \ac{eu}-optimization for participating life insurance contracts from Lin et al. \cite{lin2017optimal}. To measure the utility, we use, as Lin et al., the following S-shaped utility function: $U(x) = \begin{cases}
	x^{\tilde{\gamma}} & x \geq 0,\\
	-\tilde{\lambda} (-x)^{-\tilde{\gamma}} & x < 0
\end{cases}$ with $\tilde{\lambda}=2$ and different values for $\tilde{\gamma}$ (Lin et al. \cite{lin2017optimal} used $\tilde{\gamma}=0.5$ in their numerical analysis). Moreover, we added the case that there is no additional surplus participation of the policyholders, i.e., $\alpha_2=0$, $k_0=0$, and $k_1=0$. Both participating optimizations share that the optimal wealth function has three points of significant behavioral change, i.e., at $\xi_1^*$, $\xi_2^*$ and $\xi_3^*$ as in Theorem \ref{optimal wealth}. This theorem shows that the optimal wealth decreases before $\xi_1^*$ and between $\xi_2^*$ and $\xi_3^*$, and that the optimal wealth is constant elsewhere. The main difference in the shape of the functions is that we have a piecewise linearity in the mean-variance case but a total non-linearity in the \ac{eu} case. This effect is already apparent from the respective formulas. Noticeably, the distance between the second and third point of behavioral change is much higher in the \ac{eu}-optimization than in the mean-variance one. Moreover, the figure shows a moral hazard of the insurance companies offering a non-protected participating life insurance contract since the insurance companies favor a portfolio value of $0$ over a portfolio value at the guarantee value level. We can observe this effect for both the mean-variance and the \ac{eu}-optimization in the figure due to the drop-down of the optimal terminal wealth to zero. Protected products have no such effect when optimizing mean-variance since the insurance company always profits from higher portfolio values. While this effect also occurs when optimizing \ac{eu}, the probability is smaller since the drop-down is at a far higher value (compared to the non-protected product).

\begin{figure}[!htb]
	\centering
	\begin{minipage}{0.48\textwidth}
		\includegraphics[trim= 1mm 6mm 10mm 8mm, clip,width=\textwidth]{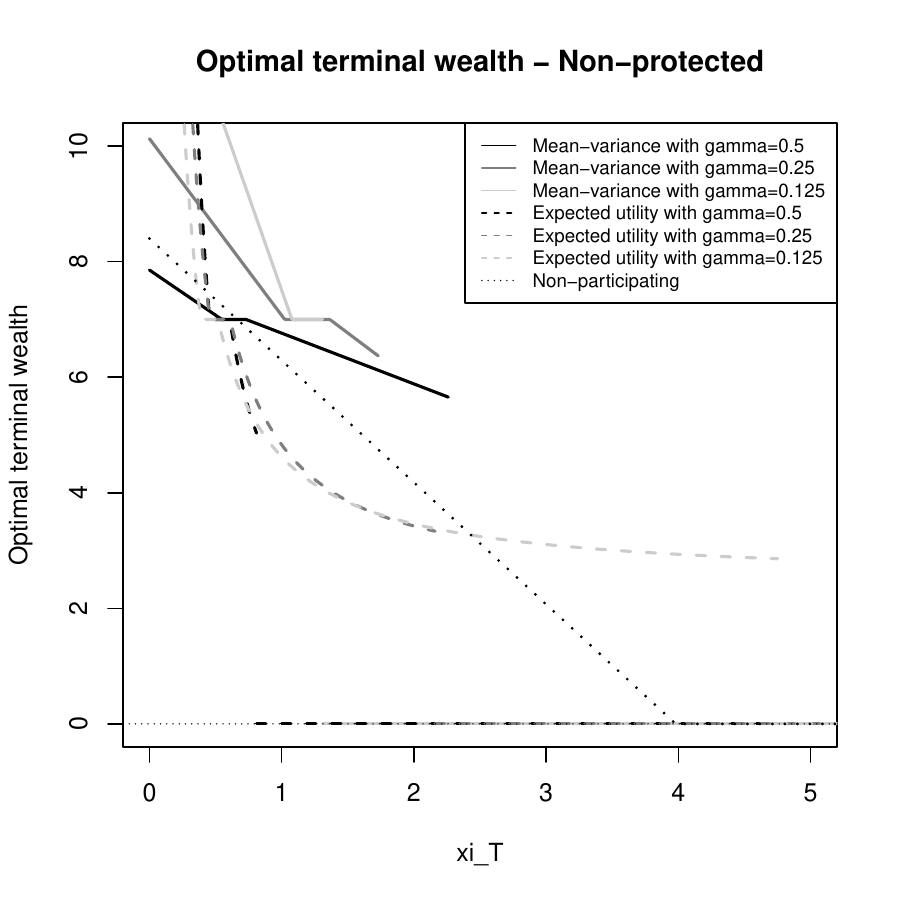}
	\end{minipage}
	\quad
	\begin{minipage}{0.48\textwidth}
		\includegraphics[trim= 1mm 6mm 10mm 8mm, clip,width=\textwidth]{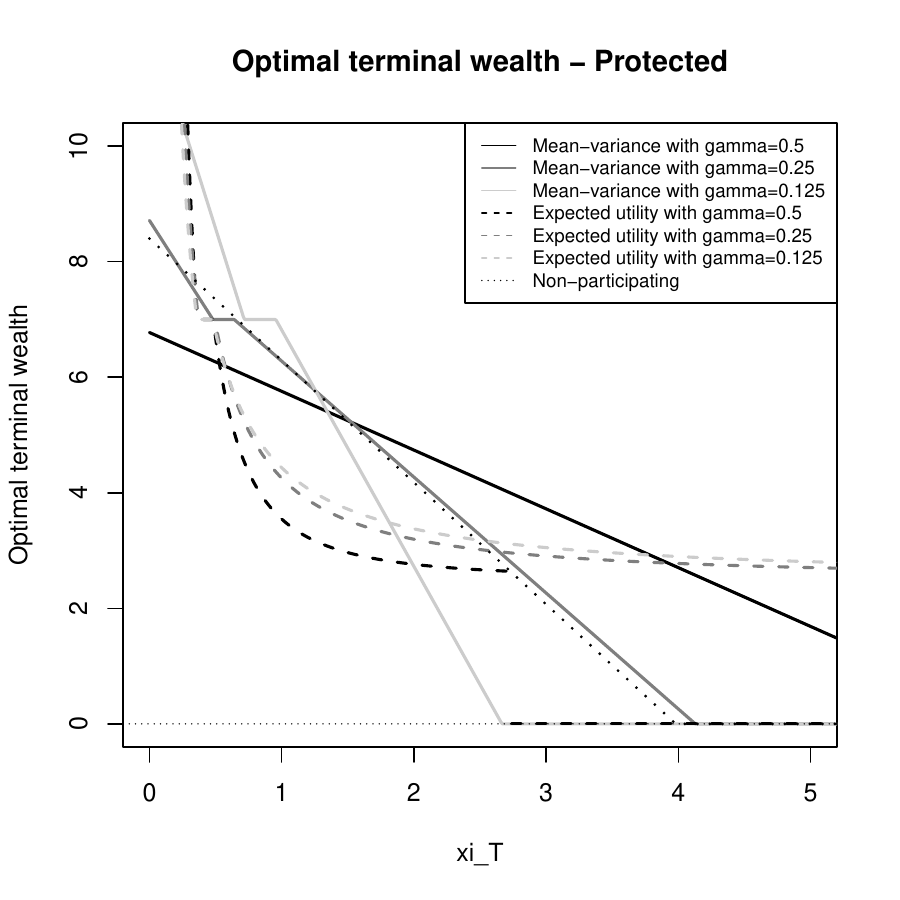}
	\end{minipage}
	\caption{Comparison of the optimal terminal wealth as a function of $\xi_T$ of mean-variance with \ac{eu} optimized participating life insurance products and no participation in the case of a not protected guarantee (left) and a protected one (right). Note that the risk aversion parameters are not comparable for mean-variance and \ac{eu}.}
	\label{fig: xiplot}
\end{figure}	

In Figure \ref{fig: alpha2}, we show the influence of the participation rate $\alpha_2$ on the optimal terminal wealth as a function of $\xi_T$. We can see that a higher rate $\alpha_2$ leads to a more extended plateau at $k_2=7$. This effect is not surprising when checking the formula \eqref{eq: optimal terminal wealth} of the optimal terminal wealth $\hat{X}_T$ since a higher $\alpha_2$ leads to a smaller $\tilde{\alpha}$. Thus, the intermediate interval $(\tilde{\alpha} \hat{\xi},\xi_2^*)$ gets bigger since for our parametrization, it always holds that $\xi_2^* = \alpha \hat{\xi}$. Moreover, we can observe that $\xi^*$ (defined in Theorem \ref{optimal wealth} as the point where in the non-protected case the optimal terminal wealth drops down to $0$) is increasing for an increasing participation rate $\alpha_2$. This effect stems from the different values of the variables $\lambda$ and $y$ when changing the participation rate.
Additionally, the figure shows that for minimal values of $\xi_T$, the optimal terminal wealth is higher for large participation rates, which we can derive directly from the formula of $\hat{X}_T$ in \eqref{eq: optimal terminal wealth}. For the economic interpretation, we infer that a higher participation rate of the policyholders leads to a higher portfolio value in bad economic states, because for a higher $\alpha_2$, the upside potential of the insurer is reduced. Hence, the optimal strategy should be less risky, which implies a higher portfolio value. For good economic states, the opposite effect happens. A somewhat surprising result is that in very good economic states (happening with an extremely low probability), the optimal terminal wealth is higher for higher participation rates, which occurs due to changes in the Lagrange multiplier of the budget constraint.

\begin{figure}[!htb]
	\centering
	\begin{minipage}{0.48\textwidth}
		\includegraphics[trim= 1mm 6mm 10mm 8mm, clip,width=\textwidth]{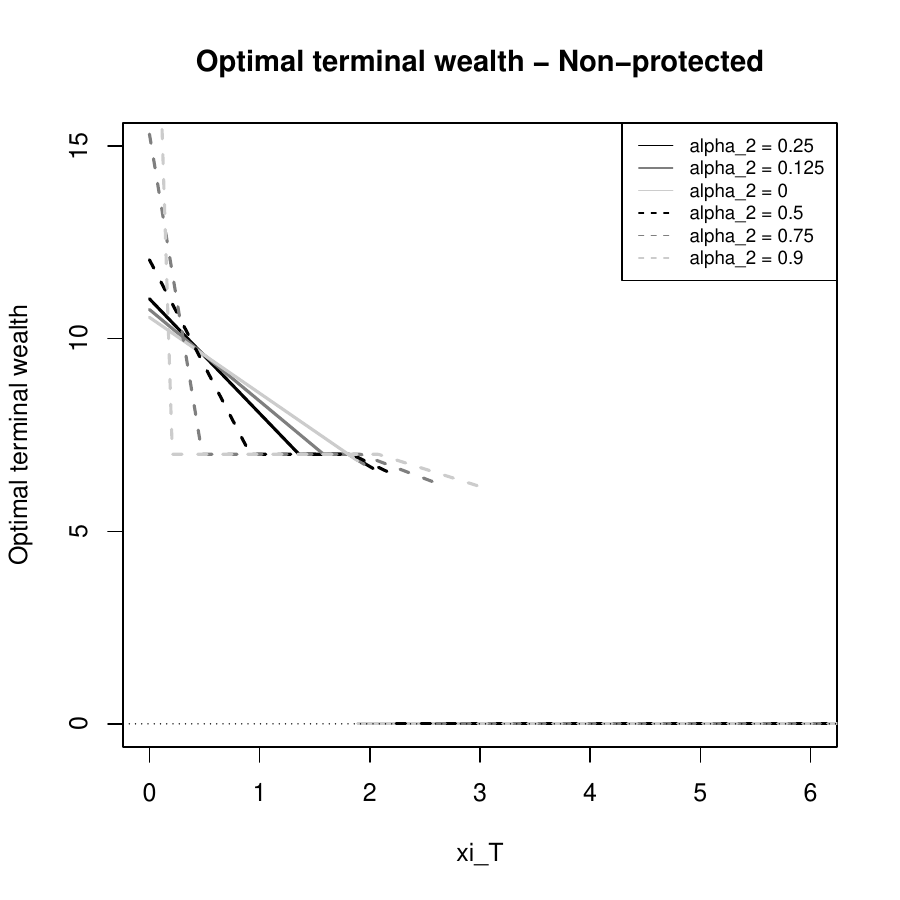}
	\end{minipage}
	\quad
	\begin{minipage}{0.48\textwidth}
		\includegraphics[trim= 1mm 6mm 10mm 8mm, clip,width=\textwidth]{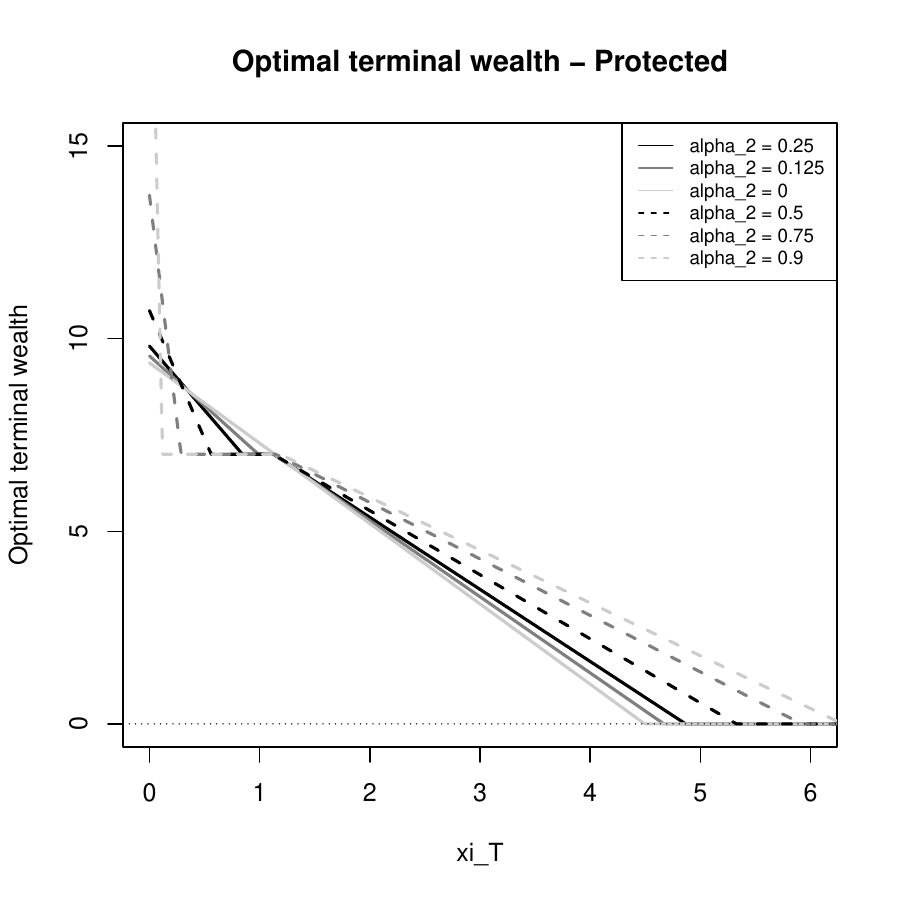}
	\end{minipage}
	\caption{Influence of the participation rate $\alpha_2$ on the optimal terminal wealth as a function of $\xi_T$ of a not protected (left) and a protected (right) participation life insurance.}
	\label{fig: alpha2}
\end{figure}	

{%In both previous figures the consideration of the alternative problem given in Remark \ref{optimal wealth remark}\eqref{item: equivalent-like problem} does not change the illustration. 
In both, Figure \ref{fig: xiplot} and Figure \ref{fig: alpha2}, we can observe a jump of the optimal terminal wealth in the non-protected case, which is typical for problems with option-like constraints, see for instance Basak and Shapiro \cite{basak2001value}, Chen et al. \cite{chen2018optimal}, Nguyen and Stadje \cite{nguyen2020nonconcave}, Chen et al. \cite{chen2024equivalence}, or Kraft and Steffensen \cite{kraft2013dynamic}. The reason is that for certain outcomes, the insurance company's incentive change and therefore it can gain additional utility by changing to a different regime.\footnote{{For instance, if the entire wealth for certain outcomes is transferred to the policyholders, the insurance company will target a terminal wealth of zero in these states, in order to gain additional wealth in states where it benefits more. (Note that the overall expectation is fixed due to the budget constraint.)}} Moreover, the piecewise linearity observed also occurs in, e.g., Avanzi et al. \cite{avanzi2024optimal} who considered a minimization of a quadratic function under solvency constraints including option-like constraints as the expected shortfall. This piecewise linearity is typical for problems with some kind of quadratic penalties, whereas the option-like constraints entail the existence of the plateaus and the jumps.}

In Figure \ref{fig: numerical non-protected}, we show the optimal wealth process and the optimal share into the risky asset over time for the non-protected participating life insurance product. In both figures, the red line shows the average of the {$10000$} realizations, and each of the ten black lines shows a single realization. Note that for the optimal strategy, we used a weighted average with the weight given by the absolute investment amount (which also holds for all other averages in this section). Furthermore, we get $\lambda \approx 3.423$ and $y \approx 0.860$. The average optimal wealth process develops from $4$ to approximately {$7.6$}, where the increase is slightly larger in the beginning than in the end. A terminal wealth of about {$7.6$} corresponds to a terminal wealth for the insurer of {$4.95$} and a wealth of {$2.65$} for the policyholder. Note that this is higher than the guarantee value, which is $2.5$. When exercising the optimal strategy, we start with high investment into the risky asset (approximately $110 \%$ of our wealth), which decreases relatively constant over time. The final optimal investment share into the risky asset is around ${45} \%$. When analyzing the strategies in more detail, one observes that the most risky investments are taken when the economy has poorly evolved until then, i.e., when $\xi_t$ is high, while the least risky investments are taken in cases where the economy has developed well. These results are reasonable since when the wealth is below the guarantee, the insurer has nothing to lose anymore, while the wealth is already over the second threshold $k_2$ (i.e., the policyholder also participates in the surplus over $k_2$), the upside potential of the insurer is reduced, whereas the downside potential is not (or only slightly).

\begin{figure}[!htb]
	\centering
	\begin{minipage}{0.48\textwidth}
		\includegraphics[trim= 1mm 6mm 10mm 8mm, clip,width=\textwidth]{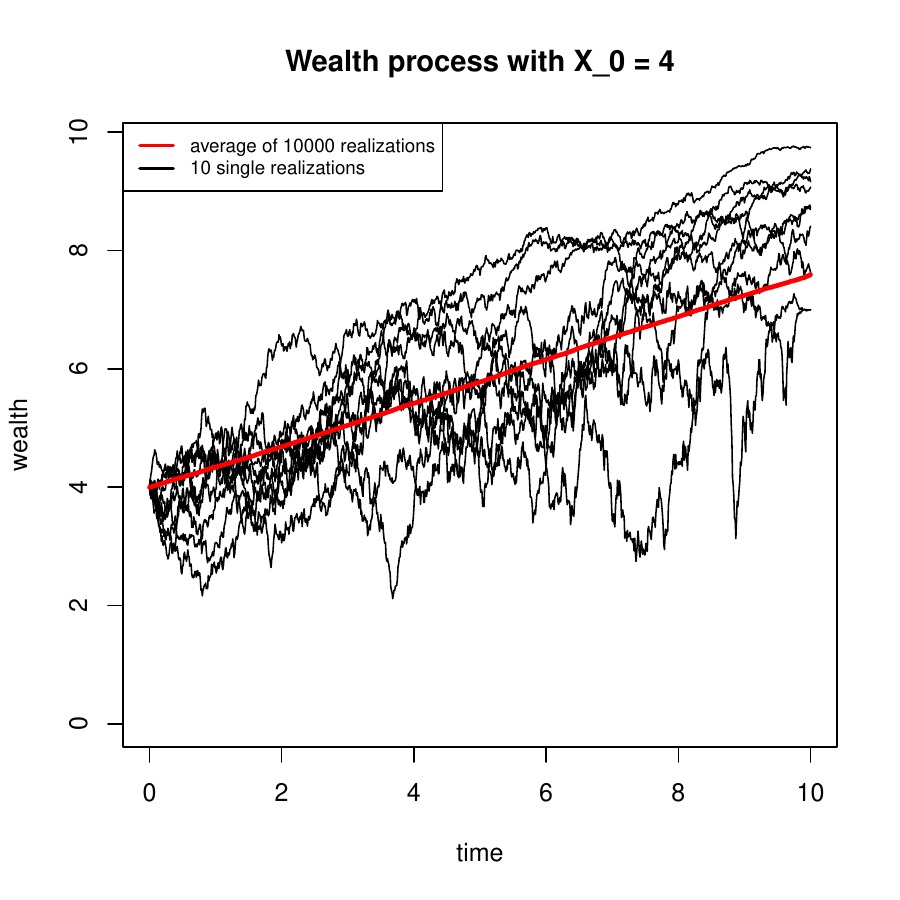}
	\end{minipage}
	\quad
	\begin{minipage}{0.48\textwidth}
		\includegraphics[trim= 1mm 6mm 10mm 8mm, clip,width=\textwidth]{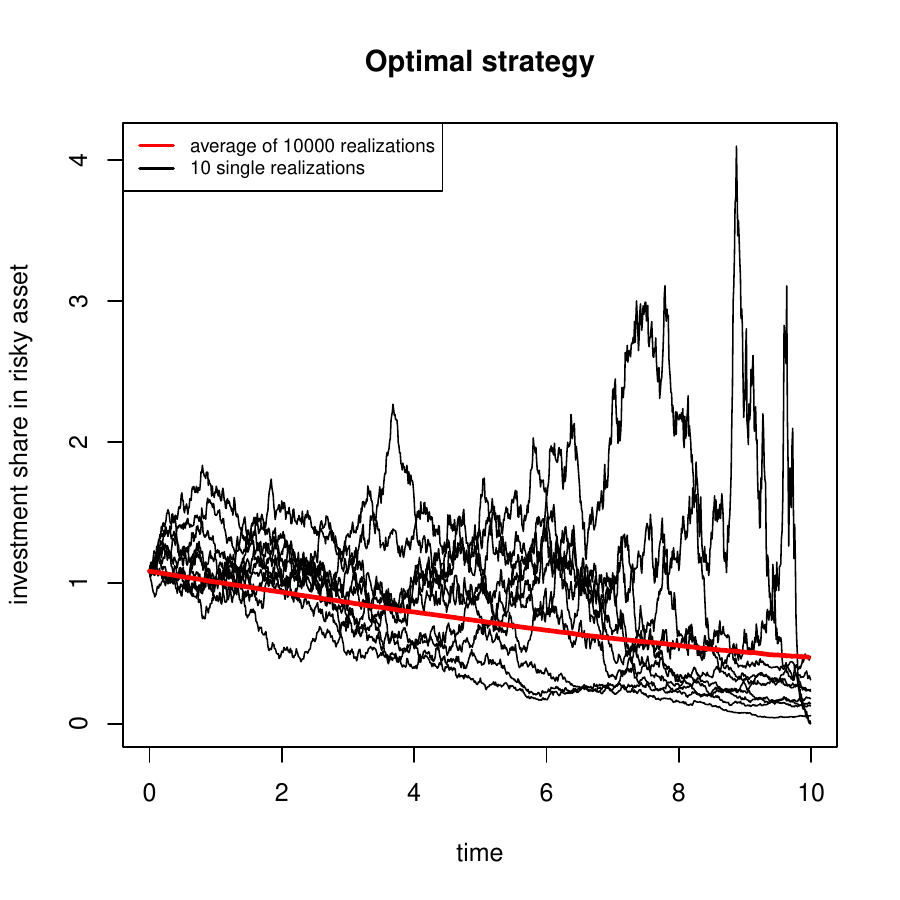}
	\end{minipage}
	\caption{Optimal wealth and optimal strategy for the non-protected participating life insurance product.}
	\label{fig: numerical non-protected}
\end{figure}	

In Figure \ref{fig: numerical protected}, we show the optimal wealth process and the optimal share into the risky asset over time for the protected participating life insurance product. Again, the red line shows the average of the {$10000$} realizations, and each of the ten black lines shows a single realization. Moreover, the variables $\lambda$ and $y$ are given by $\lambda \approx 2.893$ and $y \approx 1.003$. In this case, the optimal value evolves only to around {$6.7$}, corresponding to a terminal wealth of the insurance company of {$4.2$} since we are below $k_2$, i.e., there is no surplus participation of the policyholder. The optimal strategy starts risky but not as risky as in the non-protected case, with around $75\%$ of the wealth invested into the risky asset. It also decreases over time, but the reduction is smaller, and the final optimal investment percentage is at around {$35\%$} of the wealth. As in the non-protected case, the realizations in the best economic states correspond with the least risky investments and vice versa. One of the shown strategies is rather inconspicuous over most of the time but has a major change close to the final time point $T=10$, i.e., a drop-down from $u_{9.93} = 0.2$ to $u_{9.99} = 0.00001$. Such an effect (also in the other way, i.e., an upward move) happens for several realizations when their wealth for time points $t$, which are close to the maturity $T$, is close to the second threshold $k_2$ since the optimal terminal wealth as a function of $\xi$ has a small plateau at $k_2$, i.e., if $\xi_t \in (\tilde{\alpha} \hat{\xi}, \xi_2^*]$. Hence, for $\xi_t \approx \tilde{\alpha} \hat{\xi}$ (resp. $\xi_t \approx \tilde{\alpha} \xi_2^*$) for $t$ close to $T$, the upside (resp. downside) potential is reduced and the optimal strategy is to invest safely (resp. risky). 

\begin{figure}[!htb]
	\centering
	\begin{minipage}{0.48\textwidth}
		\includegraphics[trim= 1mm 6mm 10mm 8mm, clip,width=\textwidth]{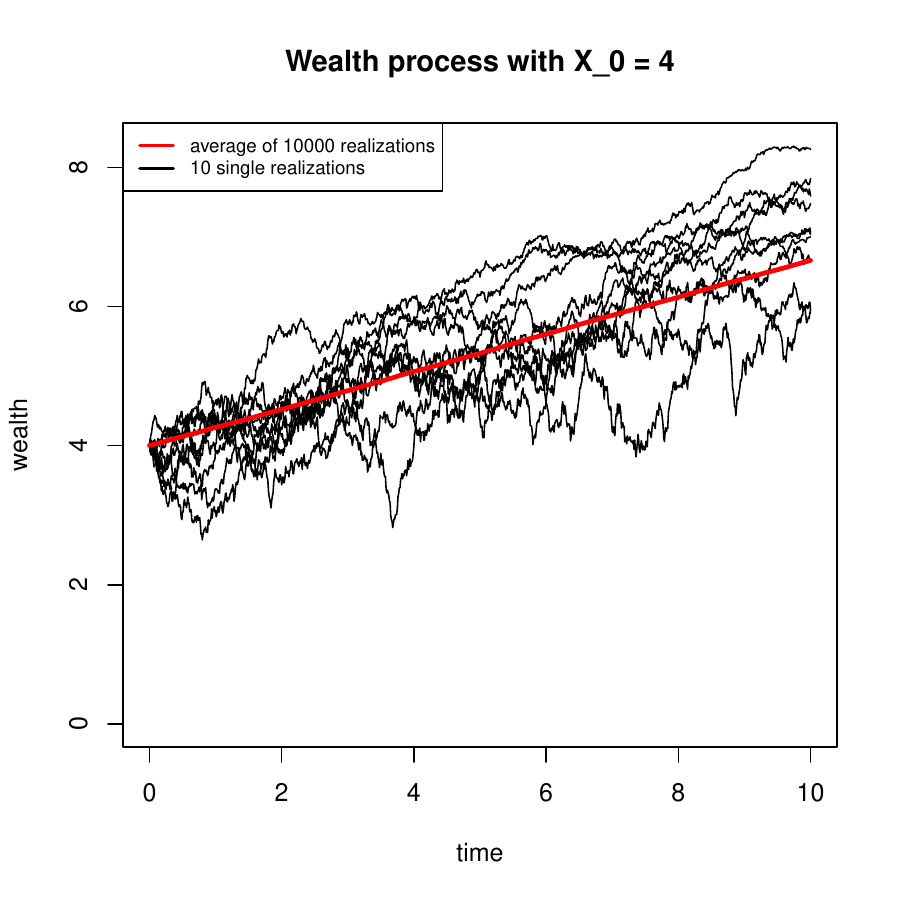}
	\end{minipage}
	\quad
	\begin{minipage}{0.48\textwidth}
		\includegraphics[trim= 1mm 6mm 10mm 8mm, clip,width=\textwidth]{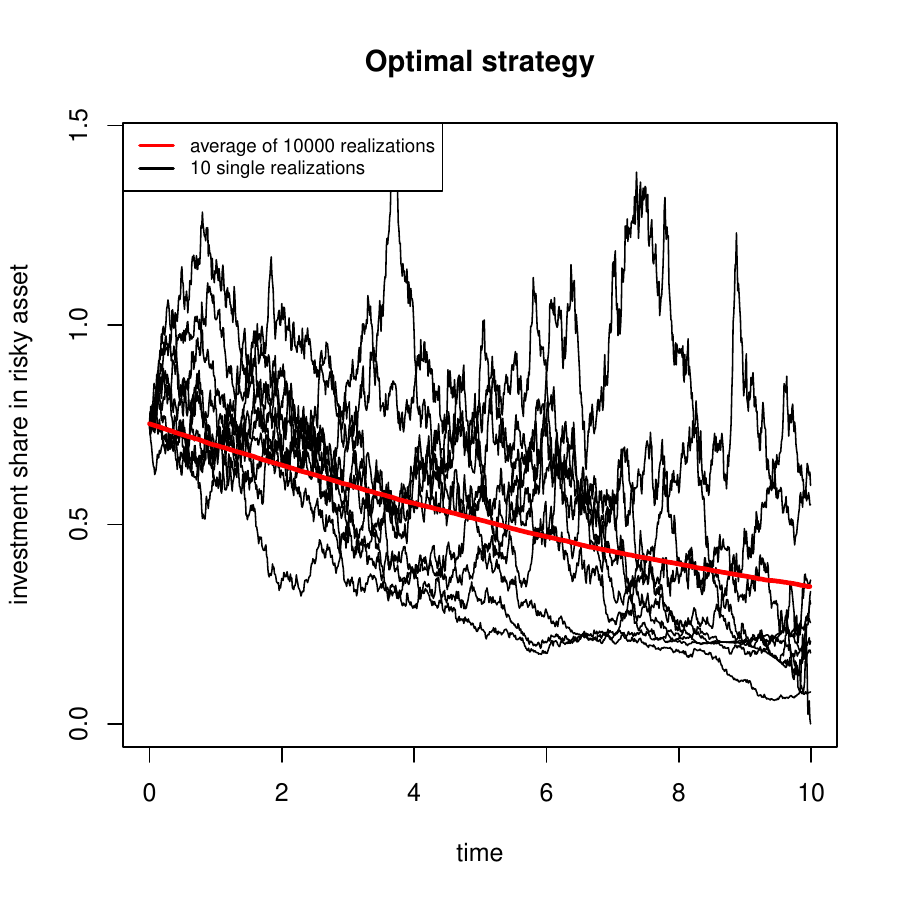}
	\end{minipage}
	\caption{Optimal wealth and optimal strategy for the protected participating life insurance product.}
	\label{fig: numerical protected}
\end{figure}	

In Figure \ref{fig: numerical compare}, we compare the optimal wealth before splitting it between the policyholders and the insurance company and the optimal strategy of our two participating life insurance products with mean-variance to a non-participating investment, i.e., we set $\alpha_2=0$, $k_0=0$, and $k_1=0$. 
Moreover, we set the initial wealth such that the mean of $\hat{X}_T$ is approximately equal to {$7.59$}. Therefore, we get the following initial values: $x_0 = {\scriptsize \begin{cases}
	4 & \text{ non-protected product}, \\
	{4.692} & \text{ protected product}, \\
	{4.665} & \text{ no participation}.
\end{cases}}$ The lines show again an average of {$10000$} realizations each. The figure shows that the investor makes the riskiest investment when offering a non-protected insurance product and approximately identical investments when offering one of the other two products. When offering the protected product compared to the non-participation product, it is remarkable that the investor invests slightly safer in the beginning and somewhat riskier close to maturity. This investment behavior leads to the highest wealth gain for the non-protected product since, on average, the risky asset performs better than the risk-less asset (due to $\mu > r$). For the other two products, the wealth gain is approximately equal. The protected life insurance investment strategy is relatively close to the non-participating strategy since a fixed-guarantee payment only minimally influences the optimal strategy. Then, the only remaining difference in the payoff shape is when the wealth is above $k_2$, which happens in our chosen parametrization either relatively late or not at all (see Figure \ref{fig: numerical protected}). The payoff structure for the non-protected insurance product differs highly in the values due to the reduced downside potential below $k_1$. 

\begin{figure}[!htb]
	\centering
	\begin{minipage}{0.48\textwidth}
		\includegraphics[trim= 1mm 6mm 10mm 8mm, clip,width=\textwidth]{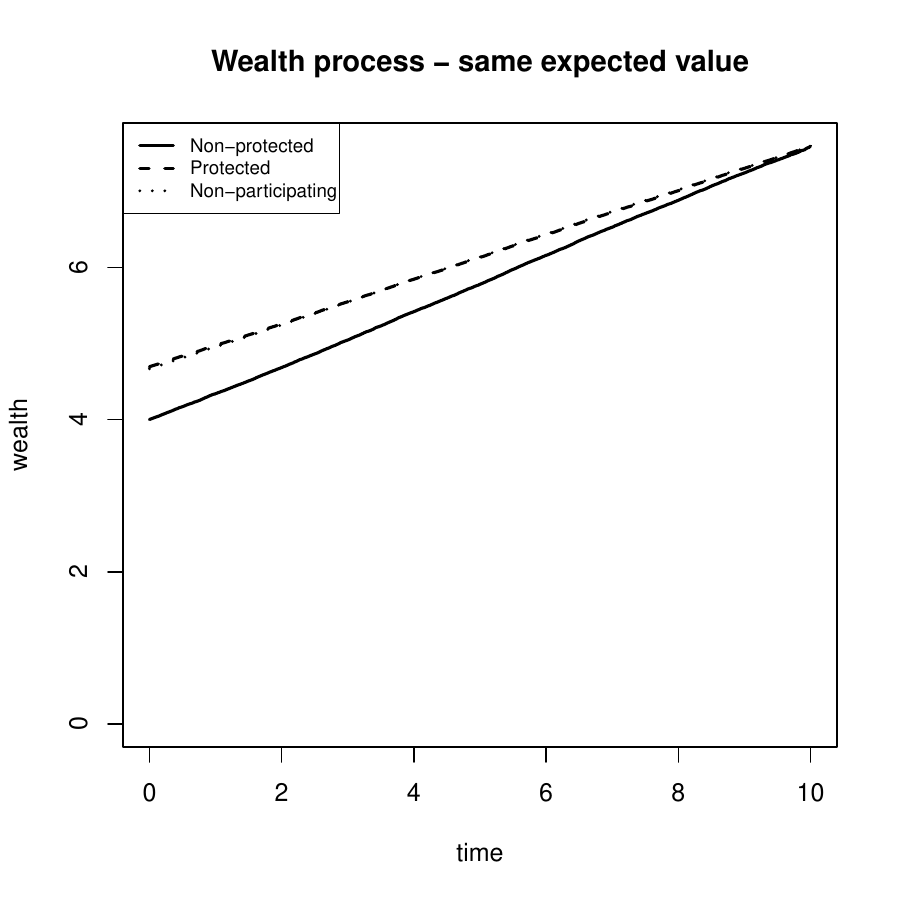}
	\end{minipage}
	\quad
	\begin{minipage}{0.48\textwidth}
		\includegraphics[trim= 1mm 6mm 10mm 8mm, clip,width=\textwidth]{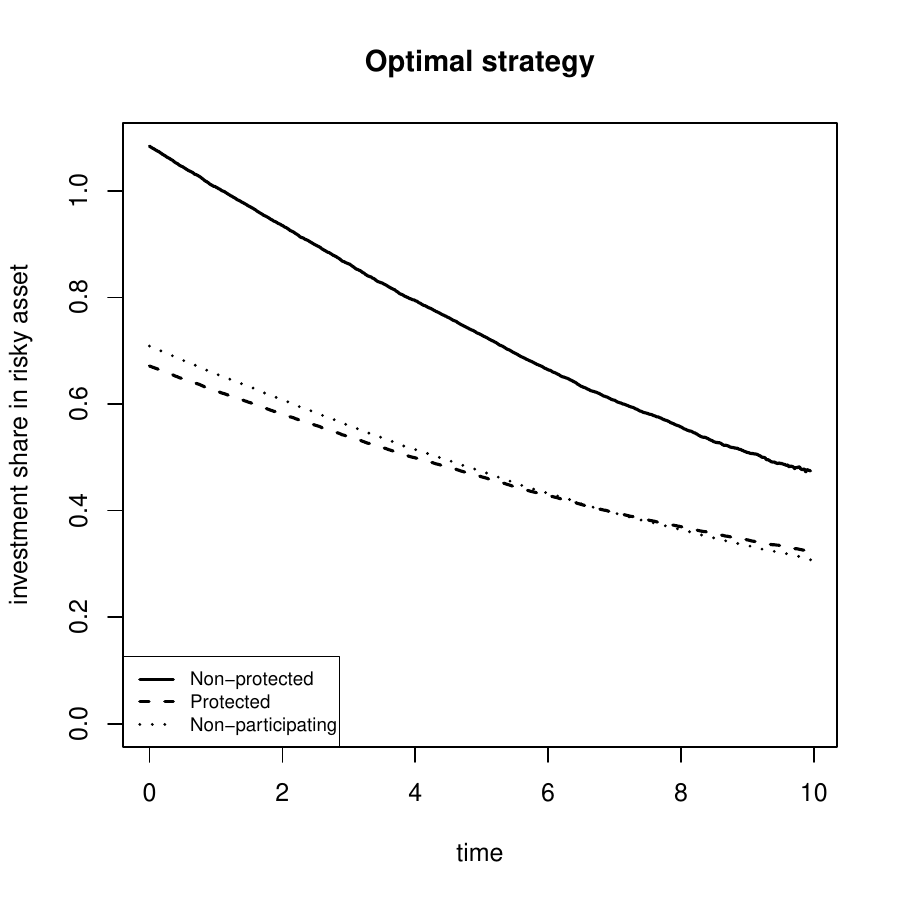}
	\end{minipage}
	\caption{Comparison of the optimal wealth (left) and strategy (right) between a (non-)protected participating life insurance product and a non-participating investment.}
	\label{fig: numerical compare}
\end{figure}	

In Figure \ref{fig: numerical compare Saunders}, we compare the optimal wealth and strategy for the non-protected and the protected life insurance contract when considering mean-variance resp. \ac{eu}-optimization. We take the result for optimizing \ac{eu} from Lin et al. \cite{lin2017optimal} as in Figure \ref{fig: xiplot} with parameter values $\tilde{\gamma}=0.125$ and $\tilde{\lambda}=2$. As in the other figures, the lines show the average of {$10000$} realizations, and we state the total wealth (before splitting it between the policyholders and the insurance company). As in the previous figure, we again choose the initial value such that the terminal wealth is approximately equal for all products, i.e., we take $x_0 = {\scriptsize \begin{cases}
	4 & \text{ non-protected product}, \\
	{4.692} & \text{ protected product}
\end{cases}}$ for the mean-variance optimization and $x_0 = {\scriptsize \begin{cases}
{3.31} & \text{ non-protected product}, \\
{3.59} & \text{ protected product}
\end{cases}}$ for the \ac{eu}-optimization. Both optimizations have in common that the insurance company offering a non-protected product invests riskier than the insurance company offering the protected product due to the reduced downside potential. When comparing the strategies, we observe that the strategies optimizing \ac{eu} become riskier over time (except for the final time points). In contrast, the strategies optimizing mean-variance get less risky over time, as already discussed. Since we chose the two (not comparable) risk aversion factors such that the initial investments in the risky asset are similar, the wealth gain of the \ac{eu}-optimization is higher. The decreasing portion invested into the risky asset is typical for a pre-commitment mean-variance optimization strategy (this can be seen by implementing the optimal mean-variance strategy by Zhou and Li \cite{zhou2000continuous} when using typical ranges for the parameters), which is structurally different from \ac{eu}-optimization strategies.

\begin{figure}[!htb]
	\centering
	\begin{minipage}{0.48\textwidth}
		\includegraphics[trim= 1mm 6mm 10mm 8mm, clip,width=\textwidth]{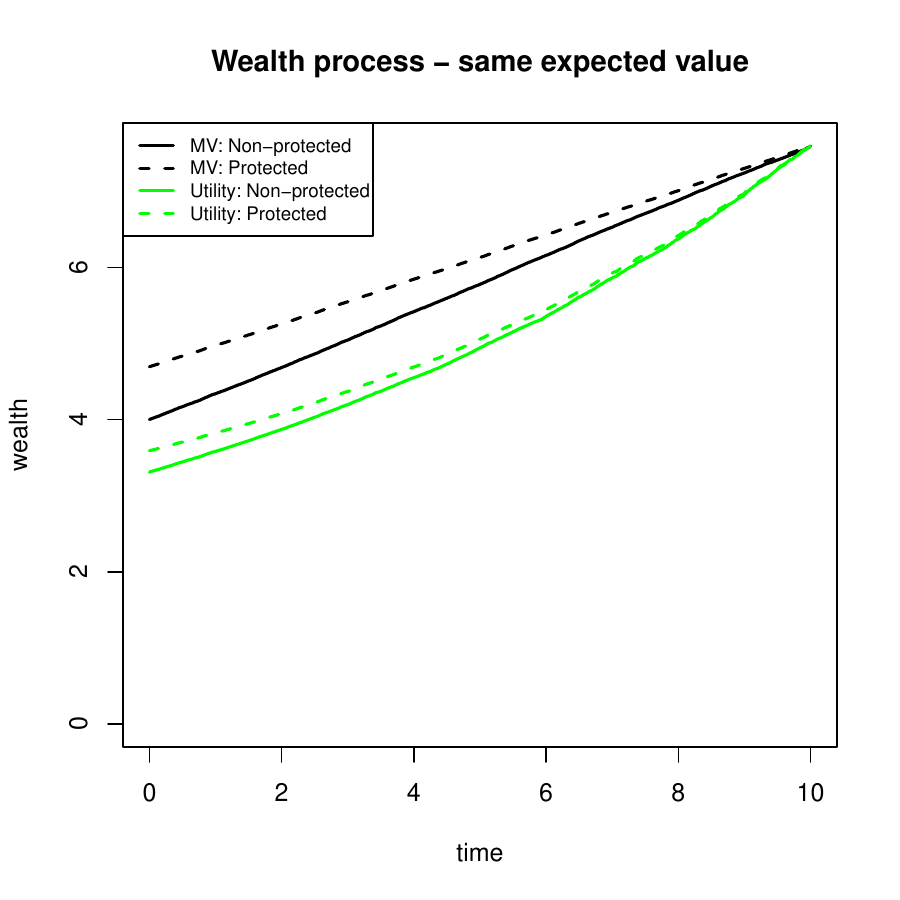}
	\end{minipage}
	\quad
	\begin{minipage}{0.48\textwidth}
		\includegraphics[trim= 1mm 6mm 10mm 8mm, clip,width=\textwidth]{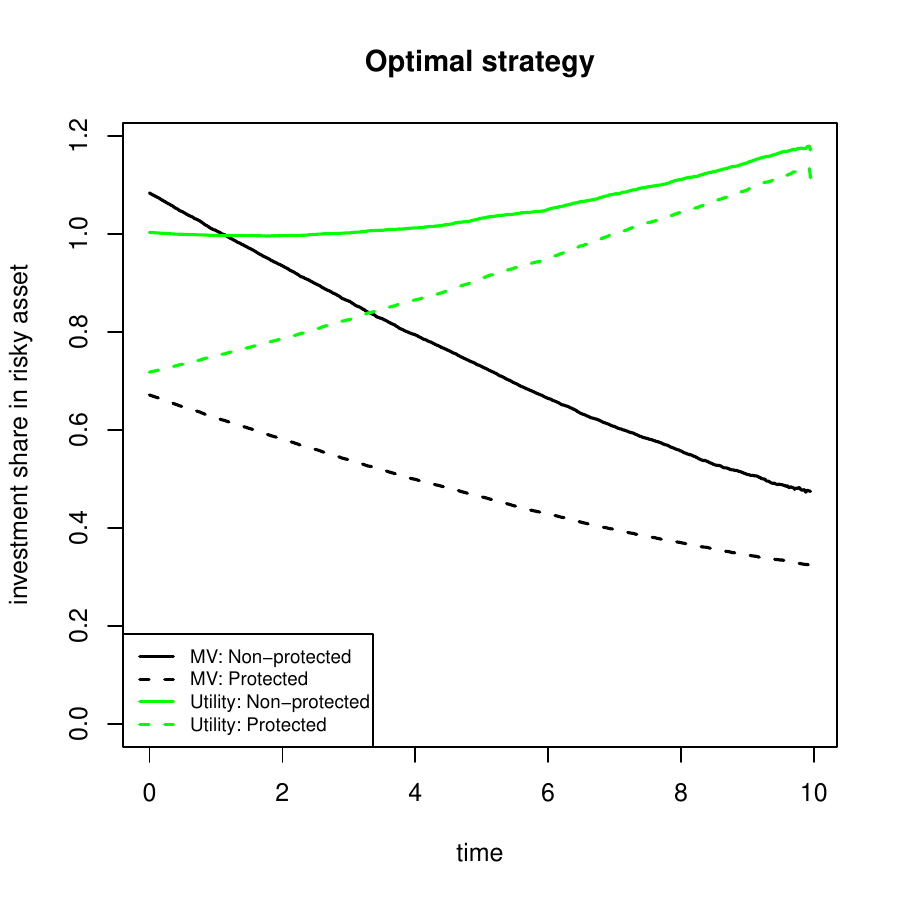}
	\end{minipage}
	\caption{Comparison of the optimal wealth (left) and the optimal strategy (right) between a (non-)protected participating life insurance contract optimizing mean-variance and \ac{eu} {with S-shaped utility from Lin et al. \cite{lin2017optimal}}.}
	\label{fig: numerical compare Saunders}
\end{figure}	

{In Figure \ref{fig: numerical compare Exp}, we show another comparison for the optimal wealth and strategy for the non-protected and the protected life insurance contract when considering mean-variance resp. \ac{eu}-optimization. Here, we take for the \ac{eu}-optimization an exponential utility function $U(x) = 1 - e^{-\hat{\gamma}x}$, where $\hat{\gamma} = 0.5$ denotes the risk aversion of the investor. Note that of course this risk aversion parameter is not directly comparable to the risk aversion parameter from our mean-variance model. (For small risks though $U$ is actually related to mean-variance by a Taylor-expansion.) The optimal wealth and strategy are calculated based on Korn and Trautmann \cite{korn1999optimal}. Due to their assumptions, we change the risk-free interest rate to $r=0$. Again, we show the average of {$10000$} realizations, state the total wealth (before splitting it between the policyholders and the insurance company), and choose the initial value such that the terminal wealth is approximately equal for all products, i.e., we take $x_0 = {\scriptsize \begin{cases}
		4 & \text{ non-protected product}, \\
		4.77 & \text{ protected product}
\end{cases}}$ for the mean-variance optimization and $x_0 = {\scriptsize \begin{cases}
		6.224 & \text{ non-protected product}, \\
		6.262 & \text{ protected product}
\end{cases}}$ for the \ac{eu}-optimization. We obtain again that the insurance company offering a non-protected product invests riskier than the insurance company offering the protected product due to the reduced downside potential, where the difference between the strategies is larger for the mean-variance optimization. In particular, we observe for exponential \ac{eu}-optimization that the optimal strategy is relatively similar for the non-protected and the protected life insurance products. Furthermore, these strategies are decreasing for longer maturities and slightly increasing for short maturities, and are closer to the mean-variance optimization than the ones discussed in Figure \ref{fig: numerical compare Saunders} for S-shaped utility. However, the change in the strategies for exponential \ac{eu}-maximization over time is less than for mean-variance optimization.}

\begin{figure}[!htb]
	\centering
	\begin{minipage}{0.48\textwidth}
		\includegraphics[trim= 1mm 6mm 10mm 8mm, clip,width=\textwidth]{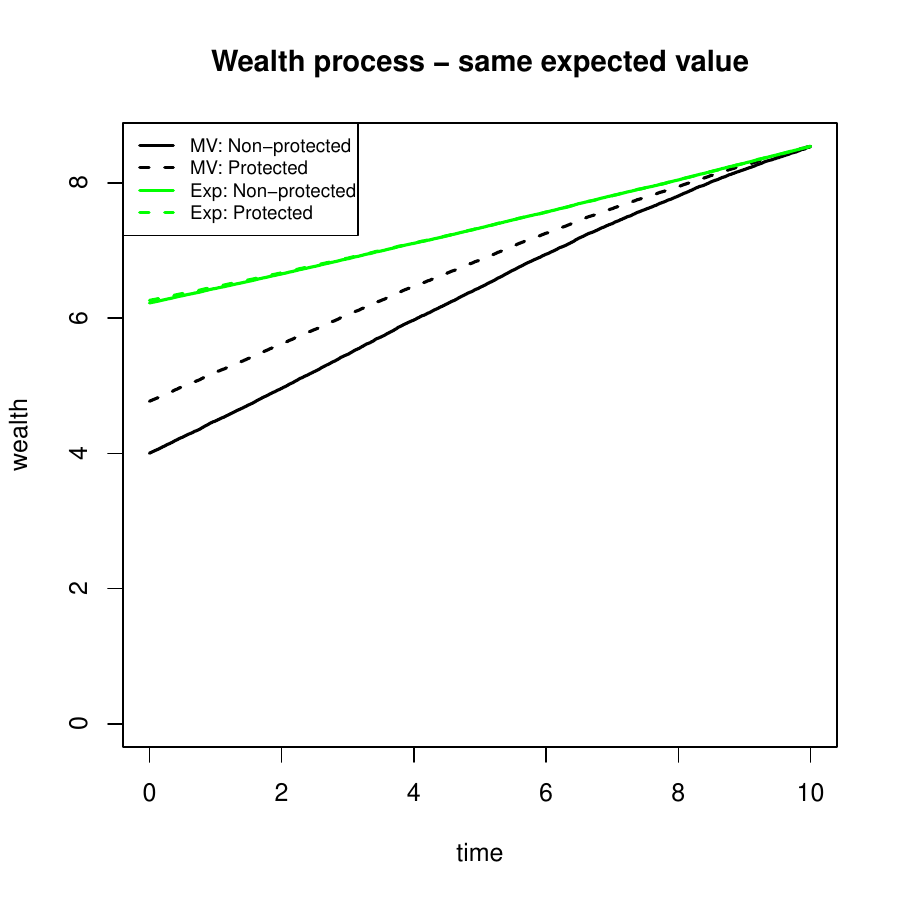}
	\end{minipage}
	\quad
	\begin{minipage}{0.48\textwidth}
		\includegraphics[trim= 1mm 6mm 10mm 8mm, clip,width=\textwidth]{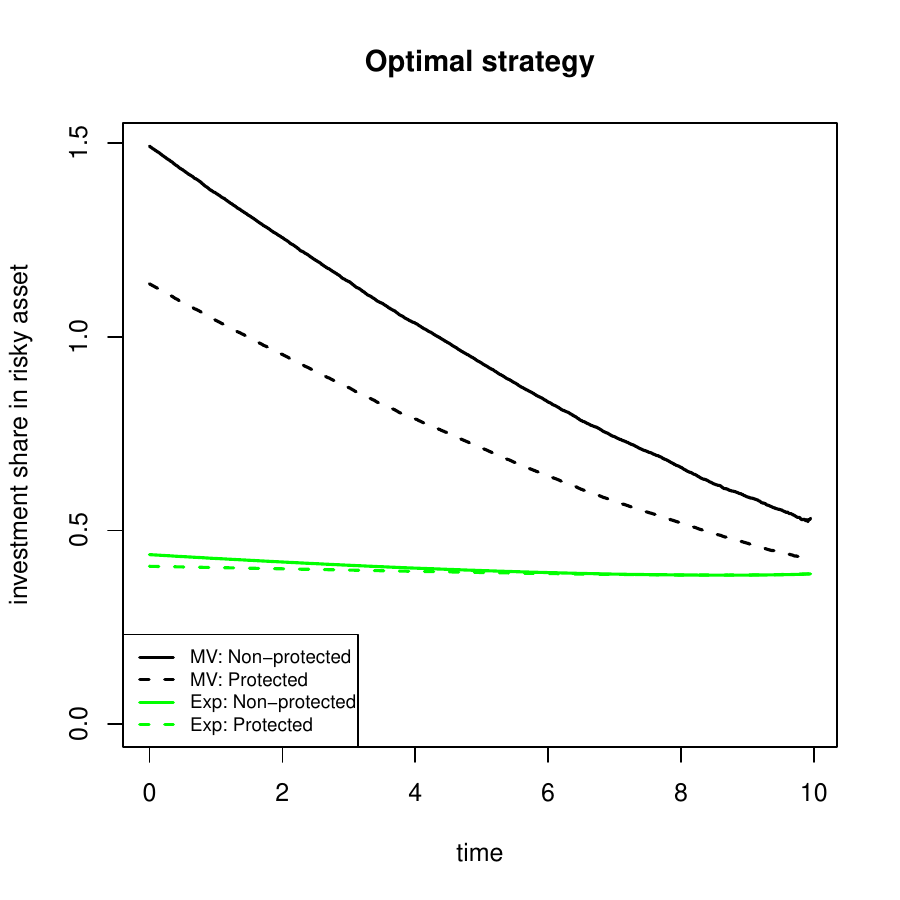}
	\end{minipage}
	\caption{{Comparison of the optimal wealth (left) and the optimal strategy (right) between a (non-)protected participating life insurance contract optimizing mean-variance, or \ac{eu} with an exponential utility function.}}
	\label{fig: numerical compare Exp}
\end{figure}	

\section{Conclusion} \label{conclusion}

In this paper, we derived explicit analytic formulas for general contracts, including participating life insurance contracts, when optimizing mean-variance in the multi-dimensional Black-Scholes market. Moreover, by showing the existence of all arising parameters, we showed the existence of the optimal solution. %We also gave the HJB equation for the value functional if the market is possibly incomplete. 
A numerical analysis shows that the mean-variance optimal strategy compared to the \ac{eu} optimal strategy becomes more conservative the shorter the maturity is and is increased in particular in bad economic states. Future research directions include possibly generalizing this approach to other stock market models.

\section*{Funding}

This research did not receive any specific grant from funding agencies in the public, commercial, or not-for-profit sectors.

\setcounter{section}{0}
\renewcommand{\thesection}{\Alph{section}}
\appendix
\section*{Appendix}

\section{Proofs} \label{proofs}

\begin{myproof}[Proof of Lemma \ref{Alternative optimization}] 
	We show this lemma by contradiction. In particular, we assume that a $\hat{u}$ exists, which is an optimal solution for $J$, but not for $\tilde{J}$. 
	Consequently, there exists a strategy $u$ which is better than $\hat{u}$ for $\tilde{J}$, i.e., 
	\begin{align} \label{proof: F F2 inequality}
		&&\tilde{J}(0,T,u,x_0) - \tilde{J}(0,T,\hat{u},x_0) &> 0 \notag \\
		&\Leftrightarrow& \lambda \left(\EX \left[F(0,T,u,x_0)- F(0,T,\hat{u},x_0)\right]\right) - \gamma \left( \EX \left[ F(0,T,u,x_0)^2 - F(0,T,\hat{u},x_0)^2 \right] \right) &> 0.
	\end{align}
	The next step is to define the function $G(x,y) := y - \gamma x + \gamma y^2$. We observe that $G$ is a convex function, and it holds that
	\begin{align*}
		G\left(\EX \left[ F(0,T,u,x_0)^2 \right], \EX \left[ F(0,T,u,x_0) \right]\right) = \EX \left[ F(0,T,u,x_0) \right] - \gamma \Var \left(F(0,T,u,x_0) \right) = J(0,T,u,x_0).
	\end{align*}
	We note that $\frac{\partial}{\partial x} G(x,y) = -\gamma$ and $\frac{\partial}{\partial y} G(x,y) = 1 + 2\gamma y$. Then, the convexity of $G$ implies
	\begin{align} \label{proof: G inequality}
		&\hspace{-1.5cm}G\left(\EX \left[ F(0,T,u,x_0)^2 \right], \EX \left[ F(0,T,u,x_0) \right] \right) \notag \\
		\geq& \  G\left(\EX \left[ F(0,T,\hat{u},x_0)^2 \right], \EX \left[ F(0,T,\hat{u},x_0) \right]\right) \notag \\
		&+ \begin{pmatrix}
			-\gamma \\
			1 + 2\gamma \EX \left[ F(0,T,\hat{u},x_0) \right]
		\end{pmatrix} \cdot \begin{pmatrix}
			\EX \left[ F(0,T,u,x_0)^2 \right] - \EX \left[ F(0,T,\hat{u},x_0)^2 \right] \\
			\EX \left[ F(0,T,u,x_0) \right] - \EX \left[ F(0,T,\hat{u},x_0) \right]
		\end{pmatrix} \notag \\
		=& \  G\left(\EX \left[ F(0,T,\hat{u},x_0)^2 \right], \EX \left[ F(0,T,\hat{u},x_0) \right]\right) \notag\\
		&-\gamma \left(\EX \left[ F(0,T,u,x_0)^2 \right] - \EX \left[ F(0,T,\hat{u},x_0)^2 \right]\right) \notag\\ &+(1+2\gamma \EX \left[ F(0,T,\hat{u},x_0) \right]) \left(\EX \left[ F(0,T,u,x_0) \right] - \EX \left[ F(0,T,\hat{u},x_0) \right]\right) \notag\\
		>& \  G\left(\EX \left[ F(0,T,\hat{u},x_0)^2 \right], \EX \left[ F(0,T,\hat{u},x_0) \right]\right),
	\end{align}
	where we used \eqref{proof: F F2 inequality} and $\lambda = 1 + 2\gamma \EX \left[ F(0,T,\hat{u},x_0) \right]$. Hence, \eqref{proof: G inequality} implies that $J(0,T,u,x_0) > J(0,T,\hat{u},x_0)$, i.e., $\hat{u}$ is not optimal, which is a contradiction and implies the lemma.
\end{myproof}

\begin{myproof}[Proof of Proposition \ref{prop: * Eigenschaften}]
	We prove this proposition by showing the following two statements from which the claim follows directly: (i) If $\xi_3^* > \alpha \hat{\xi}$, then it holds that $\xi_2^* = \alpha \hat{\xi}$, and (ii) if $\xi_2^* > \tilde{\alpha} \hat{\xi}$, then it holds that $\xi_1^* = \tilde{\alpha} \hat{\xi}$.
	
	We start with the proof of (i). First, if $\hat{\xi}=0$, then it holds that $\xi _{2}^{*}=0$ by definition, i.e., we get the claim. Hence, we assume that $\hat{\xi}>0$. Next, we reformulate the condition $\xi_3^* > \alpha \hat{\xi}$ to get a condition for $\lambda$. When plugging in the formulas as defined in Theorem \ref{optimal wealth}, it holds that $\xi_3^* > \alpha \hat{\xi}$ if and only if 
	\begin{align} \label{eq: condition formula for y}
		\dfrac{2 \gamma \alpha^2 (k_0+k_1) - 2 \gamma \alpha^2\sqrt{k_1^2 + k_1\left(2 k_0 + \frac{\lambda}{\gamma \alpha}\right)}}{y} > - \dfrac{2 \gamma \alpha^2 (k_2-k_1-k_0)}{y}.
	\end{align}
	Hence, it follows that $y<\infty$ and  $k_2 > \sqrt{k_1^2 + 2 k_0 k_1 + \frac{\lambda k_1}{\gamma \alpha}}$ which is again equivalent to $\frac{\lambda k_1}{\gamma \alpha} < k_2^2 -k_1^2 - 2k_0k_1$. Thus, we get that $\xi_3^* > \alpha \hat{\xi}$ holds if and only if $k_1=0$ or $\lambda < \gamma \alpha \left(\frac{k_2^2}{k_1} -k_1-2k_0\right)$ since $k_2>0$. Before showing the implication (i), let us also reformulate the equality $\xi_2^* = \alpha \hat{\xi}$, which we want to show. First, $\xi_2^* = \alpha \hat{\xi}$ if and only if $\tilde{\xi}_2^* \geq \alpha \hat{\xi}$ which is equivalent to:
	\begin{align*}
		-\dfrac{\gamma \alpha^2 (k_2-k_1)^2 -2 \gamma \alpha^2 k_0 (k_2-k_1) + \lambda \alpha k_1}{yk_2} \geq - \dfrac{2 \gamma \alpha^2 (k_2-k_1-k_0)}{y}.
	\end{align*}
	Since $y < \infty$ (otherwise \eqref{eq: condition formula for y} would not hold), this is equivalent to:
	\begin{align*}
		2 \gamma \alpha^2 \left(k_0(k_2-k_1) + k_2(k_2-k_1-k_0)\right) - \gamma \alpha^2 (k_2^2 -2 k_1k_2 + k_1^2) \geq \lambda \alpha k_1.
	\end{align*}
	This again is equivalent to:
	\begin{align*}
		\lambda k_1 \leq \gamma \alpha \left(2k_0k_2-2k_0k_1+2k_2^2-2k_1k_2-2k_0k_2-k_2^2+2k_1k_2-k_1^2\right) = \gamma \alpha \left(k_2^2 - k_1^2 - 2k_0k_1\right).
	\end{align*}
	Now, if $k_1=0$ or if $\lambda < \gamma \alpha \left(\frac{k_2^2}{k_1} -k_1-2k_0\right)$ the inequality is fulfilled since $k_2>0$. Thus, statement (i) follows.
	
	Second, we prove (ii). By defintion, it holds that $\hat{\xi}>0$ in order to get $\xi _{2}^{*} > \tilde{\alpha} \hat{\xi}$. As in the proof of (i), we start by reformulating the first condition. We note that $\xi_2^* > \tilde{\alpha} \hat{\xi}$ if and only if $\tilde{\xi}_2^* > \tilde{\alpha} \hat{\xi}$ which is equivalent to:
	\begin{align*}
		\dfrac{\alpha \lambda}{y}-\dfrac{\gamma \alpha^2 (k_2-k_1)^2 -2 \gamma \alpha^2 k_0 (k_2-k_1) + \lambda \alpha k_1}{yk_2} > \dfrac{\tilde{\alpha} \lambda}{y} - \dfrac{2 \gamma \alpha \tilde{\alpha} (k_2-k_1-k_0)}{y}.
	\end{align*}
	From this inequality, it follows that $y < \infty$ and 
	\begin{align*}
		\lambda (\alpha-\tilde{\alpha}) k_2 - \lambda \alpha k_1 + \gamma \alpha^2 \left( -k_2^2 + 2k_1k_2 - k_1^2 + 2k_0k_2-2k_0k_1 \right) > 2 \gamma \alpha \tilde{\alpha} \left(-k_2^2+k_1k_2+k_0k_2\right).
	\end{align*}
	Since $\alpha-\tilde{\alpha} = \alpha_2$, this is equivalent to:
	\begin{align}
		\lambda (\alpha_2k_2-\alpha k_1) > - \gamma \alpha \left( k_2^2 (\tilde{\alpha} - \alpha_2) + 2k_1k_2 \alpha_2 + 2k_0k_2 \alpha_2 - \alpha k_1^2 - 2\alpha k_0k_1 \right). \label{eq: reformulation 3.4 ii}
	\end{align}
	As before, we also reformulate the equality $\xi_1^* = \tilde{\alpha} \hat{\xi}$. This equality is indeed true if $\tilde{\xi}_1^* \geq \tilde{\alpha} \hat{\xi}$. If the term under the square root in the definition of $\tilde{\xi}_1^*$ (see Theorem \ref{optimal wealth}) is $0$, then $\tilde{\xi}_1^* \geq \tilde{\alpha} \hat{\xi}$ is fulfilled since $\frac{2 \gamma \tilde{\alpha}^2 k_2}{y} \geq 0$. Now, if this term is positive, $\tilde{\xi}_1^* \geq \tilde{\alpha} \hat{\xi}$ is equivalent to:
	\begin{align*}
		\dfrac{2 \gamma \tilde{\alpha}^2 k_2}{y} \geq \dfrac{2 \gamma \tilde{\alpha}}{y} \sqrt{ (\alpha(k_0+k_1)-\alpha_2k_2)^2 -\alpha^2k_0^2+ \frac{\lambda}{\gamma} (\alpha k_1-\alpha_2k_2)}.
	\end{align*}
	Then, since $y < \infty$, this is equivalent to:
	\begin{align*}
		\tilde{\alpha}^2 k_2^2 \geq \alpha^2k_0^2 + 2\alpha^2k_0k_1 + \alpha^2k_1^2-2\alpha\alpha_2 k_0k_2 - 2\alpha\alpha_2k_1k_2 + \alpha_2^2k_2^2 - \alpha^2k_0^2 + \frac{\lambda}{\gamma} (\alpha k_1-\alpha_2k_2),
	\end{align*}
	which is again equivalent to:
	\begin{align*}
		\lambda (\alpha_2k_2-\alpha k_1) \geq -\gamma \alpha \left( k_2^2 \left(\frac{\tilde{\alpha}^2-\alpha_2^2}{\alpha}\right)  + 2 k_1k_2 \alpha_2 + 2 k_0k_2 \alpha_2 - \alpha k_1^2 - 2 \alpha k_0k_1 \right).
	\end{align*}
	Now, we get that $\frac{\tilde{\alpha}^2-\alpha_2^2}{\alpha} = \frac{(\tilde{\alpha}-\alpha_2)(\tilde{\alpha}+\alpha_2)}{\alpha} = \tilde{\alpha}-\alpha_2$ since $\tilde{\alpha}+\alpha_2=\alpha$. Hence, the inequality is fulfilled if $\lambda (\alpha_2k_2-\alpha k_1) > - \gamma \alpha \left( k_2^2 (\tilde{\alpha} - \alpha_2) + 2k_1k_2 \alpha_2 + 2k_0k_2 \alpha_2 - \alpha k_1^2 - 2\alpha k_0k_1 \right)$ which by equation \eqref{eq: reformulation 3.4 ii} implies statement (ii) and finishes the proof.
\end{myproof}

\begin{myproof}[Proof of Proposition \ref{prop: xhat continuous}]
	Using the two properties proven in the proof of Proposition \ref{prop: * Eigenschaften}, we only have to check the boundary values $\xi_1^*=\tilde{\alpha} \hat{\xi}$ and $\xi_2^* = \alpha \hat{\xi}$ for continuity. However, this follows immediately by plugging in the values. Now, if $k_1>0$, we see that $\hat{X}_T > 0$ if $\xi^* \neq \xi_3^*$. If $\xi^* = \xi_3^*$, by the definition of $\xi^{*}$ in Theorem \ref{optimal wealth}, it actually holds that $\xi_{3}^{*}>\alpha \hat{\xi}$. We note that
	\begin{align*}
		\hat{X}_T (\xi_3^*) &= k_0 + k_1 + \dfrac{\lambda \alpha - \lambda \alpha - 2\gamma \alpha^2 k_0 +  2\gamma \alpha^2 \left(\sqrt{k_1^2 + k_1\left(2 k_0 + \frac{\lambda}{\gamma \alpha}\right)} - k_1 \right)}{2 \gamma \alpha^2} \\
		&= k_0 + k_1 - k_0 -k_1 + \sqrt{k_1^2 + k_1\left(2 k_0 + \frac{\lambda}{\gamma \alpha}\right)} = \sqrt{k_1 \left( k_1 + 2 k_0+ \frac{\lambda }{\gamma \alpha}\right)}.
	\end{align*}
	Hence, $\hat{X}_T (\xi_3^*) > 0$ if $k_1 > 0$ and $\lambda > -2 \gamma \alpha k_0 - \gamma \alpha k_1$. From the proof of Proposition \ref{y and lambda existence}, we see that $\lambda > C$ with $C \geq -2 \gamma \alpha k_0$. Thus, $\hat{X}_T (\xi_3^*) > 0$ if $k_1>0$.
\end{myproof}

\begin{myproof}[Proof of Proposition \ref{prop: equivalent-like problem}]
	{For this proof, we show that the alternative formulation of this proposition admits the same Lagrange function as our main problem given in \eqref{eq: L function}. Therefore, we first reformulate the problem as a maximization problem by replacing $\min_{u \in \mathcal{U}} \EX [\tilde{L}(X_T)]$ with $\max_{u \in \mathcal{U}} \EX [-\tilde{L}(X_T)]$. Then, we get for the Lagrange function with $X \geq 0$:
	\begin{align*}
		L(X,y,\lambda) &= -\tilde{L}(X) - \lambda [\alpha k_0-\alpha(X-k_1) \1_{k_1 \leq X < k_2} - \tilde{\alpha}(X-k_2+\tfrac{\alpha}{\tilde{\alpha}}(k_2-k_1)) \1_{X \geq k_2}] -y \xi X \\
		&= \begin{cases}
				-\gamma \alpha^2 k_0^2 - \lambda \alpha k_0 -y \xi X, & \text{if } X \in [0,k_1), \\
				-\gamma \alpha^2 (X-k_1-k_0)^2 - \lambda \alpha k_0 + \lambda \alpha(X-k_1) -y \xi X, & \text{if } X \in [k_1,k_2), \\
				-\gamma \tilde{\alpha}^2 (X-k_2+\tfrac{\alpha}{\tilde{\alpha}}(k_2-k_1-k_0))^2 - \lambda \alpha k_0 & \\
				\hspace{3cm}+ \lambda \tilde{\alpha}(X-k_2+\tfrac{\alpha}{\tilde{\alpha}}(k_2-k_1)) -y \xi X, & \text{if } X \in [k_2,\infty),
		\end{cases} \\
		&= \begin{cases}
			-\gamma \alpha^2 k_0^2 - \lambda \alpha k_0 -y \xi X, & \text{if } X \in [0,k_1), \\
			-\gamma \alpha^2 (X-k_1-k_0)^2 + \lambda \alpha(X-k_1-k_0) -y \xi X, & \text{if } X \in [k_1,k_2), \\
			-\gamma (\tilde{\alpha}(X-k_2)+\alpha(k_2-k_1-k_0))^2 & \\
			\hspace{3cm}+ \lambda (\tilde{\alpha}(X-k_2)+\alpha(k_2-k_1-k_0)) -y \xi X, & \text{if } X \in [k_2,\infty).
		\end{cases}
	\end{align*}
	By comparing the terms with \eqref{eq: L function}, we observe that the Lagrange functions are identical which gives us the proof.}
\end{myproof}

\begin{myproof}[Proof of Theorem \ref{optimal strategy}]
	From \eqref{X definition} and \eqref{xi definition}, we conclude using It{\^o}'s formula that $\xi_t \hat{X}_t$ is a martingale. Hence, we get using Theorem \ref{optimal wealth}:
	\begin{align} \label{calculation formula xhat_t}
		\hat{X}_t =&  \EX \left[ \left. \dfrac{\xi_T}{\xi_t} \hat{X}_T \right| \mathcal{F}_t\right] \notag \\
		=& \left(k_2 + \dfrac{\lambda}{2 \gamma \tilde{\alpha}} - \dfrac{\alpha}{\tilde{\alpha}}(k_2-k_1-k_0) \right) \EX \left[ \left. \dfrac{\xi_T}{\xi_t} \1_{\xi_T \leq \xi_1^*} \right| \mathcal{F}_t\right] - \dfrac{y}{2\gamma \tilde{\alpha}^2} \EX \left[ \left. \dfrac{\xi_T^2}{\xi_t} \1_{\xi_T \leq \xi_1^*} \right| \mathcal{F}_t\right] \notag \\
		&+ k_2 \EX \left[ \left. \dfrac{\xi_T}{\xi_t} \1_{\tilde{\alpha} \hat{\xi} < \xi_T \leq \xi_2^*} \right| \mathcal{F}_t\right] + \left(k_0 + k_1 + \dfrac{\lambda}{2 \gamma \alpha} \right) \EX \left[ \left. \dfrac{\xi_T}{\xi_t} \1_{\alpha \hat{\xi} < \xi_T \leq \xi_3^*} \right| \mathcal{F}_t\right] \notag \\
		&-\dfrac{y}{2\gamma \alpha^2} \EX \left[ \left. \dfrac{\xi_T^2}{\xi_t} \1_{\alpha \hat{\xi} < \xi_T \leq \xi_3^*} \right| \mathcal{F}_t\right]. 
	\end{align}
	Now, the claim for $\hat{X}_t$ follows from the formula for the conditional expectation of log-normal distributions (see \eqref{conditional expectation log normal}) using that it holds conditionally on $\mathcal{F}_t$:
	\begin{align*}
		\frac{\xi_T}{\xi_t} &\sim \mathcal{LN} \left(-\int_t^T r_s+\frac{\norm{\kappa_s}^2}{2} \diff s, \int_t^T \norm{\kappa_s}^2 \diff s\right), \\
		\frac{\xi_T^2}{\xi_t} &\sim \mathcal{LN} \left(\ln \xi_t -2 \int_t^T r_s+\frac{\norm{\kappa_s}^2}{2} \diff s, 4 \int_t^T\norm{\kappa_s}^2 \diff s\right),
	\end{align*}
	where $\mathcal{LN}$ denotes a log-normal distribution.\\
	The next step is to calculate the volatility process of $\hat{X}_t$, i.e., the term before the $\diff W_t$ in its SDE, denoted by $\sigma_{\hat{X}_t}$ using \eqref{xi definition}. There, we get $\sigma_{\hat{X}_t} = (-\kappa_t^T \xi_t) \dfrac{\partial \hat{X}_t}{\partial \xi_t} = \kappa_t^T v_t$,
	where we used that $\dfrac{\partial d_1(x)}{\partial \xi_t} = \dfrac{\partial d_2(x)}{\partial \xi_t} = \dfrac{1}{-\xi_t \sqrt{\int_t^T \norm{\kappa_s}^2 \diff s}}$. Now, we get the optimal strategy $\hat{u}_t$ by comparing the volatility processes from $\hat{X}_t$ which is $\sigma_{\hat{X}_t} = \hat{X}_t \hat{u}^T \sigma_t$. Hence, the claim follows.
\end{myproof}

\subsection{Proposition \ref{y lambda exist}}

Since the equation $\EX [\xi_T \hat{X}_T (y)]  = \xi_0 x_0$ depends on $\lambda$ and $\lambda = 1 + 2\gamma \EX \left[ F(0,T,\hat{u},x_0) \right]$ depends on $y$, we have to solve these two equations together. We show that these variables always exist in the following Proposition \ref{y lambda exist}. 

\begin{proposition} \label{y lambda exist}
	There always exists a solution for $\lambda$ and $y$ and an equation system that can be numerically solved to determine them.
\end{proposition}
	
We split this proof into two lemmas. In the first Lemma \ref{y and lambda formulas}, we give the equation system to solve for the two parameters $y$ and $\lambda$, where we added $y$ and $\lambda$ as a superscript to help the reader follow the interdependence. In the second Lemma \ref{y and lambda existence}, we show that these parameters exist: 
	
\begin{lemma} \label{y and lambda formulas}
		The variables $y$ and $\lambda$ are the solution of the following equation system:
		\begin{align*}
			x_0 =& \left(k_2 + \dfrac{\lambda}{2 \gamma \tilde{\alpha}} - \dfrac{\alpha}{\tilde{\alpha}}(k_2-k_1-k_0) \right) e^{-\int_0^T r_s \diff s} \Phi\left(d_1 \left(\xi_1^{*,y,\lambda},0\right)\right) \\
			&- \dfrac{y}{2\gamma \tilde{\alpha}^2} e^{\int_0^T -2r_s+\norm{\kappa_s}^2 \diff s} \Phi \left( d_2 \left(\xi_1^{*,y,\lambda},0\right) \right)  \\
			&+ k_2 e^{-\int_0^T r_s \diff s} \left( \Phi\left(d_1 \left(\xi_2^{*,y,\lambda},0\right)\right) - \Phi\left(d_1 \left(\tilde{\alpha} \hat{\xi}^{y,\lambda},0\right)\right) \right) \\
			&+ \left(k_0+k_1 + \dfrac{\lambda}{2 \gamma \alpha} \right) e^{-\int_0^T r_s \diff s} \left( \Phi\left(d_1 \left(\xi_3^{*,y,\lambda},0\right)\right) - \Phi\left(d_1 \left(\alpha \hat{\xi}^{y,\lambda},0\right)\right) \right) \\
			&- \dfrac{y}{2\gamma \alpha^2} e^{\int_0^T -2r_s+\norm{\kappa_s}^2 \diff s} \left(\Phi \left( d_2 \left(\xi_3^{*,y,\lambda},0\right) \right) - \Phi \left( d_2 \left(\alpha \hat{\xi}^{y,\lambda},0\right) \right) \right) ,\\
			\lambda =& \ 1-2\gamma\alpha k_0 + 2\gamma\alpha k_0 \left( \Phi\left(d_0 \left(\xi_3^{*,y,\lambda},0\right)\right) - \Phi\left(d_0 \left(\alpha \hat{\xi}^{y,\lambda},0\right)\right) + \Phi\left(d_0 \left(\xi_2^{*,y,\lambda},0\right)\right) \right. \\
			&\hspace{102pt}- \left. \Phi\left(d_0 \left(\tilde{\alpha} \hat{\xi}^{y,\lambda},0\right)\right) + \Phi\left(d_0 \left(\xi_1^{*,y,\lambda},0\right)\right) \right) \\
			&+ \lambda \left( \Phi\left(d_0 \left(\xi_3^{*,y,\lambda},0\right)\right) - \Phi\left(d_0 \left(\alpha \hat{\xi}^{y,\lambda},0\right)\right) + \Phi\left(d_0 \left(\xi_1^{*,y,\lambda},0\right)\right) \right) \\
			&- y \left( \dfrac{1}{\alpha}\Phi\left(d_1 \left(\xi_3^{*,y,\lambda},0\right)\right) - \dfrac{1}{\alpha}\Phi\left(d_1 \left(\alpha \hat{\xi}^{y,\lambda},0\right)\right) + \dfrac{1}{\tilde{\alpha}}\Phi\left(d_1 \left(\xi_1^{*,y,\lambda},0\right)\right) \right) \\
			&+2 \gamma \alpha (k_2-k_1-k_0) \left( \Phi\left(d_0 \left(\xi_2^{*,y,\lambda},0\right)\right) - \Phi\left(d_0 \left(\tilde{\alpha} \hat{\xi}^{y,\lambda},0\right)\right) \right)
		\end{align*}
		with
		\begin{align*}
			d_0 (x,t) =& \dfrac{\ln x - \ln \xi_t + \int_t^T r_s + \frac{\norm{\kappa_s}^2}{2} \diff s}{\sqrt{\int_t^T \norm{\kappa_s}^2 \diff s}}.
		\end{align*}
		The values $\hat{\xi}^{y,\lambda}$, $\xi_1^{*,y,\lambda}$, $\xi_2^{*,y,\lambda}$, and $\xi_3^{*,y,\lambda}$ are defined as in Theorem \ref{optimal wealth} and the functions $d_1$ and $d_2$ are defined as in Theorem \ref{optimal strategy}.
\end{lemma}

\begin{proof}
	The first formula follows from the definition $\EX [\xi_T \hat{X}_T (y)]  = \xi_0 x_0 = x_0$ and \eqref{calculation formula xhat_t} for $t=0$. The second formula follows from the definition $\lambda = 1 + 2\gamma \EX \left[ F(0,T,\hat{u},x_0) \right]$ and the following properties for arbitrary $0\leq a <b<\infty$:
	\begin{align}
		\hat{X}_T =& \left(k_2 + \dfrac{\lambda \tilde{\alpha} - y \xi_T}{2 \gamma \tilde{\alpha}^2} - \dfrac{\alpha}{\tilde{\alpha}} (k_2-k_1-k_0)\right) \1_{\xi_T \in (0,\xi_1^*]} + k_2 \1_{\xi_T \in [\tilde{\alpha}\hat{\xi},\xi_2^*]} \notag \\
		&+ \left(k_0 + k_1 + \dfrac{\lambda \alpha - y \xi_T}{2 \gamma \alpha^2}\right) \1_{\xi_T \in [\alpha\hat{\xi},\xi_3^*]}, \notag \\
		\ln \xi_T \sim&\ \mathcal{N} \left(-\int_0^T r_s+\frac{\norm{\kappa_s}^2}{2} \diff s, \int_0^T \norm{\kappa_s}^2 \diff s\right), \notag\\
		\EX[\xi_T \1_{\xi_T \in [a,b]}] =& \ \Phi \left( d_1(b,0)\right) - \Phi \left( d_1(a,0)\right), \label{equation: EX1 xi} \\
		\EX[\1_{\xi_T \in [a,b]}] =& \ \PP(\xi_T \in [a,b]) = \Phi \left( d_0(b,0)\right) - \Phi \left( d_0(a,0)\right). \label{equation: EX1} 
	\end{align}
	Moreover, we conclude from \eqref{eq: optimal terminal wealth} (combined with Proposition \ref{prop: xhat continuous}) that $\1_{\hat{X}_T \geq k_2} = \1_{\xi_T \in (0,\xi_1^*]} + \1_{\xi_T \in (\tilde{\alpha} \hat{\xi},\xi_2^*]}$, $\1_{\hat{X}_T \in [k_1,k_2)} = \1_{\xi_T \in (\alpha \hat{\xi} ,\xi_3^*]}$, and $\1_{\hat{X}_T < k_1} = \1_{\hat{X}_T =0} = \1_{\xi_T \not\in (0,\xi_1^*] \cup (\tilde{\alpha} \hat{\xi},\xi_2^*] \cup (\alpha \hat{\xi} ,\xi_3^*]}$.\\
	Indeed, we get:
	\begin{align*}
		\EX \left[ F(0,T,\hat{u},x_0) \right] =& \EX \left[ -\alpha k_0 \1_{\hat{X}_T < k_1} + \alpha (\hat{X}_T-k_1-k_0) \1_{\hat{X}_T \in [k_1,k_2)} + \tilde{\alpha} (\hat{X}_T-k_2) \1_{\hat{X}_T \geq k_2} \right] \\
		&+ \EX \left[ \alpha (k_2-k_1-k_0) \1_{\hat{X}_T \geq k_2} \right] \\
		=& -\alpha k_0 \PP(\hat{X}_T < k_1) + \alpha \EX[(\hat{X}_T-k_1-k_0) \1_{\xi_T \in (\alpha \hat{\xi} ,\xi_3^*]}] \\
		&+ \EX \left[ \left(\tilde{\alpha} (\hat{X}_T-k_2) + \alpha (k_2-k_1-k_0)\right) \left(\1_{\xi_T \in (0,\xi_1^*]} + \1_{\xi_T \in (\tilde{\alpha} \hat{\xi},\xi_2^*]} \right)\right] \\
		=& -\alpha k_0 + \alpha k_0 \left( \PP(\xi_T \in (0,\xi_1^*] ) + \PP(\xi_T \in (\tilde{\alpha} \hat{\xi},\xi_2^*]) + \PP(\xi_T \in (\alpha \hat{\xi} ,\xi_3^*])\right) \\
		&+ \dfrac{\lambda}{2 \gamma} \PP (\xi_T \in (\alpha \hat{\xi} ,\xi_3^*]) - \dfrac{y}{2 \gamma \alpha} \EX \left[ \xi_T \1_{\xi_T \in (\alpha \hat{\xi} ,\xi_3^*]}\right] + \dfrac{\lambda}{2 \gamma} \PP (\xi_T \in (0,\xi_1^*]) \\
		&- \dfrac{y}{2 \gamma \tilde{\alpha}} \EX \left[ \xi_T \1_{\xi_T \in (0,\xi_1^*]}\right] + \alpha (k_2-k_1-k_0) \PP(\xi_T \in (\tilde{\alpha} \hat{\xi},\xi_2^*]),
	\end{align*}
	where we used the above-discussed results in the second and \eqref{eq: optimal terminal wealth} in the third equation.
	The claim follows with \eqref{equation: EX1 xi} and \eqref{equation: EX1}.
\end{proof}

\begin{lemma} \label{y and lambda existence}
	The equation system from Lemma \ref{y and lambda formulas} admits a solution.
\end{lemma}

\begin{proof}
	For this existence proof, we make the following two definitions by writing the two equations as functions depending on $y$ and $\lambda$:
	\begin{align}
		f_1(y,\lambda) :=& -x_0 + \left(k_2 + \dfrac{\lambda}{2 \gamma \tilde{\alpha}} - \dfrac{\alpha}{\tilde{\alpha}}(k_2-k_1-k_0) \right) e^{-\int_0^T r_s \diff s} \Phi\left(d_1 \left(\xi_1^{*},0\right)\right) \notag \\
		&- \dfrac{y}{2\gamma \tilde{\alpha}^2}  e^{\int_0^T -2r_s+\norm{\kappa_s}^2 \diff s} \Phi \left( d_2 \left(\xi_1^{*},0\right) \right)  + k_2 e^{-\int_0^T r_s \diff s} \left( \Phi\left(d_1 \left(\xi_2^{*},0\right)\right) - \Phi\left(d_1 \left(\tilde{\alpha} \hat{\xi},0\right)\right) \right) \notag \\
		&+ \left(k_0+k_1 + \dfrac{\lambda}{2 \gamma \alpha} \right) e^{-\int_0^T r_s \diff s} \left( \Phi\left(d_1 \left(\xi_3^{*},0\right)\right) - \Phi\left(d_1 \left(\alpha \hat{\xi},0\right)\right) \right) \notag \\
		&- \dfrac{y}{2\gamma \alpha^2} e^{\int_0^T -2r_s+\norm{\kappa_s}^2 \diff s} \left(\Phi \left( d_2 \left(\xi_3^{*},0\right) \right) - \Phi \left( d_2 \left(\alpha \hat{\xi},0\right) \right) \right), \label{eq: definition f1} \\		
		f_2(y,\lambda) :=& \ 1-2\gamma\alpha k_0 + 2\gamma\alpha k_0 \left( \Phi\left(d_0 \left(\xi_3^{*},0\right)\right) - \Phi\left(d_0 \left(\alpha \hat{\xi},0\right)\right) + \Phi\left(d_0 \left(\xi_2^{*},0\right)\right) \right. \notag \\
		&\hspace{102pt}- \left. \Phi\left(d_0 \left(\tilde{\alpha} \hat{\xi},0\right)\right) + \Phi\left(d_0 \left(\xi_1^{*},0\right)\right) \right) \notag \\
		&+ \lambda \left( -1 + \Phi\left(d_0 \left(\xi_3^{*},0\right)\right) - \Phi\left(d_0 \left(\alpha \hat{\xi},0\right)\right) + \Phi\left(d_0 \left(\xi_1^{*},0\right)\right) \right) \notag \\
		&- y \left( \dfrac{1}{\alpha}\Phi\left(d_1 \left(\xi_3^{*},0\right)\right) - \dfrac{1}{\alpha}\Phi\left(d_1 \left(\alpha \hat{\xi},0\right)\right) + \dfrac{1}{\tilde{\alpha}}\Phi\left(d_1 \left(\xi_1^{*},0\right)\right) \right) \notag \\
		&+2 \gamma \alpha (k_2-k_1-k_0) \left( \Phi\left(d_0 \left(\xi_2^{*},0\right)\right) - \Phi\left(d_0 \left(\tilde{\alpha} \hat{\xi},0\right)\right) \right). \notag
	\end{align}
	Note that we suppress again the dependence of $\hat{\xi}$ and $\xi_i^*$, $i \in \{1,2,3\}$, on $y$ and $\lambda$.
	By definition, the equation system from Lemma \ref{y and lambda formulas} is solved if there exist $\lambda^*$, $y^*$ such that $f_1(y^*,\lambda^*) = 0 = f_2(y^*,\lambda^*)$. 
	
	To show this, we define 
	\begin{align*}
		C := \begin{cases}
			-2\gamma \alpha k_0 & x_0 e^{\int_0^T r_s \diff s} \leq k_1, \\
			2 \gamma \alpha (x_0 e^{\int_0^T r_s \diff s}-k_1-k_0) & x_0 e^{\int_0^T r_s \diff s} \in (k_1,k_2], \\
			2 \gamma \left(\tilde{\alpha} x_0 e^{\int_0^T r_s \diff s} + \alpha_2 k_2 - \alpha (k_0+k_1)\right) & x_0 e^{\int_0^T r_s \diff s} > k_2,
		\end{cases}
	\end{align*}
	and $h:(C,\infty) \to \Real_{\geq 0}$ an arbitrary continuous function with $\lim_{x \to C} h(x)=0$ if $x_0 e^{\int_0^T r_s \diff s} > k_1$ resp. $\liminf_{x \to C} h(x) \geq 0$ if $x_0 e^{\int_0^T r_s \diff s} \leq k_1$, and $\lim_{x \to \infty} h(x) = \infty$. Then, we make the following four statements, which are shown at the end of this proof:
	\begin{align}
		\lim_{y \to 0} f_1(y,\lambda) &> 0 \text{ for all $\lambda > C$},\label{eq: f1 0} \\
		\lim_{y \to \infty} f_1(y,\lambda) &\leq -x_0 \text{ for all $\lambda > C$},\label{eq: f1 infty} \\
		\liminf_{\lambda \searrow C} f_2(h(\lambda),\lambda) &\geq 1, \label{eq: f2 C} \\
		\lim_{\lambda \to \infty} f_2(h(\lambda),\lambda) &= -\infty. \label{eq: f2 infty} 
	\end{align}
	Due to the continuity of $y \mapsto f_1(y,\cdot)$ and $\lambda \mapsto f_2(h(\lambda),\lambda)$, there exists for each $\lambda > C$ a $y^*_{\lambda} \in [0,\infty)$ and for all functions $h$ a $\lambda^*_h \in (C,\infty)$ such that $f_1(y^*_{\lambda},\lambda)=0$ and $f_2(h(\lambda^*_h),\lambda^*_h)=0$. If there exists more than one solution, then we take $y_{\lambda}^*$ as the smallest one of those (well-defined due to $y$ being bounded from below and the continuity of $f_1$ and $f_2$). Moreover, we have the following four statements, which are also shown at the end of this proof:
	\begin{align}
		\lambda \mapsto y^*_{\lambda} &\text{ is continuous in $\lambda$ on $(C,\infty)$} \label{eq: y star continuous}, \\
		\lim_{\lambda \searrow C} y^*_{\lambda} &= 0 \label{eq: y star lambda C} \text{ if $x_0 e^{\int_0^T r_s \diff s} > k_1$}, \\
		\liminf_{\lambda \searrow C} y^*_{\lambda} &\geq 0 \label{eq: y star lambda C alt} \text{ if $x_0 e^{\int_0^T r_s \diff s} \leq k_1$}, \\
		\lim_{\lambda \to \infty} y^*_{\lambda} &= \infty \label{eq: y star lambda infty}.
	\end{align}
	Now, if we consider $y^*_{\lambda}$ as a function of $\lambda$, it fulfills the assumptions of $h$. Hence, there exists a $\lambda^*$ such that $f_2(y^*_{\lambda^*},\lambda^*)=0$. On the other hand, by the definition of $\lambda \mapsto y_\lambda$, we have that $f_1(y_{\lambda^*}^*,\lambda^*)=0$, which implies the claim.
	
	To finalize the proof, we show the remaining eight statements afterwards. But first, we note that it holds by definition (see Theorem \ref{optimal wealth}):
	\begin{align} \label{order of xis}
		0 \leq \xi_1^* \leq \tilde{\alpha} \hat{\xi} \leq \xi_2^* \leq \alpha \hat{\xi} \leq \xi_3^* \leq \bar{\xi}.
	\end{align}
	\begin{myproof}[Proof of \eqref{eq: f1 0}]
		First of all, we claim that 
		\begin{align}
			\xi_3^* > 0 \text{ if $\lambda > C$}, \label{eq: xi3* bigger 0} \\
			\xi_1^* > 0 \text{ if and only if $\lambda > 2 \gamma \alpha (k_2-k_1-k_0)$}, \label{eq: xi1* bigger 0}
		\end{align}
		which we will show directly after the proof of \eqref{eq: f1 0}. So for now assume that \eqref{eq: xi3* bigger 0} and \eqref{eq: xi1* bigger 0} hold, and let $\lambda>C$.\\
		By \eqref{eq: xi3* bigger 0}, $\xi_3^* \xrightarrow{y \to 0} +\infty$ if $\lambda>C$. Moreover, it follows directly from the definition that
		\begin{align} \label{eq: xihat 0 infty}
			\hat{\xi} \xrightarrow{y \to 0} \begin{cases}
				0, & \text{if } \lambda \leq 2 \gamma \alpha (k_2-k_1-k_0), \\
				+\infty, & \text{if }\lambda >2 \gamma \alpha (k_2-k_1-k_0).
			\end{cases}
		\end{align}
		Thus, if $\lambda \leq 2 \gamma \alpha (k_2-k_1-k_0)$ then it holds that (i) $\Phi(d_1 (\xi_3^*,0)) - \Phi(d_1 (\alpha \hat{\xi},0)) \xrightarrow{y \to 0} 1$ and $\Phi (d_{1} (\xi _{1}^{*},0)) \xrightarrow{y \to 0} 0$. 
		On the other hand, if $\lambda >2 \gamma \alpha (k_2-k_1-k_0)$, then it holds that (ii) $\Phi\left(d_1 \left(\xi_1^*,0\right)\right) \xrightarrow{y \to 0} 1$ as \eqref{eq: xi1* bigger 0} entails that $\xi_1^* \xrightarrow{y \to 0} \infty$. Now, if $x_0 e^{\int_0^T r_s \diff s} \leq k_1$, it holds that (iii) $\left(k_0+k_1 + \frac{\lambda}{2 \gamma \alpha} \right) e^{-\int_0^T r_s \diff s} > x_0$ for $\lambda > C = -2\gamma \alpha k_0$. Furthermore, if $x_0 e^{\int_0^T r_s \diff s} \in (k_1,k_2]$, we deduce that (iv) $\left(k_0+k_1 + \frac{\lambda}{2 \gamma \alpha} \right) e^{-\int_0^T r_s \diff s} > x_0$ for $\lambda > C = 2 \gamma \alpha (x_0 e^{\int_0^T r_s \diff s}-k_1-k_0)$ and, if $x_0 e^{\int_0^T r_s \diff s} \leq k_2$, it follows directly that (v) $\left(k_2 + \frac{\lambda}{2 \gamma \tilde{\alpha}} - \frac{\alpha}{\tilde{\alpha}}(k_2-k_1-k_0) \right) e^{-\int_0^T r_s \diff s} > x_0$ for $\lambda >2 \gamma \alpha (k_2-k_1-k_0) (\geq C)$. Finally, if $x_0 e^{\int_0^T r_s \diff s} > k_2$, then it holds that $C > 2 \gamma \alpha (k_2-k_1-k_0)$ and (vi) $\left(k_2 + \frac{\lambda}{2 \gamma \tilde{\alpha}} - \frac{\alpha}{\tilde{\alpha}}(k_2-k_1-k_0) \right) e^{-\int_0^T r_s \diff s} > x_0$ for $\lambda > C= 2 \gamma \left(\tilde{\alpha} x_0 e^{\int_0^T r_s \diff s} + \alpha_2 k_2 - \alpha (k_0+k_1)\right)$. Note that $\tilde{\alpha} \hat{\xi} \leq \xi_2^* \leq \alpha \hat{\xi}$ (see \eqref{order of xis}) and hence (vii) $\xi_2^*$ shows the same limiting behavior as $\hat{\xi}$ for $y \to 0$. Summarizing these results, we conclude for \underline{$\lambda>C$}:
		\begin{enumerate}[(a)]
			\item Properties (i), (ii), (v), and (vi) imply:
			\begin{align*}
				\lim_{y \to 0} \left(k_2 + \dfrac{\lambda}{2 \gamma \tilde{\alpha}} - \dfrac{\alpha}{\tilde{\alpha}}(k_2-k_1-k_0) \right) e^{-\int_0^T r_s \diff s} \Phi\left(d_1 \left(\xi_1^{*},0\right)\right) & \\
				&\hspace{-5.3cm}\begin{cases}
					> x_0 & \text{ if $\lambda > 2\gamma \alpha (k_2-k_1-k_0)$ and $x_0e^{\int_0^T r_s \diff s} \leq k_2$,} \\
					> x_0 & \text{ if $x_0e^{\int_0^T r_s \diff s} > k_2$,} \\
					\geq 0 & \text{ else.} 
				\end{cases}
			\end{align*}
			\item One can see directly that $\lim_{y \to 0}\dfrac{y}{2\gamma \tilde{\alpha}^2}  e^{\int_0^T -2r_s+\norm{\kappa_s}^2 \diff s} \Phi \left( d_2 \left(\xi_1^{*},0\right) \right) = 0$.
			\item It holds that $\lim_{y \to 0}k_2 e^{-\int_0^T r_s \diff s} \left( \Phi\left(d_1 \left(\xi_2^{*},0\right)\right) - \Phi\left(d_1 \left(\tilde{\alpha} \hat{\xi},0\right)\right) \right) = 0$ due to (vii).
			\item Properties (i), (iii), (iv), $\lambda > -2\gamma\alpha k_0$, and \eqref{order of xis} imply:
			\begin{align*}
				\lim_{y \to 0}\left(k_0+k_1 + \dfrac{\lambda}{2 \gamma \alpha} \right) e^{-\int_0^T r_s \diff s} \left( \Phi\left(d_1 \left(\xi_3^{*},0\right)\right) - \Phi\left(d_1 \left(\alpha \hat{\xi},0\right)\right) \right) &\\
				&\hspace{-6cm} \begin{cases}
					> x_0 & \text{ if $\lambda \in (C,2 \gamma \alpha (k_2-k_1-k_0)]$ and $x_0e^{\int_0^T r_s \diff s}\leq k_2$,} \\
					\geq 0 & \text{ else.}
				\end{cases}
			\end{align*}
			\item One can see directly that $\lim_{y \to 0}\dfrac{y}{2\gamma \alpha^2} e^{\int_0^T -2r_s+\norm{\kappa_s}^2 \diff s} \left(\Phi \left( d_2 \left(\xi_3^{*},0\right) \right) - \Phi \left( d_2 \left(\alpha \hat{\xi},0\right) \right) \right) = 0$.
		\end{enumerate}
		Note that since either $x_{0}e^{\int _{0}^{T} r_{s} \diff s} > k_{2}$, or $x_{0}e^{\int _{0}^{T} r_{s} \diff s} \leq k_{2}$ and $\lambda > 2\gamma \alpha (k_{2}-k_{1}-k_{0})$ or $\lambda \in (C,2 \gamma \alpha (k_{2}-k_{1}-k_{0})]$, we always have in (a) or in (d) "$> x_0$". Thus, the claim follows.
	\end{myproof}
	\begin{myproof}[Proof of \eqref{eq: xi3* bigger 0}]
		We show this equation by considering the three different cases of $C$:
		
		\underline{Case 1:} $x_0 e^{\int_0^T r_s \diff s} \leq k_1$, i.e., $C = -2\gamma \alpha k_0$: \\
		Let $\lambda > C$, i.e., there exists an $\varepsilon >0$ such that $\lambda = C + 2 \gamma \alpha \varepsilon$. First, we notice that $\xi_3^* >0$ if $\tilde{\xi}_3^* := \dfrac{\lambda \alpha}{y} + \dfrac{2 \gamma \alpha^2 k_0}{y} - \dfrac{2 \gamma \alpha^2}{y} \left( \sqrt{k_1^2 +k_1 \left(2k_0 + \dfrac{\lambda}{\gamma \alpha}\right)} - k_1 \right)>0$. Now, it holds if we plug in $\lambda$:
		\begin{align*}
			\tilde{\xi}_3^* &= \dfrac{2 \gamma \alpha^2}{y} \left( -k_0 + \varepsilon + k_0 + k_1 - \sqrt{k_1^2 + k_1 \left( 2 k_0 - 2k_0 + 2\varepsilon \right)}\right) 
			= \dfrac{2 \gamma \alpha^2}{y} \left( k_1 + \varepsilon - \sqrt{(k_1 + \varepsilon)^2 - \varepsilon^2}\right).
		\end{align*}
		Since $y, \gamma, \alpha, \varepsilon >0$ and $k_1 \geq 0$, we get that $\tilde{\xi}_3^* > 0$ if and only if $(k_1 + \varepsilon)^2 > \left(\sqrt{(k_1 + \varepsilon)^2 - \varepsilon^2}\right)^2$,
		which is equivalent to $\varepsilon^2>0$. Thus, the claim, i.e., $\xi_3^*>0$ for $\lambda>C$, follows. 
		
		\underline{Case 2:} $x_0 e^{\int_0^T r_s \diff s} \in (k_1,k_2]$, i.e., $C = 2\gamma \alpha (x_0 e^{\int_0^T r_s \diff s} - k_1 - k_0)$: \\
		Let $\lambda > C$, i.e., there exists an $\varepsilon >0$ such that $\lambda = C + 2 \gamma \alpha \varepsilon$. As in the first case, we get that $\xi_3^*>0$ if $\tilde{\xi}_3^*>0$ (defined as in Case 1). Then, we get after plugging in $\lambda$:
		\begin{align*}
			\tilde{\xi}_3^* &= \dfrac{2 \gamma \alpha^2}{y} \left( x_0 e^{\int_0^T r_s \diff s} - k_1 - k_0 + \varepsilon + k_0 + k_1 - \sqrt{k_1^2 + k_1 \left( 2 k_0 +2 x_0 e^{\int_0^T r_s \diff s} - 2 k_1 - 2k_0 + 2\varepsilon\right)}\right) \\
			&= \dfrac{2 \gamma \alpha^2}{y} \left( x_0 e^{\int_0^T r_s \diff s} + \varepsilon - \sqrt{ -k_1^2 + 2 k_1 x_0 e^{\int_0^T r_s \diff s} - \left(x_0 e^{\int_0^T r_s \diff s}\right)^2 + \left(x_0 e^{\int_0^T r_s \diff s}\right)^2 + 2\varepsilon k_1}\right) \\
			&= \dfrac{2 \gamma \alpha^2}{y} \left( x_0 e^{\int_0^T r_s \diff s} + \varepsilon - \sqrt{-\left(x_0 e^{\int_0^T r_s \diff s}-k_1\right)^2 + \left(x_0e^{\int_0^T r_s \diff s}\right)^2 + 2\varepsilon k_1} \right).
		\end{align*}
		Since $y, \gamma, \alpha, \varepsilon >0$ and $x_0 e^{\int_0^T r_s \diff s} \geq 0$, we get that $\tilde{\xi}_3^* > 0$ if and only if 
		\begin{align*}
			\left( x_0 e^{\int_0^T r_s \diff s} + \varepsilon \right)^2 > -\left(x_0 e^{\int_0^T r_s \diff s}-k_1\right)^2 + \left(x_0e^{\int_0^T r_s \diff s}\right)^2 + 2\varepsilon k_1.
		\end{align*}
		Since $\left( x_0 e^{\int_0^T r_s \diff s} + \varepsilon \right)^2 = \left(x_0 e^{\int_0^T r_s \diff s}\right)^2 + 2 \varepsilon x_0 e^{\int_0^T r_s \diff s} + \varepsilon^2$, this is equivalent to:
		\begin{align*}
			0 < \left(x_0 e^{\int_0^T r_s \diff s}-k_1\right)^2 - 2\varepsilon k_1 + 2 \varepsilon x_0 e^{\int_0^T r_s \diff s} + \varepsilon^2 = \left( \left(x_0 e^{\int_0^T r_s \diff s}-k_1\right) + \varepsilon\right)^2.
		\end{align*}
		Now, the claim, i.e., $\xi_3^*>0$ for $\lambda>C$, follows.
		
		\underline{Case 3:} $x_0 e^{\int_0^T r_s \diff s} > k_2$, i.e., $C = 2 \gamma \left(\tilde{\alpha} x_0 e^{\int_0^T r_s \diff s} + \alpha_2 k_2 - \alpha (k_0+k_1)\right)$: \\
		Then, we get immediately that $\hat{\xi}>0$ for $\lambda>C (> 2\gamma \alpha (k_2-k_1-k_0) \text{ since } \alpha_2+\tilde{\alpha}=\alpha)$. Hence, it holds that $\xi_3^*>0$ due to $\xi_3^* \geq \alpha \hat{\xi}$ (see \eqref{order of xis}).
	\end{myproof}
	\begin{myproof}[Proof of \eqref{eq: xi1* bigger 0}]
		First, let $\lambda > 2 \gamma \alpha (k_2-k_1-k_0)$, i.e., there exists an $\varepsilon >0$ such that $\lambda = 2 \gamma \alpha (k_2-k_1-k_0) + 2 \gamma \tilde{\alpha} \varepsilon$. We note that $\hat{\xi} >0$ for this $\lambda$ which implies that $\xi_1^* > 0$ if and only if $\tilde{\xi}_1^* > 0$. Now, if the term under the square root in the maximum in the formula of $\tilde{\xi}_1^*$ (for the definition, see Theorem \ref{optimal wealth}) is negative, the claim follows immediately. Hence, we assume that the term is non-negative. If we plug in $\lambda$ into $\tilde{\xi}_1^*$, we get:
		\begin{align*}
			\tilde{\xi}_1^* &= \dfrac{2 \gamma \tilde{\alpha}}{y} \left( \dfrac{\lambda}{2 \gamma} - \alpha (k_2-k_1-k_0) + \tilde{\alpha} k_2 - \sqrt{ (\alpha (k_0+k_1) - \alpha_2 k_2)^2 - \alpha^2k_0^2 + \dfrac{\lambda}{\gamma} (\alpha k_1 - \alpha_2 k_2)}\right) \\
			&= \dfrac{2 \gamma \tilde{\alpha}}{y} \left( \tilde{\alpha} \varepsilon + \tilde{\alpha} k_2 - \sqrt{ \alpha_2^2 k_2^2 + 2 \alpha^2 k_1k_2 -2 \alpha \alpha_2 k_2^2 - \alpha^2 k_1^2 +2\alpha \tilde{\alpha} k_1 \varepsilon - 2 \tilde{\alpha} \alpha_2 k_2 \varepsilon}\right) \\
			&= \dfrac{2 \gamma \tilde{\alpha}^2}{y} \left( k_2 + \varepsilon - \frac{1}{\tilde{\alpha}}\sqrt{ k_2^2 \alpha_2 (\alpha_2-2\alpha) + \alpha^2 k_1 (2k_2-k_1)+2\tilde{\alpha} \varepsilon (\alpha k_1 -\alpha_2 k_2)}\right).
		\end{align*}
		Since $y, \gamma, \tilde{\alpha}, \varepsilon >0$ and $k_2 \geq 0$, it holds that $\tilde{\xi}_1^* > 0$ if and only if 
		\begin{align*}
			k_2^2 + 2 \varepsilon k_2 + \varepsilon^2 > \frac{1}{\tilde{\alpha}^2} \left( k_2^2 \alpha_2 (-\tilde{\alpha}-\alpha) + \alpha^2 k_1 (2k_2-k_1)+2\tilde{\alpha} \varepsilon (\alpha k_1 -\alpha_2 k_2) \right).
		\end{align*}
		This is equivalent to:
		\begin{align*}
			L := k_2^2 \left( 1 + \dfrac{\alpha_2}{\tilde{\alpha}} + \dfrac{\alpha \alpha_2}{\tilde{\alpha}^2}\right) + 2\varepsilon \left( k_2 \left(1 + \dfrac{\alpha_2}{\tilde{\alpha}}\right) - k_1 \frac{\alpha}{\tilde{\alpha}}\right)  + \varepsilon^2 - \dfrac{\alpha^2}{\tilde{\alpha}^2} k_1 (2 k_2 -k_1) >0.
		\end{align*}
		Now, we notice with $\alpha=\tilde{\alpha}+\alpha_2$ that
		\begin{align*}
			1 + \dfrac{\alpha_2}{\tilde{\alpha}} + \dfrac{\alpha \alpha_2}{\tilde{\alpha}^2} = \dfrac{\tilde{\alpha}^2+\alpha_2 \tilde{\alpha}+\tilde{\alpha}\alpha_2+\alpha_2^2}{\tilde{\alpha}^2} = \dfrac{(\tilde{\alpha}+\alpha_2)^2}{\tilde{\alpha}^2} = \dfrac{\alpha^2}{\tilde{\alpha}^2}.
		\end{align*}
		Then, it follows with $L$ as above since $1 + \frac{\alpha_2}{\tilde{\alpha}} = \frac{\alpha}{\tilde{\alpha}}$:
		\begin{align*}
			\dfrac{\tilde{\alpha}^2}{\alpha^2} L &= k_2^2 + 2 \dfrac{\tilde{\alpha}}{\alpha} \varepsilon (k_2-k_1) + \dfrac{\tilde{\alpha}^2}{\alpha^2} \varepsilon^2 - 2 k_1 k_2 + k_1^2 \\
			&= (k_2-k_1)^2 + 2 \dfrac{\tilde{\alpha}}{\alpha} \varepsilon (k_2-k_1) + \left(\dfrac{\tilde{\alpha}}{\alpha} \varepsilon\right)^2 = \left( (k_2-k_1) + \frac{\tilde{\alpha}}{\alpha} \varepsilon\right)^2 .
		\end{align*}
		Hence, $L > 0$. 
		Second, we note that $\hat{\xi} = 0$ (and thus also $\xi_1^*$) for $\lambda \leq 2 \gamma \alpha (k_2-k_1-k_0)$ by definition. Hence, the claim $\xi_1^* > 0$ if and only if $\lambda > 2 \gamma \alpha (k_2-k_1-k_0)$ follows.
	\end{myproof}
	\begin{myproof}[Proof of \eqref{eq: f1 infty}]
		By definition, it holds that $\bar{\xi} \xrightarrow{y \to \infty} 0$. Then, we get the claim due to \eqref{order of xis} and $\lim_{x \to 0} \Phi (d_i(x,0)) = 0$ for $i \in \{1,2\}$.
	\end{myproof}
	\begin{myproof}[Proof of \eqref{eq: f2 C}]
		Recall that $h$ is a continuous function from $(C,\infty) \to \Real_{\geq 0}$ with some limiting properties at $C$ and $\infty$, and we plug this function into $f_2$ with $y=h(\lambda)$. To start, we give a small overview of the following proof: To show \eqref{eq: f2 C}, we must find the limiting behavior of $\xi_1^*$, $\tilde{\alpha} \hat{\xi}$, $\xi_2^*$, $\alpha \hat{\xi}$, and $\xi_3^*$ when $\lambda \searrow C$. However, this behavior heavily depends on $x_0 e^{\int_0^T r_s \diff s}$, i.e., the different values of $C$ and the limiting properties of $h$. Hence, we have to distinguish several cases. First, we differentiate between the different limiting properties for $\lambda \searrow C$ depending if $x_0 e^{\int_0^T r_s \diff s} > k_1$ is true. If this is true, we first analyze the limit of $\xi_3^*$. After that, we have to separate the cases depending if $C$ is bigger, smaller, or equal to $2 \gamma \alpha (k_2-k_1-k_0)$. For the equality case, we even have to separate along the limiting behavior of $\hat{\xi}$ to derive the limiting behavior for $\xi_1^*$ and $\xi_2^*$. If $x_0 e^{\int_0^T r_s \diff s} > k_1$ is not true, we only have to separate along the possible limiting values of $h$. 
		
		\underline{Case 1:} $x_0 e^{\int_0^T r_s \diff s} > k_1$, i.e., $\lim_{x \to C} h(x)=0$: \\
		Since $\lim_{\lambda \searrow C} h(\lambda)=0$ and $C > -2\gamma \alpha k_0$, it holds that (i) $\xi_3^* \to + \infty$ for $\lambda \searrow C$, i.e., $y=h(\lambda) \xrightarrow{\lambda \searrow C} 0$. Indeed, it holds: If $x_0 e^{\int_0^T r_s \diff s} > k_2$, we get that $C>2\gamma\alpha (k_2-k_1-k_0)$ and hence $\hat{\xi} \xrightarrow{\lambda \searrow C} + \infty$ since $y = h(\lambda) \xrightarrow{\lambda \searrow C} 0$. Now, \eqref{order of xis} implies that $\xi_3^* \xrightarrow{\lambda \searrow C} + \infty$. If $x_0 e^{\int_0^T r_s \diff s} \leq k_2$, it holds that $C=2\gamma \alpha \left( x_0 e^{\int_{0}^{T} r_s \diff s}-k_1-k_0\right)$ since $x_0 e^{\int_0^T r_s \diff s} > k_1$ by assumption and we get that:
		\begin{align*}
			\xi_3^* \geq \dfrac{\alpha \lambda + 2\gamma \alpha^2 \left(k_0+k_1-\sqrt{k_1^2 +k_1(2k_0+\frac{\lambda}{\alpha \gamma})}\right)}{h(\lambda)} =: \dfrac{D}{h(\lambda)}.
		\end{align*}
		When calculating the limit for $\lambda \searrow C$ for $D$, it holds that
		\begin{align*}
			\lim_{\lambda \to C} D &= 2 \gamma \alpha^2 \left( x_0 e^{\int_0^T r_s \diff s} -k_1-k_0+k_0+k_1- \sqrt{k_1^2 +k_1(2k_0+2 x_0 e^{\int_0^T r_s \diff s} -2k_1-2k_0)}\right) \\
			&= 2 \gamma \alpha^2 \left( x_0 e^{\int_0^T r_s \diff s} -\sqrt{-k_1^2 + 2 k_1 x_0 e^{\int_0^T r_s \diff s} }\right).
		\end{align*}
		Now, we see that $\lim_{\lambda \searrow C} D > 0$ if and only if $\left( x_0 e^{\int_0^T r_s \diff s} - k_1 \right)^2 > 0$ which is true by assumption. 
		Thus, $\xi_3^* \xrightarrow{\lambda \searrow C} + \infty$ since $y=h(\lambda) \xrightarrow{\lambda \searrow C} 0$, i.e., (i) is proven. \\
		Next, if $C > 2 \gamma \alpha (k_2-k_1-k_0)$, we get that ((ii).1) $\xi_1^* \xrightarrow{\lambda \searrow C} + \infty$ due to \eqref{eq: xi1* bigger 0} and $y=h(\lambda) \xrightarrow{\lambda \searrow C} 0$ and thus also ((ii).2) $\alpha \hat{\xi}, \xi_2^*,\tilde{\alpha}\hat{\xi} \xrightarrow{\lambda \searrow C} + \infty$ due to \eqref{order of xis}. If $C < 2 \gamma \alpha (k_2-k_1-k_0)$, we get that ((iii).1) $\alpha \hat{\xi} \xrightarrow{\lambda \searrow C} 0$ due to \eqref{eq: xihat 0 infty} and $y=h(\lambda) \xrightarrow{\lambda \searrow C} 0$ which can be factored out. Hence, also ((iii).2) $\xi_1^*, \tilde{\alpha} \hat{\xi} , \xi_2^*  \xrightarrow{\lambda \searrow C} 0$ due to \eqref{order of xis}. The remaining case that $C = 2 \gamma \alpha (k_2-k_1-k_0)$ needs a closer look. First, we note that then $k_2 = x_0 e^{\int_0^T r_s \diff s}$ and hence $k_2>k_1$ by the assumption of Case 1. There, we have to distinguish three more cases for the limiting behavior of $\xi_1^*$, $\tilde{\alpha} \hat{\xi}$, $\xi_2^*$, and $\alpha \hat{\xi}$ depending on the behavior of $\hat{\xi}$. Note that $\hat{\xi}$ can have multiple accumulation points for $\lambda \searrow C$ (all in [0,$\infty$]) despite $\hat{\xi}$ being continuous in $\lambda$. Now, let $\lambda_n$ be a sequence converging to the $\liminf$, i.e., $\lim_{n \to \infty} f_2(h(\lambda_n),\lambda_n) = \liminf_{\lambda \searrow C} f_2(h(\lambda),\lambda)$. Note that $\lambda_n \xrightarrow{n \to \infty} C$ from above. We show \eqref{eq: f2 C} then for each accumulation point separately by possibly switching to subsequences, i.e., we can assume without loss of generality that $\lim_{n \to \infty} \hat{\xi}^{\lambda_n}$ exists (in [0,$\infty$]) and have to consider the following cases:
		
		\underline{Case 1.1:} $\hat{\xi}^{\lambda_n} \xrightarrow{n \to\infty} 0$:\\
		Then, it follows immediately that (iv) $\xi_1^{*,\lambda_n}, \xi_2^{*,\lambda_n} \xrightarrow{n \to \infty} 0$ due to \eqref{order of xis}.
		
		\underline{Case 1.2:} $\hat{\xi}^{\lambda_n} \xrightarrow{n \to \infty} + \infty$:\\
		Here, we first note that when applying the same argument as in the proof of \eqref{eq: xi1* bigger 0} (with $\varepsilon=0$) to the second term in the definition of $\tilde{\xi}_1^{*,\lambda_n}$, this second term is non-negative for $\lambda=C$, i.e., $\tilde{\xi}_1^{*,\lambda_n} \geq \tilde{\alpha} \hat{\xi}^{\lambda_n}$ when taking the limit $n \to \infty$. Then, we obtain ((v).1) $\xi_1^{*,\lambda_n} \xrightarrow{n \to \infty} + \infty$. In particular, we get ((v).2) $\xi_1^{*,\lambda_n},\tilde{\alpha} \hat{\xi}^{\lambda_n},\xi_2^{*,\lambda_n},\alpha \hat{\xi}^{\lambda_n} \xrightarrow{n \to \infty} + \infty$ due to \eqref{order of xis}.
		
		\underline{Case 1.3:} $\hat{\xi}^{\lambda_n} \xrightarrow{n \to \infty} c \in (0,\infty)$:\\
		First, since $k_2>k_1$, we get by plugging in $\lambda = 2 \gamma \alpha (k_2-k_1-k_0)$ into $\tilde{\xi}_2^{*}$ that $\tilde{\xi}_2^{*,\lambda_n} \xrightarrow{n \to \infty} + \infty$. The reason is that $y_n=h(\lambda_n) \xrightarrow{n \to \infty} 0$ and the nominator in $\tilde{\xi}_2^{*,\lambda_n}$ converges to a positive number, i.e.:
		\begin{align*}
			\alpha \lambda k_2 - \gamma \alpha^2 &(k_2-k_1)^2 + 2 \gamma \alpha^2 k_0 (k_2-k_1) - \lambda_n \alpha k_1 \\
			\xrightarrow{n \to \infty}& \ 2 \gamma \alpha^2 k_2 (k_2-k_1-k_0) - \gamma \alpha^2 (k_2-k_1)^2 + 2 \gamma \alpha^2 k_0 (k_2-k_1) - 2 \gamma \alpha^2 k_1 (k_2-k_1-k_0) \\
			=& \ 2 \gamma \alpha^2 (k_2-k_1-k_0) (k_2-k_1) - \gamma \alpha^2 (k_2-k_1) (k_2-k_1-2k_0) \\
			=& \ \gamma \alpha^2 (k_2-k_1) (2k_2-2k_1-2k_0-k_2+k_1+2k_0) = \gamma \alpha^2 (k_2-k_1)^2 > 0.
		\end{align*}
		Hence, ((vi).1) $\xi_2^{*,\lambda_n} \xrightarrow{n \to \infty} \alpha c$. Second, we note that $\tilde{\xi}_1^{*,\lambda_n} \geq \tilde{\alpha} \hat{\xi}^{\lambda_n}$ for $n \to \infty$ by the same argument as in Case 1.2. Thus, ((vi).2) $\xi_1^{*,\lambda_n} \xrightarrow{n \to \infty} \tilde{\alpha} c$.
		
		Summarizing, we get in all subcases of Case 1 with $C=2 \gamma \alpha (k_2-k_1-k_0)$, that $\Phi(d_0 (\xi_3^*,0)) \xrightarrow{\lambda \searrow C} 1$ since $\lim_{x \to \infty} \Phi \left(d_0 \left(x,0\right)\right) =1$ due to (i) and $-\Phi(d_0 (\alpha \hat{\xi}^{\lambda_n},0)) + \Phi(d_0 (\xi_2^{*,\lambda_n},0)) - \Phi(d_0 (\tilde{\alpha} \hat{\xi}^{\lambda_n},0)) + \Phi(d_0 (\xi_1^{*,\lambda_n},0)) \xrightarrow{n \to \infty} 0$ due to either (ii), (iii), (iv), (v), or (vi) depending on $C$. Hence, it holds that $\Phi(d_0 (\xi_3^{*,\lambda_n},0)) - \Phi(d_0 (\alpha \hat{\xi}^{\lambda_n},0)) + \Phi(d_0 (\xi_2^{*,\lambda_n},0)) - \Phi(d_0 (\tilde{\alpha} \hat{\xi}^{\lambda_n},0)) + \Phi(d_0 (\xi_1^{*,\lambda_n},0)) \xrightarrow{n \to \infty} 1$ since $\lim_{x \to \infty} \Phi \left(d_0 \left(x,0\right)\right) =1$ and $\lambda_n \xrightarrow{n \to \infty} C$ from above. In particular, we get that $\lambda ( -1 + \Phi(d_0 (\xi_3^{*,\lambda_n},0)) - \Phi(d_0 (\alpha \hat{\xi}^{\lambda_n},0)) + \Phi(d_0 (\xi_1^{*,\lambda_n},0)) ) +2 \gamma \alpha (k_2-k_1-k_0) ( \Phi(d_0 (\xi_2^{*,\lambda_n},0)) - \Phi(d_0 (\tilde{\alpha} \hat{\xi}^{\lambda_n},0)) ) \xrightarrow{n \to \infty} 0$ since $C = 2 \gamma \alpha (k_2-k_1-k_0)$ and $\lambda_n \xrightarrow{n \to \infty} C$ from above. Thus, the claim follows.
		
		\underline{Case 2:} $x_0 e^{\int_0^T r_s \diff s} \leq k_1$, i.e., $\liminf_{x \searrow C} h(x) \geq 0$ and $C = -2\gamma \alpha k_0$:\\
		For this case, we have to distinguish two more subcases depending on the possible limit of $h$. However, $h$ can have multiple accumulation points. As before, we consider each accumulation point separately, denote by $\lambda_n$ a sequence converging to the $\liminf$ of $h$, and assume, without loss of generality by possibly switching to a subsequence, that $\lambda_n \xrightarrow{n \to \infty} C$ from above and $\lim_{n \to \infty} h(\lambda_n)$ exists (in $[0,\infty]$). This gives us the following cases:
		
		\underline{Case 2.1:} $\lim_{n \to \infty} h(\lambda_n) > 0$: \\
		Then, we get that $\bar{\xi}^{\lambda_n} \xrightarrow{n \to \infty} 0$ since $\bar{\xi}^{\lambda_n} = \frac{\alpha (\lambda_n-C)}{h(\lambda_n)}$ for $C = -2\gamma \alpha k_0$ which implies that $\hat{\xi}^{\lambda_n},\xi_1^{*,\lambda_n},\xi_2^{*,\lambda_n},\xi_3^{*,\lambda_n} \xrightarrow{n \to \infty} 0$ using \eqref{order of xis}. Thus, it holds that $\lim_{n \to \infty} f_2 (h(\lambda_n),\lambda_n) = 1-2\gamma \alpha k_0 - C = 1$ using that $\lim_{x \to 0} \Phi (d_i(x,0)) = 0$ for $i \in \{0,1\}$.
		
		\underline{Case 2.2:} $\lim_{n \to \infty} h(\lambda_n) = 0$:\\
		Here, we first rewrite $f_2$ into:
		\begin{align*}
			f_2(h(\lambda_n),\lambda_n) :=& \ 1 + 2\gamma\alpha k_0 \left(-1 + \Phi\left(d_0 \left(\xi_3^{*,\lambda_n},0\right)\right) - \Phi\left(d_0 \left(\alpha \hat{\xi}^{\lambda_n},0\right)\right) + \Phi\left(d_0 \left(\xi_1^{*,\lambda_n},0\right)\right) \right) \\
			&+ \lambda_n \left( -1 + \Phi\left(d_0 \left(\xi_3^{*,\lambda_n},0\right)\right) - \Phi\left(d_0 \left(\alpha \hat{\xi}^{\lambda_n},0\right)\right) + \Phi\left(d_0 \left(\xi_1^{*,\lambda_n},0\right)\right) \right) \\
			&- h(\lambda_n) \left( \dfrac{1}{\alpha}\Phi\left(d_1 \left(\xi_3^{*,\lambda_n},0\right)\right) - \dfrac{1}{\alpha}\Phi\left(d_1 \left(\alpha \hat{\xi}^{\lambda_n},0\right)\right) + \dfrac{1}{\tilde{\alpha}}\Phi\left(d_1 \left(\xi_1^{*,\lambda_n},0\right)\right) \right) \\
			&+2 \gamma \alpha (k_2-k_1) \left( \Phi\left(d_0 \left(\xi_2^{*,\lambda_n},0\right)\right) - \Phi\left(d_0 \left(\tilde{\alpha} \hat{\xi}^{\lambda_n},0\right)\right) \right) \\
			=& \ 1 + (\lambda_n-C) \left(-1 + \Phi\left(d_0 \left(\xi_3^{*,\lambda_n},0\right)\right) - \Phi\left(d_0 \left(\alpha \hat{\xi}^{\lambda_n},0\right)\right) + \Phi\left(d_0 \left(\xi_1^{*,\lambda_n},0\right)\right) \right) \\
			&- h(\lambda_n) \left( \dfrac{1}{\alpha}\Phi\left(d_1 \left(\xi_3^{*,\lambda_n},0\right)\right) - \dfrac{1}{\alpha}\Phi\left(d_1 \left(\alpha \hat{\xi}^{\lambda_n},0\right)\right) + \dfrac{1}{\tilde{\alpha}}\Phi\left(d_1 \left(\xi_1^{*,\lambda_n},0\right)\right) \right) \\
			&+2 \gamma \alpha (k_2-k_1) \left( \Phi\left(d_0 \left(\xi_2^{*,\lambda_n},0\right)\right) - \Phi\left(d_0 \left(\tilde{\alpha} \hat{\xi}^{\lambda_n},0\right)\right) \right),
		\end{align*}
		where we used that $C = -2\gamma \alpha k_0$. Now, the claim follows, i.e., $\liminf_{\lambda \searrow C} f_2(h(\lambda),\lambda) \geq 1$ since $\lim_{n \to\infty} h(\lambda_n) = 0$, $\lambda_n \xrightarrow{n \to \infty} C$ from above, $\xi_2^{*,\lambda_n} \geq \tilde{\alpha} \hat{\xi}^{\lambda_n}$, and $\Phi(d_0(\cdot,0))$ being non-decreasing. 
	\end{myproof}
	\begin{myproof}[Proof of \eqref{eq: f2 infty}]
		For this proof, we have to consider three cases depending on the limiting behavior of $\frac{\lambda}{h(\lambda)}$ for $\lambda\to \infty$. As in the proof of \eqref{eq: f2 C}, we have here possibly multiple accumulation points that we treat separately and denote by $\lambda_n$ the sequence converging to the $\limsup$. Note that $\lambda_n \xrightarrow{n \to \infty} \infty$ in this case. Hence, we can assume, without loss of generality by possibly switching to a subsequence, that $\lim_{n \to \infty} \frac{\lambda_n}{h(\lambda_n)}$ exists (in $[0,\infty]$) and consider the three cases that $\lim_{n \to \infty} \frac{\lambda_n}{h(\lambda_n)} = 0$, $\lim_{n \to \infty} \frac{\lambda_n}{h(\lambda_n)} = c \in (0,\infty)$, and $\lim_{n \to \infty} \frac{\lambda_n}{h(\lambda_n)} = \infty$. Note that the factor to the right of "$\lambda$" in $f_2$ is negative, while the factor to the right of ``$y$'' in $f_2$ is positive by \eqref{order of xis}.
		
		\underline{Case 1:} $\lim_{n \to \infty} \frac{\lambda_n}{h(\lambda_n)} = 0$:\\
		It holds that $\bar{\xi}^{\lambda_n} \xrightarrow{n \to \infty} 0$ using $y_n=h(\lambda_n)$ and the assumption. Therefore, we also get that $\hat{\xi}^{\lambda_n},\xi_1^{*,\lambda_n},\xi_2^{*,\lambda_n},\xi_3^{*,\lambda_n} \xrightarrow{n \to \infty} 0$ due to \eqref{order of xis}. Thus, the claim follows using $\lim_{x \to 0} \Phi (d_i(x,0)) = 0$ for $i \in \{0,1\}$ and $h(\lambda _{n}) \xrightarrow{n \to \infty} + \infty $.
		
		\underline{Case 2:} $\lim_{n \to \infty} \frac{\lambda_n}{h(\lambda_n)} = c \in (0,\infty)$:\\
		In this case, it holds by assumption that $\hat{\xi}^{\lambda_n} \xrightarrow{n \to \infty} c$, and $\bar{\xi}^{\lambda_n},\xi_3^{*,\lambda_n} \xrightarrow{n \to \infty} \alpha c$. Moreover, it holds that $\lim_{n \to \infty} \xi_2^{*,\lambda_n} \in [\tilde{\alpha}c,\alpha c]$ by \eqref{order of xis} and  $\xi_1^{*,\lambda_n} \xrightarrow{n \to \infty} \tilde{\alpha} c$ since $\frac{\lambda_n-\sqrt{l(\lambda_n)}}{h(\lambda_n)} \xrightarrow{n \to \infty} c$ for any affine function $l$. Next, we note that $\Phi(d_i(ac,0)) \in (0,1)$ for all $a>0$ and $i \in \{0,1\}$. In total, we get that $\lim_{n \to \infty} -1 + \Phi(d_0 (\xi_3^{*,\lambda_n},0)) - \Phi(d_0 (\alpha \hat{\xi}^{\lambda_n},0)) + \Phi(d_0 (\xi_1^{*,\lambda_n},0)) = -1 + \Phi(d_0 (\tilde{\alpha} c,0)) < 0$ and $\lim_{n \to \infty} \frac{1}{\alpha}\Phi(d_1 (\xi_3^{*,\lambda_n},0)) - \frac{1}{\alpha}\Phi(d_1 (\alpha \hat{\xi}^{\lambda_n},0)) + \frac{1}{\tilde{\alpha}}\Phi(d_1 (\xi_1^{*,\lambda_n},0)) = \frac{1}{\tilde{\alpha}}\Phi(d_1 (\tilde{\alpha} c,0)) > 0$. Thus, the claim follows since $\lambda_n,h(\lambda_n) \xrightarrow{n \to \infty} + \infty$.
		
		\underline{Case 3:} $\lim_{n \to \infty} \frac{\lambda_n}{h(\lambda_n)} = \infty$:\\
		In this case, we get that $\xi_1^{*,\lambda_n} \xrightarrow{n \to \infty} \infty$ as $\lim_{n \to \infty} \frac{\lambda_n-\sqrt{l(\lambda_n)}}{h(\lambda_n)} = \infty$ for any affine function $l$. Hence, it holds that $\hat{\xi}^{\lambda_n},\bar{\xi}^{\lambda_n},\xi_2^{*,\lambda_n},\xi_3^{*,\lambda_n} \xrightarrow{n \to \infty} \infty$ due to \eqref{order of xis}. Thus, we get the claim using $\lim_{n \to \infty} h(\lambda_n) = \infty$, $\lim_{n \to \infty} \lambda_n = \infty$, and $\lim_{x \to \infty} \Phi(d_i(x,0)) = 1$ for $i \in \{0,1\}$.
	\end{myproof}
	\begin{myproof}[Proof of \eqref{eq: y star continuous}]
		To show this, it is sufficient to show that (i) $\lambda \mapsto f_1(\cdot,\lambda)$ is strictly increasing, (ii) $y \mapsto f_1(y,\cdot)$ is strictly decreasing and (iii) $f_1$ is jointly continuous on $\Real_{>0} \times (C,\infty)$. Indeed, let $\lambda, \lambda_n > C$ with $\lambda_n \xrightarrow{n \to \infty} \lambda$. Due to the existence of the zero root (see the main part of the proof) and the uniqueness of the zero root (by (ii)), there exist then unique $y_{\lambda}^*,y_{\lambda_n}^* \in (0,\infty)$ such that $f_1(y_{\lambda}^*,\lambda)=0=f_1(y_{\lambda_n}^*,\lambda_n)$. Then, we have to show that $y_{\lambda_n}^* \xrightarrow{n \to \infty} y_{\lambda}^*$. Due to (i), (ii), and $\lambda_n \xrightarrow{n \to \infty} \lambda$, it holds that $y_{\lambda_n}^* \in [y_{\max_{k \in \N}\{\lambda,\lambda_k\}}^*,y_{\min_{k \in \N}\{\lambda,\lambda_k\}}^*]$ (maximum and minimum exist due to the convergence of $\lambda_n$ to $\lambda$). Then, there exist a subsequence $y_{\lambda_{n_l}}^*$ and a $\tilde{y}$ such that $y_{\lambda_{n_l}}^* \xrightarrow{l \to \infty} \tilde{y}$. Moreover, we know that $f_1(y_{\lambda}^*,\lambda) = 0 = f_1(y_{\lambda_{n_l}}^*,\lambda_{n_l}) \xrightarrow{l \to \infty} f_1(\tilde{y},\lambda)$ due to the joint continuity of $f_1$. Now, the uniqueness of the zero root implies that $\tilde{y}=y_{\lambda}^*$. Hence, all subsequences converge to $y_{\lambda}^*$ and thus also the sequence itself. Therefore, \eqref{eq: y star continuous} would follow provided we can show (i), (ii), and (iii) from the beginning:
		
		First, the joint continuity (i.e., property (iii)) follows directly from the definition of $f_1$. 
		
		Second, we derive property (i), i.e., the strict monotonicity in $\lambda$, from the formula of $\hat{X}_T$ (see Theorem \ref{optimal wealth}), which gives us the claim. Note that we add in the following paragraph a superscript to $\hat{X}_T$, $\xi_1^*$, $\xi_2^*$, and $\xi_3^*$ when we take these values for a certain fixed $\lambda$. \\
		Indeed, if, for all $\omega \in \Omega$, we can show that $\hat{X}_T^{\lambda_1} (\xi_T(\omega)) \geq \hat{X}_T^{\lambda_2} (\xi_T(\omega))$ with a strict inequality for a set with positive probability for all $\lambda_1 > \lambda_2 \, (>C)$, then also $f_1(\cdot,\lambda_1)>f_1(\cdot,\lambda_2)$ since $f_1 = \EX[\xi_T \hat{X}_T] -x_0$ and $\xi_T >0$. Therefore, let $\lambda_1 > \lambda_2 \, (>C)$: Due to \eqref{eq: xi3* bigger 0} and Proposition \ref{prop: * Eigenschaften}, we conclude that for all $\lambda>C$ (and hence for $\lambda_1$ and $\lambda_2$) $\hat{X}_T \not \equiv 0$, i.e., $\hat{X}_T (\xi_T) > 0$ for $\xi_T$ small enough. The formula of $\hat{X}_T$ (see Theorem \ref{optimal wealth}) implies that for fixed $y$ the slope remains unchanged in each interval $(0,\xi_1^*]$, $(\tilde{\alpha}\hat{\xi},\xi_2^*]$, resp. $(\alpha\hat{\xi},\xi_3^*]$ when changing $\lambda$, but the interval boundaries change. Note that the slope (as a function of $\xi_T$) is strictly negative in $(0,\xi_1^*]$ and $(\alpha\hat{\xi},\xi_3^*]$, and constant otherwise. Therefore, due to $\hat{X}_T$ being non-increasing and a non-increasing function getting larger when being shifted to the right, it is sufficient to show that all interval boundaries do not decrease when $\lambda$ increases and at least one boundary value strictly increases when $\lambda$ increases: \\
		It follows directly from the definition of $\hat{X}_T$ that $\lim_{\xi_T \to 0} \hat{X}_T^{\lambda_1} (\xi_T) > \lim_{\xi_T \to 0} \hat{X}_T^{\lambda_2} (\xi_T) \, (>0)$. Hence, due to having the same slope and the continuity of $\hat{X}_T^{\lambda_1}$ (resp. $\hat{X}_T^{\lambda_2}$) in $\xi_T$ on $(0,\xi_1^{*,\lambda_1}]$ (resp. $(0,\xi_1^{*,\lambda_2}]$), we conclude that $\xi_1^{*,\lambda_1} > \xi_1^{*,\lambda_2}$ since $\hat{X}_T^{\lambda_1} (\xi_T=\xi_1^{*,\lambda_1}) = k_2 = \hat{X}_T^{\lambda_2} (\xi_T=\xi_1^{*,\lambda_2})$, i.e., the interval is strictly increasing in $\lambda$. Next, we observe immediately from its definition that $\hat{\xi}^{\lambda_1} \geq \hat{\xi}^{\lambda_2}$. For $\xi_2^*$ and $\xi_3^*$, we show this property by proving that their derivatives with respect to $\lambda$ are non-negative. Let $\tilde{\xi}_2^*$ be defined as in Theorem \ref{optimal wealth} and $\tilde{\xi}_3^* := \bar{\xi} - \frac{2 \gamma \alpha^2}{y} \left( \sqrt{k_1^2 +k_1 \left(2k_0 + \frac{\lambda}{\gamma \alpha}\right)} - k_1 \right)$ with $\bar{\xi}$ as in Theorem \ref{optimal wealth}. Then, we get:
		\begin{align*}
			&&\frac{\partial}{\partial \lambda} \tilde{\xi}_2^* &= \frac{\alpha}{y} - \frac{\alpha k_1}{y k_2} = \frac{\alpha}{y} \left(1- \frac{k_1}{k_2}\right) \geq 0 \\
			&\Rightarrow& \frac{\partial^-}{\partial \lambda} \xi_2^* &\geq 0, \\
			&&\frac{\partial}{\partial \lambda} \tilde{\xi}_3^* &= \begin{cases}
				\dfrac{\alpha}{y} & \textit{ if $k_1=0$}, \\
				\dfrac{\alpha}{y} - \dfrac{2 \gamma \alpha^2}{y} \dfrac{\frac{1}{2} \cdot \frac{k_1}{\gamma \alpha}}{\sqrt{k_1^2+k_1(2k_0+\frac{\lambda}{\gamma \alpha})}} = \frac{\alpha}{y} \left(1- \dfrac{k_1}{\sqrt{k_1^2+k_1(2k_0+\frac{\lambda}{\gamma \alpha})}}\right) & \textit{ if $k_1>0$},
			\end{cases} \\
			&&&\geq 0, \\
			&\Rightarrow& \frac{\partial^-}{\partial \lambda} \xi_3^* &\geq 0,
		\end{align*}
		where $\frac{\partial^-}{\partial \lambda}$ denotes the left side derivative with respect to $\lambda$. Summarizing, $\xi_1^*$ strictly increases and $\xi_2^*$ and $\xi_3^*$ non-decrease in $\lambda$. Thus, property (i) follows.
		
		Third, we prove property (ii), i.e., that $f_1$ is strictly decreasing in $y$. Let $\hat{X}_T (y)$ be a function of $y$ as in Theorem \ref{optimal wealth}. It follows directly from the definitions of $\hat{\xi}$, $\xi_1^*$, $\xi_2^*$, and $\xi_3^*$ that they are non-increasing in $y$ for fixed $\lambda$ since they are non-negative by definition and $\tfrac{1}{y}$ can be factored out. Now, let $y_1 < y_2$ and $\xi_T>0$ arbitrary. Then, the claim follows if $\hat{X}_T (y_1,\xi_T) \geq \hat{X}_T (y_2,\xi_T)$ and $\hat{X}_{T} (y_{1},\xi _{T}) > \hat{X}_{T} (y_{2},\xi _{T})$ on a set with positive probability. If $\xi_T$ is in the same interval for both values $y_1$ and $y_2$, ``$\geq$'' is obvious due to \eqref{eq: optimal terminal wealth}. Moreover, we get "$ > $" if $\xi _{T} \in (0,\xi _{1}^{*}] \cup (\alpha\hat{\xi},\xi _{3}^{*}]$ (attained with positive probability). Since $\hat{\xi}$, $\xi_1^*$, $\xi_2^*$, and $\xi_3^*$ are non-increasing in $y$, $\hat{X}_T$ is non-increasing in $\xi_T$ (check Proposition \ref{prop: xhat continuous}), and for all $y \geq 0$ holds that $\hat{X}_{T} (y,a) > \hat{X}_{T} (y,b) > \hat{X}_{T} (y,c) > \hat{X}_{T} (y,d)$ for $a \in (0,\xi _{1}^{*})$, $b \in (\tilde{\alpha}\hat{\xi},\xi _{2}^{*}]$, $c \in (\alpha\hat{\xi},\xi _{3}^{*})$, and $d > \xi _{3}^{*}$, the claim also follows if $\xi_T$ is in different intervals for $y_1$ and $y_2$.
	\end{myproof}
	\begin{myproof}[Proof of \eqref{eq: y star lambda C}]
		For the proof of this equation, we have to consider the two cases $x_0 e^{\int_0^T r_s \diff s} < k_2$ and $x_0 e^{\int_0^T r_s \diff s} \geq k_2$. Note that if $x_0 e^{\int_0^T r_s \diff s} = k_2$, we get that $C = 2 \gamma \alpha (x_0 e^{\int_0^T r_s \diff s}-k_1-k_0) =
		2 \gamma \left(\tilde{\alpha} x_0 e^{\int_0^T r_s \diff s} + \alpha_2 k_2 - \alpha (k_0+k_1)\right)$. Moreover, notice that $x_0 e^{\int_0^T r_s \diff s} > k_1$ by assumption.
		
		\underline{Case 1:} $x_0 e^{\int_0^T r_s \diff s} < k_2$, i.e., $C = 2 \gamma \alpha (x_0 e^{\int_0^T r_s \diff s}-k_1-k_0)$: \\
		For $\lambda_{\varepsilon}=C+2\gamma \alpha e^{\int_0^T r_s \diff s}\varepsilon$, it holds that $\hat{\xi}=0$ for $\varepsilon$ small enough and hence also $\xi_1^* = \xi_2^* = 0$ due to \eqref{order of xis}. However, it holds with \eqref{eq: xi3* bigger 0} that $\xi_3^* >0$. Note that $\lim_{x \to 0} \Phi(d_j(x,0))=0$ for $j \in \{0,1,2\}$. Thus, $f_1$ reduces for all $\varepsilon>0$ small enough to
		\begin{align} \label{eq: f1 formula in the proof}
			f_1(y,C+2\gamma \alpha e^{\int_0^T r_s \diff s}\varepsilon) =&-x_0 + \left(k_0+k_1 + \dfrac{C+2\gamma \alpha e^{\int_0^T r_s \diff s}\varepsilon}{2 \gamma \alpha} \right) e^{-\int_0^T r_s \diff s} \Phi\left(d_1 \left(\xi_3^*,0\right)\right) \notag \\
			&- \dfrac{y}{2\gamma \alpha^2} e^{\int_0^T -2r_s+\norm{\kappa_s}^2 \diff s} \Phi \left( d_2 \left(\xi_3^*,0\right) \right) \\
			=&-x_0 + (x_0+\varepsilon) \Phi\left(d_1 \left(\xi_3^*,0\right)\right) - \dfrac{y}{2\gamma \alpha^2} e^{\int_0^T -2r_s+\norm{\kappa_s}^2 \diff s} \Phi \left( d_2 \left(\xi_3^*,0\right) \right). \notag
		\end{align}
		Now, if $\limsup_{\varepsilon \to 0} y_{\lambda_{\varepsilon}}^* >0$, it follows that $\liminf_{\varepsilon \to 0} f_1(y_{\lambda_{\varepsilon}}^*,\lambda_{\varepsilon}) < 0$ (since $\Phi\left(d_1 \left(\xi_3^*,0\right)\right) < 1$ for all $y>0$) which is a contradiction to $0 = \liminf_{\varepsilon \to 0} 0 = \liminf_{\varepsilon \to 0} f_{1}(y_{\lambda _{\varepsilon}}^{*},
		\lambda _{\varepsilon})$ by definition of $y_{\lambda _{\varepsilon}}^{*}$.
		Thus, the claim follows, i.e., $\lim_{\lambda \to C} y_{\lambda}^*=0$, when taking the limit of $\varepsilon \to 0$ (and hence $\lambda \to C$) on both sides, since $y_{\lambda}^* \geq 0$. Note that for every $y$, the existence of a solution $\lambda_y$ of \eqref{eq: f1 formula in the proof} was already ensured in the main part of the proof.
		
		\underline{Case 2:} $x_0 e^{\int_0^T r_s \diff s} \geq k_2$, i.e., $C = 2 \gamma \left(\tilde{\alpha} x_0 e^{\int_0^T r_s \diff s} + \alpha_2 k_2 - \alpha (k_0+k_1)\right)$: \\
		For $\lambda_{\varepsilon}=C+2\gamma \tilde{\alpha} e^{\int_0^T r_s \diff s}\varepsilon$, it holds that:
		\begin{align*}
			f_1(y,\lambda_{\varepsilon}) =& -x_0 + (x_0+\varepsilon) \Phi\left(d_1 \left(\xi_1^*,0\right)\right) - \dfrac{y}{2\gamma \tilde{\alpha}^2}  e^{\int_0^T -2r_s+\norm{\kappa_s}^2 \diff s} \Phi \left( d_2 \left(\xi_1^*,0\right) \right)  \\
			&+ k_2 e^{-\int_0^T r_s \diff s} \left( \Phi\left(d_1 \left(\xi_2^*,0\right)\right) - \Phi(d_1 (\tilde{\alpha} \hat{\xi},0)) \right) \\
			&+ \left( \dfrac{\tilde{\alpha}}{\alpha} (x_0+\varepsilon) + \dfrac{\alpha_2}{\alpha} k_2 e^{-\int_0^T r_s \diff s} \right) \left( \Phi\left(d_1 \left(\xi_3^{*},0\right)\right) - \Phi(d_1 (\alpha \hat{\xi},0)) \right) \\
			&- \dfrac{y}{2\gamma \alpha^2} e^{\int_0^T -2r_s+\norm{\kappa_s}^2 \diff s} \left(\Phi \left( d_2 \left(\xi_3^{*},0\right) \right) - \Phi ( d_2 (\alpha \hat{\xi},0)) \right) \\
			\leq& -x_0 + x_0 \left( \Phi\left(d_1 \left(\xi_1^{*},0\right)\right) + \Phi\left(d_1 \left(\xi_2^*,0\right)\right) - \Phi(d_1 (\tilde{\alpha} \hat{\xi},0)) + \Phi\left(d_1 \left(\xi_3^{*},0\right)\right) - \Phi(d_1 (\alpha \hat{\xi},0)) \right) \\
			&+ \varepsilon \left( \Phi\left(d_1 \left(\xi_1^{*},0\right)\right) + \Phi\left(d_1 \left(\xi_3^{*},0\right)\right) - \Phi(d_1 (\alpha \hat{\xi},0)) \right) \\
			&- \dfrac{y}{2\gamma}  e^{\int_0^T -2r_s+\norm{\kappa_s}^2 \diff s}  \left( \frac{1}{\alpha^2}\Phi \left( d_2 \left(\xi_3^{*},0\right) \right) - \frac{1}{\alpha^2} \Phi ( d_2 (\alpha \hat{\xi},0)) + \frac{1}{\tilde{\alpha}^2}\Phi \left( d_2 \left(\xi_1^{*},0\right) \right)\right) \\
			\leq& -x_0 + x_0 + \varepsilon - \dfrac{y}{2\gamma \tilde{\alpha}^2}  e^{\int_0^T -2r_s+\norm{\kappa_s}^2 \diff s} \left ( \Phi \left ( d_{2} \left (\xi _{3}^{*},0
			\right ) \right ) - \Phi ( d_{2} (\alpha \hat{\xi},0))	+ \Phi \left ( d_{2} \left (
			\xi _{1}^{*},0\right ) \right )\right ),
		\end{align*}
		where we used that $k_2 e^{-\int_0^T r_s \diff s} \leq x_0$, $\tilde{\alpha}+\alpha_2 = \alpha$, and $\frac{\tilde{\alpha}}{\alpha} \leq 1$ in the first inequality. Moreover, we used \eqref{order of xis} and $\Phi\left(d_1 \left(\xi_3^*,0\right)\right) \leq 1$ for all $y>0$ in the second inequality. Due to $\xi _{3}^{*}>0$ and \eqref{order of xis}, it holds that $\Phi \left ( d_{2} \left (\xi _{3}^{*},0
		\right ) \right ) - \Phi ( d_{2} (\alpha \hat{\xi},0)) + \Phi \left ( d_{2} \left (
		\xi _{1}^{*},0\right ) \right ) >0$.
		Hence, the claim follows, i.e., $\lim_{\varepsilon \to 0} y_{\lambda_{\varepsilon}}^*=0$ by the same argument as in Case 1. Note that the existence of a solution was already ensured in the main part of the proof.
	\end{myproof}
	\begin{myproof}[Proof of \eqref{eq: y star lambda C alt}]
		There is nothing to show since we already have that $y_{\lambda}^* \in [0,\infty)$.
	\end{myproof}
	\begin{myproof}[Proof of \eqref{eq: y star lambda infty}]
		We prove this by contradiction. Therefore, we assume that there exists an $L >0$ such that $\limsup_{\lambda \to \infty} y^*_{\lambda} \leq L$. Under this assumption, it holds that $\xi_1^* \xrightarrow{\lambda \to \infty} \infty$ since $y_{\lambda}^{*}$ is bounded and $\lambda - \sqrt{l(\lambda)} \xrightarrow{\lambda \to \infty} \infty$ for all affine functions $l$. Then \eqref{order of xis} implies that  $\hat{\xi},\bar{\xi},\xi_2^*,\xi_3^* \xrightarrow{\lambda \to \infty} \infty$ and we know that $\Phi(d_1(\xi_3^*,0)) \geq \Phi(d_1(\alpha \hat{\xi},0))$. However, then $\lim_{\lambda \to \infty} \inf_{0 \leq y \leq L} f_1 (y,\lambda) = \infty$ which is a contradiction to $y^*_{\lambda}$ being the zero root of $f_1$.
	\end{myproof}
	Since all statements are proved, we have shown the lemma.
\end{proof}

\section{Additional lemma} \label{lemmas}

The following lemma restates a well-known result for log-normal distributions, which is used, e.g., in deriving the pricing formula of a put or call in a Black-Scholes model. The proof is just a straightforward calculation. However, we will give it for the sake of completeness.

\begin{lemma}
	Let $0 \leq a \leq b \leq + \infty$ and $X \sim \mathcal{LN} (\mu,\sigma^2)$, where $\mathcal{LN}$ denotes a log-normal distribution. Then, it holds that:
	\begin{align} \label{conditional expectation log normal}
		\EX \left[X \1_{X \in [a,b]}\right] = e^{\mu+\frac{\sigma^2}{2}} \left( \Phi \left(\frac{\ln b - \mu - \sigma^2}{\sigma}\right) - \Phi \left(\frac{\ln a - \mu - \sigma^2}{\sigma}\right) \right),
	\end{align}
	where $\Phi$ denotes the cdf of a standard normal distribution with $\Phi(+\infty)=1$ and $\Phi(-\infty)=0$. The formula remains unchanged when we replace the interval $[a,b]$ by $(a,b]$, $[a,b)$, or $(a,b)$.
\end{lemma}

\begin{proof}
	It holds with $y := \ln x$ and $z := \frac{y-\mu-\sigma^2}{\sigma}$:
	\begin{align*}
		\EX \left[X \1_{X \in [a,b]}\right] &= \int_{\ln a}^{\ln b} \dfrac{1}{\sqrt{2 \pi \sigma^2}} \exp \left( y - \frac{(y-\mu)^2}{2 \sigma^2}\right) \diff y \\
		&= \exp \left( -\frac{1}{2 \sigma^2} \left(\mu^2 - (\mu+\sigma^2)^2\right)\right) \int_{\frac{\ln a-\mu-\sigma^2}{\sigma}}^{\frac{\ln b-\mu-\sigma^2}{\sigma}} \dfrac{1}{\sqrt{2\pi}} \exp \left(\frac{1}{2} z^2\right) \diff z \\
		&= \exp \left(\mu+\frac{\sigma^2}{2}\right) \left( \Phi \left(\frac{\ln b - \mu - \sigma^2}{\sigma}\right) - \Phi \left(\frac{\ln a - \mu - \sigma^2}{\sigma}\right) \right).
	\end{align*}
\end{proof}

\footnotesize
\bibliography{bibliography}
\footnotesize
\bibliographystyle{plain}

\footnotesize
\end{document}